\documentclass[11pt,letterpaper]{article}

\usepackage{amsmath,amssymb,amsfonts,amsthm,bbm}
\usepackage{mathtools}   
\usepackage{stmaryrd} 

\usepackage{enumitem}   

\usepackage{authblk} 

\usepackage[titletoc,title]{appendix}  

\usepackage[normalem]{ulem}

\usepackage[T1]{fontenc}

\usepackage[table,xcdraw,dvipsnames]{xcolor}   
\usepackage{hyperref}
\newcommand\myshade{85}
\colorlet{mylinkcolor}{YellowOrange}
\colorlet{mycitecolor}{Aquamarine}
\colorlet{myurlcolor}{violet}

\hypersetup{
  linkcolor  = mylinkcolor!\myshade!black,
  citecolor  = mycitecolor!\myshade!black,
  urlcolor   = myurlcolor!\myshade!black,
  colorlinks = true,
}

\usepackage{graphicx}
\usepackage{caption}
\usepackage{subcaption}
\usepackage{float}

\usepackage{booktabs}
\usepackage{makecell}
\usepackage{array}
\usepackage{multirow}

\usepackage[linesnumbered,ruled,vlined]{algorithm2e}

\SetCommentSty{mycommfont}

\SetKwInput{KwInput}{Input}                
\SetKwInput{KwOutput}{Output}              

\usepackage{listings}             
\usepackage{stmaryrd} 

\renewcommand{\hat}{\widehat}
\renewcommand{\tilde}{\widetilde}
\renewcommand{\bar}{\widebar}

\newcommand{\bfm}[1]{\ensuremath{\boldsymbol{#1}}} 

\def\bbone{\mathbbm{1}} 

\def\ba{\bfm a}   \def\bA{\bfm A}  
\def\bb{\bfm b}   \def\bB{\bfm B}  
     
   \def\bD{\bfm D}  
\def\be{\bfm e}   \def\bE{\bfm E}  \def\EE{\mathbb{E}}
\def\bff{\bfm f}    
\def\bg{\bfm g}     
     
   \def\bI{\bfm I}

   \def\bM{\bfm M}  
     
     \def\OO{\mathbb{O}}
     \def\PP{\mathbb{P}}
   \def\bQ{\bfm Q}  
     \def\RR{\mathbb{R}}

\def\bu{\bfm u}   \def\bU{\bfm U}  
\def\bv{\bfm v}   \def\bV{\bfm V}  
     
   \def\bX{\bfm X}  
     
\def\bz{\bfm z}   \def\bZ{\bfm Z}

\def\calB{{\cal  B}} \def\cB{{\cal  B}}
 \def\cC{{\cal  C}}
 \def\cD{{\cal  D}}
\def\calE{{\cal  E}} \def\cE{{\cal  E}}
\def\calF{{\cal  F}} \def\cF{{\cal  F}}
 
\def\calH{{\cal  H}} \def\cH{{\cal  H}}
\def\calI{{\cal  I}}

 \def\cL{{\cal  L}}
\def\calM{{\cal  M}} \def\cM{{\cal  M}}
 \def\cN{{\cal  N}}

 \def\cR{{\cal  R}}
 \def\cS{{\cal  S}}
 \def\cT{{\cal  T}}
 
 \def\cV{{\cal  V}}
 
\def\calX{{\cal  X}} \def\cX{{\cal  X}}
\def\calY{{\cal  Y}} \def\cY{{\cal  Y}}
\def\calZ{{\cal  Z}} \def\cZ{{\cal  Z}}

\newcommand{\bfsym}[1]{\ensuremath{\boldsymbol{#1}}}

 \def\bbeta{\bfsym \beta}
              
            \def\bDelta {\bfsym {\Delta}}
               
 \def\bmu{\bfsym {\mu}}                 
 
 \def\btheta{\bfsym {\theta}}

              \def\bSigma{\bfsym \Sigma}

\providecommand{\abs}[1]{\left\lvert#1\right\rvert}
\providecommand{\norm}[1]{\left\lVert#1\right\rVert}

\providecommand{\bbrackets}[1]{\llbracket #1 \rrbracket}

\providecommand{\defeq}{:=}

\usepackage{mathtools}
\DeclarePairedDelimiterX{\infdivx}[2]{(}{)}{%
  #1 \; \delimsize\| \; #2%
}

\DeclareMathOperator{\diag}{diag}
\newcommand{\E}[1]{{\mathbb{E}} \left[ #1 \right]}

\DeclareMathOperator{\Var}{Var}

\DeclareMathOperator{\vect}{vec}

\newtheorem{definition}{Definition}[section]

\newtheorem{lemma}[definition]{Lemma}
\newtheorem{proposition}[definition]{Proposition}
\newtheorem{theorem}[definition]{Theorem}

\theoremstyle{definition}
\newtheorem{remark}{Remark}

\definecolor{royalpurple}{rgb}{0.47, 0.32, 0.66}
\definecolor{greenfresh}{HTML}{00897B}
\definecolor{bluefresh}{HTML}{1E88E5}
\definecolor{redfresh}{HTML}{E53935}

\definecolor{royalpurple}{rgb}{0.47, 0.32, 0.66}

\def\beq{\begin{equation}}
\def\eeq{\end{equation}}

\def\bet{\begin{theorem}}
\def\eet{\end{theorem}}

\def\bel{\begin{lemma}}
\def\eel{\end{lemma}}

\def\tr{\mbox{tr}}

\def\eps{\varepsilon}
\def\lam {\lambda}

\def\cond{\;|\;}

\usepackage{setspace}
\usepackage{etoolbox}
\BeforeBeginEnvironment{equation*}{\begin{singlespace}\vspace*{-\baselineskip}}
	\AfterEndEnvironment{equation*}{\end{singlespace}\noindent\ignorespaces}
\BeforeBeginEnvironment{equation}{\begin{singlespace}\vspace*{-\baselineskip}}
	\AfterEndEnvironment{equation}{\end{singlespace}\noindent\ignorespaces}
\BeforeBeginEnvironment{align}{\begin{singlespace}\vspace*{-\baselineskip}}
	\AfterEndEnvironment{align}{\end{singlespace}\noindent\ignorespaces}
\BeforeBeginEnvironment{align*}{\begin{singlespace}\vspace*{-\baselineskip}}
	\AfterEndEnvironment{align*}{\end{singlespace}\noindent\ignorespaces}
\BeforeBeginEnvironment{eqnarray}{\begin{singlespace}\vspace*{-\baselineskip}}
	\AfterEndEnvironment{eqnarray}{\end{singlespace}\noindent\ignorespaces}
\BeforeBeginEnvironment{eqnarray*}{\begin{singlespace}\vspace*{-\baselineskip}}
	\AfterEndEnvironment{eqnarray*}{\end{singlespace}\noindent\ignorespaces}

\usepackage[framemethod=TikZ]{mdframed} 

\usepackage{fancybox}

\def\mttimes{\times_{m=1}^M}

\usepackage[top=1in, bottom=1in, left=1in, right=1in]{geometry}

\usepackage{setspace}
\usepackage[authoryear]{natbib}  

\newcommand{\mybib}{main}

\usepackage{graphicx}

\usepackage{comment}

\DeclareMathOperator*{\argmax}{\arg\max}
\DeclareMathOperator*{\mat1}{\text{mat}_1}
\DeclareMathOperator*{\vec1}{\text{vec}}
\DeclareMathOperator*{\matk}{{\text{mat}}_m}
\DeclareMathOperator*{\mats}{{\text{mat}}_S}

\def\ideal{(\text{\footnotesize ideal})}

\renewcommand{\hat}{\widehat}
\renewcommand{\tilde}{\widetilde}
\newcommand{\widebar}{\overline}

\def\spacingset#1{\renewcommand{\baselinestretch}%
{#1}\small\normalsize} \spacingset{1}

\def\TITLE{High-Dimensional Tensor Classification with \\
CP Low-Rank Discriminant Structure}

\newcommand{\blind}{0}

\usepackage{titlesec}

\titleformat{\section}
  {\normalfont\Large\bfseries}{\thesection}{1em}{}
\titlespacing*{\section}
  {0pt}{0.05\baselineskip}{0.05\baselineskip}

\begin{document}
\if0\blind
{
  \title{\bf \TITLE}
  \author{
    Elynn Chen$^\dag$ \hspace{8ex}
    Yuefeng Han$^\sharp$ \hspace{8ex}
    Jiayu Li$^\flat$ \hspace{8ex} \\ \normalsize
    \medskip
    $^{\dag,\flat}$New York University \hspace{8ex}
    $^{\sharp}$ University of Notre Dame
    }
   \date{September 18, 2024}
  \maketitle
} \fi

\if1\blind
{
  \bigskip
  \bigskip
  \bigskip
  \begin{center}
    {\LARGE\bf Title}
  \end{center}
  \medskip
} \fi

\bigskip
\begin{abstract}
\spacingset{1.38}
Tensor classification has become increasingly crucial in statistics and machine learning, with applications spanning neuroimaging, computer vision, and recommender systems. However, the high dimensionality of tensors presents significant challenges in both theory and practice. 
To address these challenges, we introduce a novel data-driven classification framework based on linear discriminant analysis (LDA) that exploits the CP low-rank structure in the discriminant tensor. Our approach includes an advanced iterative projection algorithm for tensor LDA and incorporates a novel initialization scheme called Randomized Composite PCA (\textsc{rc-PCA}). \textsc{rc-PCA}, potentially of independent interest beyond tensor classification, relaxes the incoherence and eigen-ratio assumptions of existing algorithms and provides a warm start close to the global optimum. We establish global convergence guarantees for the tensor estimation algorithm using \textsc{rc-PCA} and develop new perturbation analyses for noise with cross-correlation, extending beyond the traditional i.i.d. assumption. This theoretical advancement has potential applications across various fields dealing with correlated data and allows us to derive statistical upper bounds on tensor estimation errors. Additionally, we confirm the rate-optimality of our classifier by establishing minimax optimal misclassification rates across a wide class of parameter spaces. Extensive simulations and real-world applications validate our method's superior performance.
\end{abstract}

\noindent%
\vfill


\newpage
\spacingset{1.9} 
\addtolength{\textheight}{.1in}%


\section{Introduction}  \label{sec:intro}

Tensors, representing multidimensional data, have become prevalent in diverse scientific fields, including neuroimaging \citep{zhou2013tensor}, economics \citep{chen2021statistical, chen2020constrained}, computer vision, recommendation systems \citep{bi2018multilayer}, and climate analysis \citep{hoff2015multilinear,chen2024semi}. 
Tensor classification, a key task in these domains, categorizes multidimensional data into predefined classes, extending vector-based methods to handle complex, multi-way relationships. 
Recent applications in medical diagnosis \citep{hao2013linear} and protein identification \citep{wen2024tensorview} showcase its potential for enhanced accuracy and interpretability by preserving tensor data structure. 
Despite empirical successes, the field critically lacks rigorous statistical guarantees and theoretical understandings, primarily due to the inherent complexity and delicacy of analyzing tensor learning \citep{Chien2018,Liuregression2020,jahromi2024variational}.
At the same time, while researchers are rapidly advancing theoretical guarantees for tensor data analysis, their focus has been predominantly on tensor decomposition \citep{sun2017provable,zhang2018tensor,han2023guaranteed} and tensor regression \citep{zhou2013tensor,lee2022bayesian,XiaZhangZhou2022}, leaving a significant gap in the theoretical foundation of tensor classification.

This paper investigates tensor classification within the framework of high-dimensional tensor linear discriminant analysis (HD Tensor LDA). We consider a high-dimensional tensor predictor $\calX$ of dimension ${d_1 \times d_2 \cdots \times d_M}$ and a class label $Y \in [K]$ for $K \ge 2$, under the Tensor Gaussian Mixture Model (TGMM) for $(\calX, Y)$. 
Vectorizing tensor observations not only leads to a multiplicative increase in dimensionality (see Remark \ref{rmk:vector_lda} in Section \ref{subsection: Tensor LDA} for a detailed discussion), but also fails to capture the inherent low-rank structure in the tensor discriminant \citep{wen2024tensorview}.

To tackle the challenge of high-dimensionality and preserve the unique tensor strcuture, we introduce the concept of a CP low-rank structure in the discriminant tensor and design a novel high-dimensional (HD) CP-LDA algorithm with theoretical guarantees, exploring the previously unstudied CP low-rank structure in the discriminant tensor. This low-rank discriminant structure implies that the principal discriminant direction lies in a low-rank tensor subspace, marking a significant departure from the sparsity assumptions employed in both high-dimensional vector LDA \citep{cai2011direct,fan2012road,cai2019high} and tensor LDA \citep{pan2019covariate}. 

The HD CP-LDA algorithm directly estimate the discriminant tensor through iterative refinement of an initial estimate, which can be viewed as a perturbed version of the true discriminant tensor. 
Unlike existing tensor estimation literature that typically assumes i.i.d.~(sub-)Gaussian noise and imposes restrictive eigen-ratio conditions, our setting involves perturbation error with dependent entries and may not satisfy the conventional eign-ratio conditions. Consequently, current initialization methods and theoretical guarantees are inadequate for our context.

To address these challenges, we propose Randomized Composite PCA (\textsc{rc-PCA}), a novel initialization scheme for the CP bases of the discriminant tensor. \textsc{rc-PCA} relaxes key assumptions of existing initialization algorithms \citep{anandkumar2014guaranteed, sun2017provable, han2023cp}. It accommodates non-i.i.d.~noise entries, eases eigen-ratio conditions, and permits non-orthogonal CP bases. 

Despite these relaxations, we have developed new theoretical results demonstrating that \textsc{rc-PCA} provides satisfactory initialization, ensuring that the tensor iterative projection algorithm converges to the global optimum.
Furthermore, we establish the convergence rate of HD CP-LDA, accommodating correlated entries in the noise tensor through new perturbation analysis.

Our proposed algorithms and theoretical results extend recent advancements in tensor classification theory.
Notably, tensor logistic regression is examined as a special case of generalized tensor estimation in \cite{HanWillettZhang2022}. Our work relaxes their assumption of independent elements within the tensor observation.
In the context of Tensor Gaussian Mixture Models (TGMM), \cite{pan2019covariate} investigates tensor LDA with an overall sparsity assumption on the discriminant structure. \cite{li2022tucker} assumes a Tucker low-rank structure for the discriminant tensor, incorporating group sparsity within the loading matrices. 
Our work, however, focuses on the CP low-rank structure with potentially non-orthogonal CP bases, offering a novel perspective in tensor classification.
Despite existing work in empirical tensor classification \citep{guoxian2016,Liuregression2020} and theoretical tensor learning literature \citep{chen2024semi,XiaZhangZhou2022,LuoZhang2022}, our HD Tensor LDA setting and the proposed HD CP-LDA algorithm address a crucial gap in tensor classification with theoretical guarantees. 
The subsequent section will provide a comprehensive review of related literature, highlighting our distinctive contributions to the field.

Our contributions are twofold. Methodologically, we introduce the novel HD CP-LDA algorithm, which pioneers the use of CP low-rank structure for the discriminant tensor under minimal assumptions. We also develop a new initialization scheme, \textsc{rc-PCA}, for the CP bases of the discriminant tensor. \textsc{rc-PCA} operates under relaxed assumptions and provides a beneficial warm start close to the global optimum. 
The \textsc{rc-PCA} has broader applications in tensor CP decompositions and regressions, making it of independent interest. 
Extensive simulations and real data analyses demonstrate our algorithms' superior performance over existing tensor LDA methods. 
Theoretically, we establish sharp convergence rates for the initializations and the tensor iterative projection algorithm when applied to the discriminant tensor, accommodating correlated entries in the noise tensor. This advancement effectively addresses challenges posed by the complexity of tensor data, our CP low-rank structure, and spectrum perturbation analysis amid correlated noise entries. This result surpasses most current analyses in tensor decomposition, which typically assume i.i.d. entries in the noise tensor \citep{zhang2018tensor, wang2020learning}. Additionally, we establish the first minimax optimal classification rate in the tensor classification literature, thereby advancing the theoretical foundations of the field.

\subsection{Related work}

This research contributes to the expanding fields of high-dimensional LDA and tensor learning, particularly in tensor decompositions, regressions, and classification. We provide a comprehensive review of the most relevant literature in each area, highlighting the distinctions between previous works and our study.

\smallskip
\noindent
\textbf{Tensor Learning.}
Tensor learning has become crucial in machine learning and statistics for analyzing complex, multi-dimensional array data. Key areas include tensor decomposition \citep{sun2017provable,zhang2018tensor,han2023guaranteed,lee2021smooth}, tensor regression \citep{zhou2013tensor,li2017parsimonious,lee2022bayesian,XiaZhangZhou2022}, and unsupervised learning such as tensor clustering \citep{SunLi2019,mai2021doubly,HanLuoWangZhang2022,LuoZhang2022,hu2023multiway}.


Exploiting low-rank structures has been central to managing high-dimensional tensor data. The Tucker low-rank structure is widely used in tensor regression to manage ultra-high-dimensional parameters \citep{li2018tucker,Ahmed2020,XiaZhangZhou2022}, while \cite{HanWillettZhang2022} developed a unified framework for tensor-based statistical tasks. Conversely, the CP low-rank structure has been less explored. \cite{zhou2013tensor} applied CP low-rank structures in tensor regression and introduced a Block relaxation algorithm, which lacks a convergence guarantee even under substantial assumptions, highlighting the challenges of the CP structure. Our work bridges this gap by offering a fresh perspective on CP low-rank structures in tensor classification through the novel HD CP-LDA algorithm, equipped with the Randomized Composite PCA (\textsc{rc-PCA}) initialization scheme. This ensures convergence to optimal components under relaxed assumptions about eigen-ratio and non-orthogonality among CP bases \citep{anandkumar2014guaranteed, sun2017provable, han2023cp}. The versatility of \textsc{rc-PCA} extends beyond classification, enhancing various tensor learning applications where effective initialization is key, such as in regression and decomposition.

\smallskip
\noindent
\textbf{Tensor Classification.}
Empirical studies have shown the effectiveness of tensor-based classification methods across various domains, with developments including support tensor machines \citep{hao2013linear, guoxian2016}, tensor neural networks \citep{Chien2018, jahromi2024variational, wen2024tensorview}, and tensor logistic regression \citep{wimalawarne2016theoretical, Liuregression2020, SONG2023430}. However, these methods often neglect the dependence structure of tensor predictors and lack theoretical guarantees. 
Our research fills this gap by deriving minimax optimal rates for misclassification risk, explicitly considering the covariance structure of tensor data.
While our approach shares similarities with 
tensor multilinear discriminant analysis (MDA), it stands out through its specific objectives and theoretical depth. 
Tensor MDA \citep{lai2013sparse, Qunli2014, Franck2023}, though effective at feature extraction, typically requires subsequent vector classification, highlighting the need for methods that directly classify tensor data with strong theoretical support.

In tensor LDA, \cite{pan2019covariate} investigates group sparsity assumptions on the discriminant structure, \cite{li2022tucker} employs a Tucker low-rank structure for the discriminant tensor with group sparsity in the loading matrices, while \cite{wang2024parsimonious} explores discriminant analysis using Tucker low-rank tensor envelopes. Our work introduces a novel approach by focusing on the CP low-rank structure of the discriminant tensor, which offers greater flexibility and simplicity compared to Tucker structures and leads to theoretical advancements in tensor classification.

\smallskip
\noindent
\textbf{High-dimensional vector discriminant analysis.}
High-dimensional discriminant analysis has been widely studied in the vector domain, with much of the work rooted in Fisher's discriminant analysis \citep{cai2011direct, Clemmensen2011, witten2011penalized, fan2012road, mai2012direct, Xu2015, cai2019high}. Advanced methods have since emerged, such as \cite{Pan2016}'s feature screening for pairwise LDA in multiclass problems and \cite{cai2021convex}'s exploration of quadratic discriminant analysis. These approaches rely on the discriminant vector, often imposing sparsity through regularization techniques for efficient computation and interpretable feature selection in high dimensions. However, these methods are limited by their reliance on sparsity assumptions. In contrast, multi-dimensional discriminant tensors offer richer structures. Our CP low-rank approach captures complex interactions in tensor data that one-dimensional discriminant vectors cannot fully represent. By leveraging the intrinsic low-rank nature of tensor data, our method provides a more comprehensive analysis than traditional vector-based methods.


\subsection{Organization.} 
This paper is structured as follows: We begin with an introduction to basic tensor analysis notations and preliminaries. Section \ref{sec:model} introduces Tensor LDA with a discriminant tensor featuring CP low-rank structure. Section \ref{sec:method} details the HD CP-LDA algorithm for the discriminant tensor, including the \textsc{rc-PCA} initialization scheme and an iterative refinement algorithm. Section \ref{sec:theorems} develops the theoretical properties of our methods, establishing statistical bounds on estimation errors and deriving minimax optimal rates for misclassification risk. Section \ref{sec:simu} examines the finite sample performance of our CP low-rank approach through comprehensive simulations and comparative studies, demonstrating its advantages over traditional vector-based and sparse tensor methods. Section \ref{sec:appl} showcases the practical utility of our method through empirical studies on real-world tensor data. Section \ref{sec:summ} concludes the paper. All proofs and technical lemmas are in the supplementary appendix.

\section{The Model}
\label{sec:model}
\subsection{Notation and Preliminaries}
An $M$-th order tensor is a multi-dimensional array $\calX \in \RR^{d_1 \times \cdots \times d_M}$, with $\calI = (i_1, \dots, i_M)$ denoting the index of an element. The inner product $\langle \mathcal{X},\mathcal{Y} \rangle = \sum_{\calI} \calX_{\calI} \calY_{\calI}$ induces the Frobenius norm $\| \mathcal{X} \|_{\rm F} = \langle \mathcal{X},\mathcal{X} \rangle^{1/2}$. For $\mathnormal{A} \in \mathbb{R}^{\tilde{d_m} \times d_m}$, the mode-$m$ product $\mathcal{X} \times_m \mathnormal{A}$ yields an $M$-th order tensor of size $d_1 \times \cdots \times d_{m-1} \times \tilde{d_m} \times d_{m+1} \times \cdots \times d_M$. Elementwise, $(\mathcal{X} \times_m A)_{i_1, \ldots, i_{m-1}, j, i_{m+1}, \ldots, i_M} = \sum_{i_m=1}^{d_m} \mathcal{X}_{i_1, \ldots, i_M} \cdot A_{j, i_m}$. The mode-$m$ matricization $\text{mat}_m(\mathcal{X})$ is a $d_m \times \prod_{k\neq m} d_k$ matrix, where $\mathcal{X}_{i_1\dots i_M}$ maps to element $(i_m, j)$, with $j=1+\sum\nolimits_{k \neq m} (i_k -1) \prod_{l<k,l \neq m}d_l$. For $S \subseteq \{1, \dots, M\}$, the multi-mode matricization ${\rm mat}_S(\calX)$ is a $\prod_{m \in S} d_m$-by-$\prod_{m \notin S} d_m$ matrix with elements mapped as:
$i = 1 + \sum_{m \in S} (i_m - 1) \prod_{\ell \in S, \ell < m} d_\ell$, 
$j = 1 + \sum_{m \notin S} (i_m - 1) \prod_{\ell \notin S, \ell < m} d_\ell$. Vectorization is denoted as $\text{vec}(\mathcal{X}) \in \mathbb{R}^d$, where $d=\prod_{m=1}^M d_m$. For a comprehensive review of tensor algebra and operations, we refer to \cite{kolda2009tensor}. 
For a matrix $\bA = (a_{ij})\in \RR^{m\times n}$, the matrix Frobenius norm is denoted as $\|\bA\|_{\rm F} = (\sum_{ij} a_{ij}^2)^{1/2}$ and its spectral norm is denoted as $\|\bA\|_{2}$. For two sequences of real numbers $\{a_n\}$ and $\{b_n\}$, write $a_n=O(b_n)$ (resp. $a_n\asymp b_n$) if there exists a constant $C$ such that $|a_n|\leq C |b_n|$ (resp. $1/C \leq a_n/b_n\leq C$) for all sufficiently large $n$, and write $a_n=o(b_n)$ if $\lim_{n\to\infty} a_n/b_n =0$. Write $a_n\lesssim b_n$ (resp. $a_n\gtrsim b_n$) if there exist a constant $C$ such that $a_n\le Cb_n$ (resp. $a_n\ge Cb_n$). We use $C, C_1,c,c_1,...$ to denote generic constants, whose actual values may vary from line to line. 

The Tensor Normal (TN) distribution generalizes the matrix normal distribution to higher-order tensors \citep{hoff2011TN}. Let \(\calZ \in \RR^{d_1 \times \cdots \times d_M}\) be a random tensor with i.i.d. \(N(0,1)\) elements. Given a mean tensor \(\mathcal{M} \in \mathbb{R}^{d_1 \times \cdots \times d_M}\) and mode-specific covariance matrices \(\Sigma_m \in \mathbb{R}^{d_m \times d_m}\) for \(m \in [M]\), we define \(\calX=\calM + \calZ \times_{m=1}^M \Sigma_m^{1/2}\). 

Then, \(\calX \sim \cT\cN(\cM; \bSigma)\), where \(\bSigma := [\Sigma_m]_{m=1}^M\). The probability density function is:

\begin{equation} 
    f(\calX | \cM, \bSigma)= (2 \pi)^{-d/2}
    \left(\prod\nolimits_{m=1}^M |{\Sigma_m}|^{-\frac{d}{2 d_m}}\right)
    \exp\left(- \norm{ (\calX - \cM) \mttimes \Sigma_m^{-1/2}}_{\rm F}^2/2\right). \label{eqn:tenor-normal-density}
\end{equation}
Importantly, \(\calX \sim \cT\cN(\cM; \bSigma)\) is equivalent to \(\vect(\calX) \sim \cN(\vect(\cM); \Sigma_M\otimes\cdots\otimes\Sigma_1)\), where \(\otimes\) denotes the Kronecker product.

\subsection{Tensor High-Dimensional Linear Discriminant Analysis}\label{subsection: Tensor LDA}
We explore the classification of tensor variables using the Tensor Gaussian Mixture Model (TGMM), which extends the Gaussian Mixture Model to accommodate tensor-variate $\calX$ drawn from $K$ classes $\{\cC_k\}_{k=1}^K$. The TGMM distribution for $(\calX, Y)$, where $Y$ is the class label, is specified by:
\begin{equation} \label{eqn:tgmm2}
   \calX \cond Y = k \sim \cT\cN(\cM_k; \bSigma_k),
   \quad \pi_k \defeq \PP(Y = k), 
   \quad \sum\nolimits_{k=1}^K \pi_k = 1, 
   \quad \text{for any } k \in [K],
\end{equation}
where $\calM_k \in \RR^{d_1 \times \cdots \times d_M}$ is the mean tensor of the $k$th class, $\bSigma_k$ represents the within-class covariance matrices, and $0<\pi_k<1$ is the prior probability for class $k$. In this paper, we focus on the binary classification case ($K = 2$) with equal covariance structure $\bSigma_k := \bSigma= [\Sigma_m]_{m=1}^M$ for all $k\in[K]$. Given independent observations $\calX_1^{(1)}, \ldots, \calX_{n_1}^{(1)} \overset{\text{i.i.d.}}{\sim} \cT\cN(\cM_1; \bSigma)$ and $\calX_1^{(2)}, \ldots, \calX_{n_2}^{(2)} \overset{\text{i.i.d.}}{\sim} \cT\cN(\cM_2; \bSigma)$, our goal is to classify a future data point $\calZ$ into one of the two classes.

In the ideal setting where all parameters $\btheta := (\pi_1, \pi_2, \cM_1, \cM_2, \bSigma)$ are known, Fisher's linear discriminant rule for Tensor LDA is given by:
\begin{equation}
\label{eqn:lda-rule}
\Upsilon(\cZ) 
= \bbone\left\{ \langle \cZ - \cM, \; \calB \rangle +\log\Big(\pi_2/\pi_1\Big) \ge 0 \right\},
\end{equation}
where $\cM = (\cM_1 + \cM_2)/2$, $\cD=(\cM_2-\cM_1)$, and the discriminant tensor $\calB = \cD \times_{m=1}^M \Sigma_m^{-1}$. The indicator function $\bbone(\cdot)$ assigns $\cZ$ to class $\cC_1$ (value 0) or $\cC_2$ (value 1). This rule, generalizing the classical vector case \citep{anderson2003}, is optimal under the TGMM \citep{mai2021doubly}. The optimal misclassification error is:
$R_{\text{opt}}=\pi_1\phi(\Delta^{-1}\log(\pi_2/\pi_1)-\Delta/2)+\pi_2(1-\phi(\Delta^{-1}\log(\pi_2/\pi_1)+\Delta/2)),$
where $\phi$ is the standard normal CDF and $\Delta=\sqrt{\langle \calB, \; \cD \rangle}$ is the signal-to-noise ratio. Building on the framework in vector LDA \citep{cai2011direct, mai2012direct, cai2019high}, this paper focuses on the estimation of the discriminant tensor $\cB$ while enforcing a CP low-rank structure.

\begin{remark}
The classification rule can also be derived from Bayes' Theorem. Let $f_k (\calZ) \equiv \PP(\mathcal{X} \cond Y=k)$ be the conditional probability density function \eqref{eqn:tenor-normal-density} for the $k$th class. Bayes' Theorem states that $\PP(Y=k \cond \mathcal{X}) = \pi_k f_k(\calZ) / (\sum_{k=1}^K \pi_k f_k(\calZ))$. The Bayes' rule classifies an observation to the class maximizing $\PP(Y=k \cond \mathcal{Z})$, achieving the lowest possible error rate \citep{domingos1997optimality}: $\hat{Y} = \mathop{\argmax}_{k \in [K]} \PP(Y=k \cond \mathcal{Z}) = \mathop{\argmax}_{k \in  [K]} \pi_k f_k(\calZ)$. For binary classification ($K = 2$), this simplifies to:
\(\Upsilon(\mathcal{Z}) = \bbone \left\{ \pi_2 f_2(\mathcal{X}) - \pi_1 f_1(\mathcal{X}) \geq 0 \right\}\). Rearranging and taking logarithms yields:
\begin{equation*}
\Upsilon(\mathcal{Z}) = \bbone \left\{ \log(\pi_2 / \pi_1) + \frac{1}{2} \left\| (\mathcal{X} - \mathcal{M}_1) \times_{m=1}^M \Sigma_m^{-1/2} \right\|_{\rm F}^2 - \frac{1}{2} \left\| (\mathcal{X} - \mathcal{M}_2) \times_{m=1}^M \Sigma_m^{-1/2} \right\|_{\rm F}^2 \geq 0 \right\}.
\end{equation*}
The equivalence of $\Upsilon(\mathcal{Z})$ to \eqref{eqn:lda-rule} can be demonstrated through tensor algebra.
\end{remark}

\begin{remark} \label{rmk:vector_lda}
Vectorizing tensor-variate $\calX$ to apply high-dimensional vector LDA leads to the discriminant rule:
\begin{equation*}
\Upsilon(\bz) 
= \bbone\left\{ (\bmu_2-\bmu_1)^\top (\Sigma_M\otimes\cdots\otimes\Sigma_1)^{-1} (\bz - \frac{\bmu_1+\bmu_2}{2}) +\log\Big(\pi_2/\pi_1\Big) \ge 0 \right\}, 
\end{equation*}
where $\bz=\vect(\calZ)$ and $\bmu_k = \vect(\calM_k)$. Studies in high-dimensional LDA \citep{cai2011direct, mai2012direct, cai2019high} estimate the discriminant direction $\bbeta=(\bmu_2-\bmu_1)^\top (\Sigma_M\otimes\cdots\otimes\Sigma_1)^{-1}$ directly under sparsity assumptions. However, this approach may overlook tensor structures and face ultrahigh-dimensionality issues. For instance, \cite{cai2018rate}'s adaptive estimation of $\bbeta$ requires solving linear programming (LP) problems. For a tensor $\calX \in \RR^{30 \times 30 \times 30}$, the resulting optimization involves 27,000 inequalities, becoming computationally extremely intensive.
\end{remark}

\begin{remark}
While \cite{pan2019covariate} utilize group sparsity of $\calB$ with a group LASSO penalty, we recognize and utilize that the \emph{multi-dimensional} discriminant tensor $\calB$ often exhibits richer low-rank structures, which captures complex dependencies among tensor dimensions beyond simple sparsity. Our proposed decomposable discriminant tensor $\calB$ will be elaborated upon later. One could employ the method from \cite{HanLuoWangZhang2022} by framing the problem as an optimization task using the negative log-likelihood function $\cL$ of logistic regression,
\begin{equation*}
\hat{\calB} = \arg\min_{\calB \; \text{is Tucker low-rank}} \mathcal{L}(\calB; \cX_1^{(k)},...,\cX_{n_k}^{(k)},k=1,2).
\end{equation*} 
However, measuring the statistical noise of the model---a crucial plug-in quantity in their theoretical framework for estimation errors---is non-trivial. This challenge hinders the validation of theoretical guarantees under their method. Additionally, their work typically assumes that $\cX_i^{(k)}$ has i.i.d. elements, a requirement that is often too stringent in practice. Consequently, our work necessitates the introduction of a new estimation procedure and the development of novel perturbation analysis to establish theoretical guarantees. Moreover, our approach diverges by imposing a CP low-rank structure on $\calB$, relaxing the orthogonality assumption, and providing a novel perspective in the field of tensor classification.
\end{remark}

\section{Estimation of CP-Low-Rank Discriminant Tensor} \label{sec:method}
In this paper, we propose an efficient method for estimating the discriminant tensor $\cB$, utilizing the CANDECOMP/PARAFAC (CP) low-rank structure:
\begin{equation}\label{eqn:lda-cp}
    \calB =\sum\nolimits_{r=1}^R w_r \cdot (\ba_{r1}\circ\ba_{r2}\circ\cdots\circ \ba_{rM}), 
\end{equation}
where $R$ is the CP low rank, $\circ$ denotes the tensor outer product, $w_r>0$ represents signal strength, and $\ba_{rm}$ is a unit vector of dimension ${d_m}$. Our estimation strategy encompasses two key components: firstly, developing an initial estimator that closely approximates the true discriminant tensor, ensuring a high probability of achieving an optimal solution; and secondly, algorithmically devising a computationally efficient method to refine this estimator, with provable convergence guarantees.

\begin{remark}
An alternative decomposition of $\calB$ is the Tucker low-rank structure:
\begin{equation*}
\calB =\cF\times_1 \bU_1\times_2\cdots\times_M \bU_M,
\end{equation*}
where $\calF \in \RR^{r_1 \times r_2 \times \cdots \times r_M}$ is the core tensor and $\bU_i \in \RR^{d_i \times r_i}$ are loading matrices. However, the CP low-rank structure of $\calB$ offers several advantages \citep{han2023cp,erichson2020randomized}. Unlike Tucker decomposition, which suffers from ambiguity due to invertible transformations, CP structure 
is uniquely defined up to sign changes, facilitating interpretation. Furthermore, while CP can be viewed as a special case of Tucker with a superdiagonal core tensor, our work allows for non-orthogonal CP bases $\{ \ba_{rm}, r\in [R] \}$. This generalization, though more challenging, offers greater flexibility compared to the common orthonormal representation in Tucker decomposition. Additionally, CP's sequence of signal strengths is more parsimonious than Tucker's redundant core tensor, enhancing modeling, interpretation, and practical estimation. These characteristics render the CP low-rank structure superior for our discriminant tensor analysis.
\end{remark}
In the subsequent sections, we introduce the sample discriminant tensor and explore the advantages of applying low-rank structures. We propose two algorithms: 
(i) an iterative projection algorithm utilizing power iteration techniques to ensure linear local convergence, and 
(ii) a method to generate a beneficial warm start from the sample discriminant tensor, which facilitates global convergence with high probability when used as the initial CP bases in the refinement process.

\subsection{Sample Discriminant Tensor}
For each class $\cC_k$, let $n_{k}$ be the sample size, $\calX_i^{(k)}$ the $i$-th sample, and $\bar\calX^{(k)} = (n_{k})^{-1}\sum_{i=1}^{n_{k}} \calX_i^{(k)}$ the sample mean. We define the sample discriminant tensor as:
\begin{equation}\label{eqn:lda-discrim-tensor}
\widehat\calB = (\widebar\calX^{(2)}-\widebar\calX^{(1)}) \times_{m=1}^{M} \widehat\Sigma_m^{-1},
\end{equation}
where $d=\prod_{\ell} d_{\ell}, d_{-m}=d/d_{-m}$,
\begin{align*}
& \widehat\Sigma_m = (n_1+n_2)^{-1} d_{-m}^{-1}\sum\nolimits_{k=1}^2\sum\nolimits_{i=1}^{n_k}\matk(\calX_{i}^{(k)}-\widebar\calX^{(k)}){\matk}^{\top}(\calX_{i}^{(k)}-\widebar\calX^{(k)}), \quad\text{for } m\in[M],\\
&\hat C_{\sigma} = \prod_{m=1}^M \hat\Sigma_{m,11}/\hat{\Var}(\cX_{1,1\cdots1}), \quad \text{and update} \quad\widehat\Sigma_M = \hat C_{\sigma}^{-1} \cdot \widehat\Sigma_M.
\end{align*}
We estimate $\pi_k$ by $\widehat\pi_k=\frac{n_k}{n_1+n_2}$ for $k=1,2$. Denote $n=\min(n_1, n_2)$. Inspired by \cite{drton2021existence}, $\hat\Sigma_m$ is positive definite and \eqref{eqn:lda-discrim-tensor} is well-defined with probability 1 if $n_{}> d_m/d_{-m}$. This condition is mild, especially when dimensions across modes are comparable. For violations, we can use a perturbed $\hat\Sigma_m^{'} =\hat\Sigma_m +\gamma I_{d_m}$ for a relatively small $\gamma$ \citep{ledoit2004well}, or a sparse precision matrix estimator \citep{friedman2008sparse,cai2016estimating,lyu2019tensor} for estimating covariance matrices.

The sample discriminant tensor $\hat\calB$, viewed as a perturbed version of the true $\calB$, can be refined iteratively under CP low-rankness \eqref{eqn:lda-cp}. However, the low-rank structure's benefits extend beyond improved estimation accuracy. Such structural assumptions are fundamental for constructing consistent high-dimensional LDA rules. According to \cite{cai2021convex}, no data-driven method can mimic the optimal misclassification error $R_{\text{opt}}$ in high-dimensional settings ($d\gtrsim n$) with unknown means, even with known identity covariance matrices. In contrast, our high-dimensional Tensor LDA rule, utilizing an estimated low-rank discriminant tensor, achieves minimax optimal misclassification rate, as demonstrated in Theorems \ref{thm:class-upp-bound} and \ref{thm:class-lower-bound}.

\subsection{Refinement of the Discriminant Tensor}

The core of the refinement procedure is iterative tensor projection (Algorithm \ref{alg:tensorlda-cp}). Each iteration performs simultaneous orthogonalized projections across tensor modes to refine the CP components $\ba_{rm},\; r \in [R]$, followed by estimating signal strengths $w_r, \; r \in [R]$. This approach, based on the CP structure, preserves the signal strength of $\calB$ while reducing the perturbation $\hat\calB - \calB$, thereby improving the signal-to-noise ratio (SNR) and estimation accuracy..

To illustrate the functionality, consider estimating mode-$m$ CP bases $\{\ba_{r m}\}_{r\in[R]}$ given true mode-$\ell$ CP bases $\{\ba_{r\ell}\}_{r\in[R]}$ for all $\ell\neq m$. Let $\bA_{\ell}=[\ba_{1\ell},\ldots,\ba_{R\ell}]$ and $\bB_{\ell} = \bA_{\ell}(\bA_{\ell}^{\top} \bA_{\ell})^{-1} = [\bb_{1\ell},...,\bb_{R\ell}]$ as a right inverse of $\bA_{\ell}$, where $\ba_{k\ell}^\top\bb_{r\ell }=\bbone\{k=r\}$. Projecting the noisy $\hat\calB$ with $\bb_{r \ell}$ on all modes $\ell \ne m$ extracts the signal component containing only $\ba_{rm}$:
\begin{equation} \label{eq:cp-ideal}
\bz_{rm} 
= \hat\calB \times_{\ell=1, \ell\ne m}^{M} \bb_{r\ell}^\top
= \underbrace{w_{r} \ba_{r m}}_\text{signal}
+ \underbrace{ (\hat\calB - \calB) \times_{\ell=1, \ell\ne m}^{M} \bb_{r\ell}^\top}_\text{noise}.
\end{equation}
The signal term retains the strength $w_r$ along $\ba_{rm}$, while the noise term is filtered through projections. When the true $\{\bB_{\ell},\; \ell \neq m\}$ are known, the relative magnitude of the noise term is small compared to the signal term. This characteristic enables the distinction of the signal component with theoretical guarantees, utilizing perturbation theory \citep{wedin1972perturbation}. 
In practice, with unknown $\{\bB_{\ell},\; \ell \neq m\}$, we update CP bases $\{\ba_{rm},\; r \in [R], m \in [M]\}$ at the $t$-th iteration:
\begin{align*}
\bz_{rm}^{(t)} &= \hat\calB \times_1 \hat\bb_{r1}^{(t)\top} \times_2 \cdots \times_{m-1} \hat\bb_{r,m-1}^{(t)\top} \times_{m+1} \hat\bb_{r, m+1}^{(t-1)\top} \times_{m+2} \cdots \times_M \hat\bb_{rM}^{(t-1)\top},\\ 
\hat \ba_{rm}^{(t)} &= \bz_{rm}^{(t)}/\| \bz_{rm}^{(t)} \|_2.
\end{align*}
The projection error is:
\begin{alignat*}{2}
\bz_{rm} - \bz_{rm}^{(t)} &= w_r\ba_{rm} - \sum_{i=1}^R \widetilde w_{i,r} \ba_{im} + \be_{rm} - \tilde \be_{rm}, \\
\text{where} \quad \widetilde w_{i,r} &= w_i \prod_{\ell=1}^{m-1} \left[ \ba_{i,\ell}^{\top} \hat \bb_{r\ell}^{(t)} \right] \prod_{\ell=m+1}^{M} \left[ \ba_{i,\ell}^{\top} \hat \bb_{r\ell}^{(t-1)} \right], \\
\be_{rm} = (\hat\calB - \calB)\times_{\ell=1}^{M} \bb_{r,\ell}^\top &\quad \text{and} \quad \tilde \be_{rm} = (\hat\calB - \calB)\times_{\ell=1}^{m-1} \widehat \bb_{r,\ell}^{(t)\top} \times_{\ell=m+1}^{M} \widehat \bb_{r,\ell}^{(t-1)\top}.
\end{alignat*}

The multiplicative measure in $\widetilde w_{i,r}$ decays rapidly as $\ba_{i,\ell}^{\top} \hat \bb_{r\ell}^{(t)}$ approaches zero for $i \neq r$ with increasing iterations, and $\widetilde w_{i,r}$ has a product of $M-1$ such terms. The remaining terms are projected noise, controlled by the signal-to-noise ratio. Specifically, we can establish upper bounds on the estimation error of $\ba_{rm}$, ensuring the accuracy of our recovered CP bases.

\begin{algorithm}[htpb!]
    \SetKwInOut{Input}{Input}
    \SetKwInOut{Output}{Output}
    \Input{Initial tensor $\hat\calB$, 
    CP rank $R$, warm-start $\hat \ba_{rm}^{(0)}, 1\le r\le R, 1\le m\le M$, tolerance parameter $\epsilon>0$, maximum number of iterations $T$}
    
    Let $t=0$
    
    \Repeat{$t = T$ {\bf or} $\max_{r,m}\| \hat \ba_{rm}^{(t)}  \hat \ba_{rm}^{(t)\top} - \hat \ba_{rm}^{(t-1)} \hat \ba_{rm}^{(t-1)\top}\|_{2}\le \epsilon$}{
        Set $t=t+1$. 
        
        \For{$m = 1$ to $M$}{
            \For{$r = 1$ to $R$}{
                Compute $\bz_{rm}^{(t)} = \hat\calB \times_1 \hat\bb_{r1}^{(t)\top} \times_2 \cdots \times_{m-1} \hat\bb_{r,m-1}^{(t)\top} \times_{m+1} \hat\bb_{r, m+1}^{(t-1)\top} \times_{m+2} \cdots \times_M \hat\bb_{rM}^{(t-1)\top}$
                
                
                Compute $\hat \ba_{rm}^{(t)} = \bz_{rm}^{(t)}/\| \bz_{rm}^{(t)} \|_2$
            }
            Compute $(\hat \bb_{1m}^{(t)},\ldots, \hat \bb_{rm}^{(t)})$ as the right inverse of $(\hat \ba_{1m}^{(t)}, \ldots, \hat \ba_{rm}^{(t)})^\top$
            
            Set $(\hat \bb_{1m}^{(t+1)},..., \hat \bb_{rm}^{(t+1)})=(\hat \bb_{1m}^{(t)},..., \hat \bb_{rm}^{(t)})$ 
        }
        Compute $\hat w_{r}^{(t)} = \big|\hat\calB\times_{m=1}^M (\hat\bb_{rm}^{(t)})^\top\big|, 1\le r\le R$
        }

    \Output{$\hat  \ba_{rm}=\hat \ba_{rm}^{(t)}$, 
        $\hat w_{r} = \big|\hat\calB\times_{m=1}^M (\hat  \bb_{rm}^{(t)})^\top\big|$, 
        $\hat\calB^{\rm cp}=\sum_{r=1}^R \hat w_r \circ_{m=1}^M \hat\ba_{rm}$, $1\le r\le R,  m\le M$}
    \caption{Discriminant Tensor Iterative Projection for CP low rank (DISTIP-CP) }
    \label{alg:tensorlda-cp}
\end{algorithm}
\vspace{-1em}

\subsection{Initialization of CP Bases}
The refinement procedure in Algorithm \ref{alg:tensorlda-cp} requires initialization of CP bases $\{\hat\ba_{rm}^{(0)}, r \in [R], m \in [M]\}$. Given the non-convex nature of this problem and its propensity for multiple local optima, a carefully designed initialization is crucial for convergence to the optimal solution. We propose a novel Randomized Composite PCA (\textsc{rc-PCA}) initialization method, detailed in Algorithm \ref{alg:initialize-cp}. This method efficiently produces initial estimates within the convergence zone of the iterative procedure, typically requiring only a polynomial number of attempts to achieve global convergence, as confirmed in Theorem \ref{thm:cp-initilization}.

Our initialization strategy addresses two scenarios: one with a significant gap between signal strengths $\{w_r, r \in [R]\}$, and another where all signal strengths are of similar magnitude. Consider the multi-modes matricization:
\[ {\rm mat}_S(\calB) = \sum\limits_{r=1}^R w_r \vect(\circ_{m\in S} \ba_{rm}) \vect(\circ_{m\in S^c} \ba_{rm})^\top, \]
where $S, S^c$ are complementary subsets of $[M]$. For orthogonal tensor components, signal strengths directly correspond to singular values. 
In the non-orthogonal case, Proposition \ref{prop:singular_value_gap} shows that under `mild' non-orthogonality, the difference between signal strengths and singular values remains relatively small. Thus, we use the gap between singular values as a proxy for signal strength differences in Algorithm \ref{alg:initialize-cp}.

\begin{proposition}
\label{prop:singular_value_gap}
Let $\bA_S= (\ba_{1,S},\ldots, \ba_{R,S})\in \RR^{d_S\times R}$ and $\bA_{S^c}= (\ba_{1,S^c},\ldots, \ba_{R,S^c})\in \RR^{d_{S^c}\times R}$, where $\ba_{r,S} = \vect(\circ_{m\in S} \ba_{rm}), \; \ba_{r,S^c} = \vect(\circ_{m\in S^c} \ba_{rm})$. Then
$ \mathrm{mat}_S(\calB) = \bA_S W \bA_{S^c}^\top,$
where \(W = \mathrm{diag}(w_1, w_2, \ldots, w_R)\). The measure of non-orthogonality is defined as
\[
\delta = \| \bA_S^\top \bA_S - I_R \|_2 \vee \| \bA_{S^c}^\top \bA_{S^c} - I_R \|_2.
\]
\noindent The bound on the difference between the singular values \(\lambda_r\) obtained from the SVD of \(\mathrm{mat}_S(\calB)\) and the signal strengths \(w_r\) is: $|\lambda_r - w_r| \leq \sqrt{2}\delta w_1.$
\end{proposition}

Algorithm \ref{alg:initialize-cp} first employs the Composite PCA (CPCA) algorithm from \cite{han2023cp} to address the large gap scenario. For cases where CPCA fails due to insufficient signal strength gaps, Algorithm \ref{alg:initialize-cp} then call Procedure \ref{alg:initialize-random}. This approach involves two key steps. Initially, it uses random projection to create a sufficient gap between the first and second largest singular values, satisfying the gap condition in CPCA. This allows CPCA to extract the CP basis corresponding to the leading singular value of the projected tensor. Subsequently, it applies clustering to extract CP bases. As demonstrated in Proof III of Appendix \ref{append:proof:initia}, with high probability, applying random projection $L$ times effectively cover all good initializations of the CP bases with guaranteed upper bounds. 
\begin{remark}
Finally, given the estimated discriminant tensor $\hat\calB^{\rm cp}$, the high-dimensional tensor CP-LDA rule is given by
\vspace{1ex}
\begin{equation}
\hat\Upsilon_{\rm cp}(\cZ) = \bbone\left\{ \langle \cZ - (\widebar\calX^{(1)}+\widebar\calX^{(2)})/2, \; \hat\calB^{\rm cp} \rangle + \log(\hat\pi_2/\hat\pi_1) \ge 0 \right\}.   \label{eqn:lda-rule-cp}
\end{equation}
It assigns $\cZ$ to class $\cC_1$ with $\hat\Upsilon_{\rm cp}(\cZ)=0$ or $\cC_2$ with $\hat\Upsilon_{\rm cp}(\cZ)=1$.
\end{remark}

\begin{algorithm}[h!]
    \SetKwInOut{Input}{Input}
    \SetKwInOut{Output}{Output}
    \Input{Initial tensor $\hat\calB$, CP rank $R$, $S \subset [M]$, small constant $0<c_0<1$.}

    If $S=\emptyset$, pick $S$ to maximize $\min(d_S,d/d_S)$ with $d_S = \prod_{m\in S}d_m$ and $d = \prod_{m=1}^M d_m$.

    Unfold $\hat\calB$ to be a $d_S\times (d/d_S)$ matrix ${\rm mat}_S(\hat\calB)$.

    Compute  $\widehat\lambda_r,\widehat \bu_r, \widehat \bv_r$, $1\le r\le R$ as the top $R$ components in the SVD ${\rm mat}_S(\hat\calB) = \sum_{r}\widehat\lambda_r \widehat \bu_r \widehat \bv_r^\top$. Set $\widehat \lambda_0=\infty$ and $\widehat\lambda_{R+1}=0$.

\If{$\min\{|\widehat\lambda_r-\widehat\lambda_{r-1}|,|\widehat\lambda_r-\widehat\lambda_{r+1}| \} > c_0 \widehat\lambda_R $  }{
 Compute $\widehat \ba_{rm}^{\rm rcpca}$ as the top left singular vector of $\matk (\widehat \bu_r)\in\RR^{d_m\times (d_S/d_m)}$ for $m \in S$, or $\matk (\widehat \bv_r)\in\RR^{d_m\times (d_{S^c}/d_m)}$ for $m \in S^c$.
}
\Else{
Form disjoint index sets $I_1,...,I_N$ from all contiguous indices $1\le r\le R$ that do not satisfy the above criteria of the eigengap.

For each $I_j$, form $d_{S}\times (d/d_{S})$ matrix $\Xi_{j}=\sum_{\ell\in I_j} \widehat \lambda_{\ell} \widehat \bu_{\ell}\widehat \bv_{\ell}^\top$, and formulate it into a tensor $\Xi_j\in \RR^{d_1\times\cdots\times d_M}$. Then run Procedure \ref{alg:initialize-random} on $\Xi_j$ to obtain $\widehat \ba_{rm}^{\rm rcpca}$ for all $r\in I_j,1\le m\le M$.
}

    \Output{Warm initialization $\hat \ba_{rm}^{\rm rcpca}, 1\le r\le R, 1\le m\le M$}
    \caption{Randomized Composite PCA (\textsc{rc-PCA})}
    \label{alg:initialize-cp}
\end{algorithm}
\vspace{-10pt}

\SetAlgorithmName{Procedure}{procedure}{List of Procedures}
\begin{algorithm}[h!]
\caption{Randomized Projection}\label{alg:initialize-random}
    \SetKwInOut{Input}{Input}
    \SetKwInOut{Output}{Output}
    \Input{Noisy tensor $\Xi\in\RR^{d_1\times\cdots\times d_M}$, rank $s$, $S_1 \subset [M]\backslash\{1\}$, number of random projections $L$, tuning parameter $\nu$.}

    If $S_1=\emptyset$, pick $S_1$ to maximize $\min(d_{S_1},d_{-1}/d_{S_1})$ with $d_{S_1} = \prod_{m\in S_1}d_m$ and $d_{-1} = \prod_{m=2}^M d_m$. Let $S_1\vee S_1^c = [M]\backslash\{1\}$.
    
   \For{$\ell = 1$ to $L$}{
        Randomly draw a standard Gaussian vector $\theta\sim\cN(0, I_{d_1})$. 
        
        Compute $\Xi\times_{1}\theta$, unfold it to be $d_{S_1}\times (d_{-1}/d_{S_1})$ matrix, and compute its leading singular value and singular vector $\eta_\ell,\widetilde \bu_{\ell}, \widetilde \bv_{\ell}$. 
        
        Compute $\widetilde \ba_{\ell m}$ as the top left singular vector of $\matk(\widetilde  \bu_{\ell})\in\RR^{d_m\times (d_{S_1}/d_m)}$ for $m \in S_1$, or $\matk (\widetilde \bv_{\ell})\in\RR^{d_m\times (d_{S_1^c}/d_m)}$ for $m \in S_1^c$. 
        
        Compute $\widetilde \ba_{\ell 1}$ as the top left singular vector of $\Xi\times_{m=2}^M \widetilde  \ba_{\ell m}$.
        
        Add the tuple $(\widetilde \ba_{\ell m},1\le m\le M)$ to $\cS_L$. 
    }
    \For{$r = 1$ to $s$}{
        Among the remaining tuples in $\cS_L$, choose one tuple $(\widetilde \ba_{\ell m},1\le m\le M)$ that correspond to the largest $\|\Xi\times_{m=1}^{M} \widetilde  \ba_{\ell m}\|_2$. Set it to be $\widehat \ba_{rm}^{\rm rcpca}=\widetilde \ba_{\ell m}$.

        Remove all the tuples with $\max_{1\le m\le M} |\widetilde  \ba_{\ell' m}^\top \widehat \ba_{rm}^{\rm rcpca}|>\nu$.
    }

    \Output{Warm initialization $\hat \ba_{rm}^{\rm rcpca}, 1\le r\le s, 1\le m\le M$}
\end{algorithm}
\SetAlgorithmName{Algorithm}{algorithm}{List of Algorithms}

\section{Theoretical Results}
\label{sec:theorems}

In this section, we delve into the statistical attributes of the algorithms introduced previously. Our theoretical framework offers guarantees for consistency and outlines the statistical error rates for the estimated discriminant tensor and the misclassification error, given certain regularity conditions. 
We define $d=\prod_{m=1}^M d_m$, $d_{-m}=d/d_m$, $d_{\min}=\min\{d_1,\cdots,d_M\}$. 
We use $\norm{\hat\ba_{rm}{\hat\ba_{rm}}^{\top} - \ba_{rm}\ba_{rm}^{\top}}=\sqrt{1-(\hat\ba_{rm}^\top\ba_{rm})^2}$ to measure the distance between $\hat\ba_{rm}$ and $\ba_{rm}$. Assume $w_1 \ge w_2 \ge \cdots \ge w_R$. As $\| \ba_{im}\|_2^2=1$, the correlation among columns of $\bA_m$ can be measured by 
\begin{align}\label{corr-k}
\delta_m = \| \bA_m^\top \bA_m - I_{R}\|_{2}. 
\end{align}
We denote $\delta_{\max} =\max\{\delta_1,\cdots, \delta_M \}$.

\medskip
\noindent
\textbf{Global convergence.} We first consider the global convergence rate of Algorithm \ref{alg:tensorlda-cp}, given a suitable warm initialization.
Define
\begin{align}
\alpha&=\sqrt{(1-\delta_{\max})(1-1/(4R))}-(R^{1/2}+1)\psi_0, \label{eq:alpha} \\  
\rho &= 2\alpha^{1-M}\sqrt{R-1}(w_1/w_R) \psi_{0}^{M-2}      \label{eq:rho} , \\
\psi^{\rm ideal} &= \frac{ \sqrt{ \sum_{k=1}^M d_k}}{\sqrt{n}  w_R}   +   \frac{w_1 }{ w_R} \max_{1\le k\le M} \sqrt{\frac{d_k}{n  d_{-k}}} .
\end{align}
The following theorem provides the convergence rates for the CP low rank discriminant tensor. 

\begin{theorem}\label{thm:cp-converge}
Suppose there exist constant $C_0 > 0$ such that $C_0^{-1} \leq \lambda_{\min}(\otimes_{m=1}^M \Sigma_m) \leq \lambda_{\max}(\otimes_{m=1}^M \Sigma_m) \leq C_0,$ where $\lambda_{\min}(\cdot)$ and $\lambda_{\max}(\cdot)$ being the smallest and largest eigenvalues of a matrix, respectively. Assume 
$n_{1} \asymp n_{2}$ and $n=\min\{n_{1},n_{2}\}$.
Let $\Omega_0 = \{\underset{r\in[R],m\in[M]}{\max}\|\hat\ba_{rm}^{(0)}\hat\ba_{rm}^{(0)\top} - \ba_{rm}\ba_{rm}^{\top})\|_2 \le \psi_0 \}$ for any initial estimates $\hat\ba_{rm}^{(0)}$. Suppose $\alpha>0$, $\rho<1$, with the quantities defined in \eqref{eq:alpha} and \eqref{eq:rho}.
After at most $T=O(\log \log (\psi_0/\psi^{\rm ideal}))$ iterations, Algorithm \ref{alg:tensorlda-cp} produces final estimates $\hat{\ba}_{rm}$ for $r\in[R]$ and $m\in[M]$ satisfying 
\begin{align} 
\norm{\hat{\ba}_{rm}\hat{\ba}_{rm}^\top - \ba_{rm}\ba_{rm}^{\top}}_2 &\le C\frac{ \sqrt{\sum_{k=1}^M d_k}}{\sqrt{n}  w_r}   +   C\frac{w_1 }{ w_r} \max_{1\le k\le M} \sqrt{\frac{d_k}{n  d_{-k}}}, \label{eqn:a-bound} \\
\abs{\hat w_r -w_r} &\le C\frac{ \sqrt{\sum_{k=1}^M d_k}}{\sqrt{n} }   +   C w_1  \max_{1\le k\le M} \sqrt{\frac{d_k}{n  d_{-k}}}, \label{eqn:w-bound} \\
\norm{\hat \cB^{\rm cp} - \cB}_{\rm F} &\le C\frac{ \sqrt{\sum_{k=1}^M d_k R}}{\sqrt{n} }   +   C w_1\sqrt{R}  \max_{1\le k\le M} \sqrt{\frac{d_k}{n  d_{-k}}}, \label{eqn:b-bound}
\end{align}
with probability at least $\PP(\Omega_0) - n^{-c} -\sum_{m=1}^M\exp(-c d_m)$. 
\end{theorem}

The statistical convergence rates in \eqref{eqn:a-bound}, \eqref{eqn:w-bound}, \eqref{eqn:b-bound} contain two parts. The first part is similar to the rate in conventional tensor CP decomposition. The second part comes from the estimation accuracy of the mode-$k$ sample precision matrices $\hat\Sigma_k^{-1}$.
As we update the estimation of each individual CP basis vector $\ba_{rm}$ separately in the algorithm, this would remove the bias on the non-orthogonality of the loading vectors.
The theoretical analysis of Algorithm \ref{alg:tensorlda-cp} will be based on various large deviation inequalities and a spectrum analysis of the projected discriminant tensor estimator using spectrum perturbation theorems \citep{wedin1972perturbation,cai2018rate,han2020iterative}. It is worth noting that our analysis accommodates correlated entries in the noise tensor, in contrast to most existing work in tensor decomposition that assumes i.i.d. entries \citep{zhang2018tensor, wang2020learning}.
The proof of Theorem \ref{thm:cp-converge} is provided in Section \ref{sec:proof-cp-converge} in the supplementary material.

In Theorem \ref{thm:cp-converge}, $\psi_0$ is the required accuracy of the initial estimator of CP basis vectors $\ba_{rm}$. When $\psi_0$ is sufficiently small, the conditions $\alpha>0$, $\rho<1$ hold. In view of the definition of $\alpha$ in \eqref{eq:alpha}, 
condition $\alpha>0$ requires $R^{1/2}\psi_0$ be small, with an extra factor $R^{1/2}$ on the initial error in the estimation of individual basis vectors. This is a technical issue due to the need to invert the estimated $A_m^\top A_m$ in our analysis to construct the mode-$m$ projection in Algorithm \ref{alg:tensorlda-cp}. Although the condition $\rho<1$ with $\rho$ defined in \eqref{eq:rho} may seem complex, it is designed to ensure the error contraction effect in each iteration. This ensures that as iterations progress, the error bound will approach the desired statistical upper bound.

\medskip
\noindent
\textbf{Convergence of initialization.}
Theorem \ref{thm:cp-initilization} below presents the performance bounds of the \textsc{rc-PCA} initialization in Algorithm \ref{alg:initialize-cp}, which depends on the coherence (the degree of non-orthogonality) of the CP bases. For simplicity, we focus on two cases: one with a sufficiently large eigengap, and another where all signals $w_r$ are of the same order.
Recall $ \bA_m= (\ba_{1m},\ldots,\ba_{Rm})$. We use
\begin{align}\label{corr-all}
\delta = \| \bA_S^\top \bA_S - I_{R}\|_{2} \vee \| \bA_{S^c}^\top \bA_{S^c} - I_{R}\|_{2}
\end{align}
to measure the correlation of the matrix $\bA_S= (\ba_{1,S},\ldots, \ba_{R,S})\in \RR^{d_S\times R}$ and $\bA_{S^c}= (\ba_{1,S^c},\ldots, \ba_{R,S^c})\in \RR^{d_{S^c}\times R}$, where $S$ maximizes $\min(d_S,d/d_S)$ with $d_S=\prod_{m\in S} d_m, d_{S^c}=\prod_{m\in S^c} d_m$, and $\ba_{r,S} = \vect(\circ_{m\in S} \ba_{rm}), \ba_{r,S^c} = \vect(\circ_{m\in S^c} \ba_{rm})$.

\begin{theorem}\label{thm:cp-initilization}
Suppose there exist constant $C_0 > 0$ such that $C_0^{-1} \leq \lambda_{\min}(\otimes_{m=1}^M \Sigma_m) \leq \lambda_{\max}(\otimes_{m=1}^M \Sigma_m) \leq C_0,$ where $\lambda_{\min}(\cdot)$ and $\lambda_{\max}(\cdot)$ being the smallest and largest eigenvalues of a matrix, respectively. Assume 
$n_{1} \asymp n_{2}$ and $n=\min\{n_{1},n_{2}\}$.

\noindent (i). The eigengaps satisfy $\min\{w_r-w_{r+1},w_{r}-w_{r-1}\} \le c_0 w_R$ for all $1\le r\le R$, with $w_0=\infty, w_{R+1}=0$, and $c_0$ is sufficiently small constant. With probability at least $1-n^{-c}-\sum_{m=1}^M \exp(-c d_m)$, the following error bound holds for the estimation of the loading vectors $\ba_{rm}$ using CPCA in Algorithm \ref{alg:initialize-cp}, 
\begin{align}\label{thm:initial:eq1}
\|\widehat \ba_{rm}^{\rm rcpca}\widehat \ba_{rm}^{\rm rcpca\top}  - \ba_{rm} \ba_{rm}^\top \|_{2} &\le \left(1+2\sqrt{2}(w_1/w_R)\right)\delta+C \phi_{1},
\end{align}
for all $1\le r\le R$, $1\le m\le M$, where $C$ is some positive constants, and
\begin{align}\label{thm:initial:eq2}
\phi_{1} &= \frac{\sqrt{d_S}+\sqrt{d_{S^c}}}{\sqrt{n} w_R} +  \frac{w_1}{w_R} \max_{1\le k\le M} \sqrt{\frac{d_k}{n  d_{-k}}}  .    
\end{align}

\noindent (ii). The eigengaps condition in (i) is not satisfied. Assume $w_1\asymp w_R$ and the number of random projections $L\ge C_0 d_1^2 \vee C_0 d_1 R^{2(w_1/w_R)^2}$. With probability at least $1-n^{-c}-d_1^{-c}-\sum_{m=2}^M e^{-c d_m}$, the following error bound holds for the estimation of the loading vectors $\ba_{rm}$ using Algorithm \ref{alg:initialize-cp},
\begin{align}\label{thm:initial:eq1*}
\|\widehat \ba_{rm}^{\rm rcpca}\widehat \ba_{rm}^{\rm rcpca\top}  - \ba_{rm} \ba_{rm}^\top \|_{2} &\le C \sqrt{ \delta_{\max} }+ C \sqrt{\phi_{2} }.
\end{align}
where $C$ is some positive constants, and
\begin{align}\label{thm:initial:eq3}
\phi_{2} &= \frac{\sqrt{d_1 d_{S_1}}+\sqrt{d_1 d_{S_1^c}} }{\sqrt{n} w_R} +  \frac{w_1}{w_R} \max_{1\le k\le M} \sqrt{\frac{d_k}{n  d_{-k}}}  .    
\end{align}
\end{theorem}

The first term of the upper bounds of $\ba_{rm}$ in \eqref{thm:initial:eq1} and \eqref{thm:initial:eq1*} arises due to the loading vectors $\ba_{rm}$ not being orthogonal, which may be seen as bias. Meanwhile, the subsequent term in \eqref{thm:initial:eq1} is derived from a concentration bound concerning random noise and the estimation of mode-$m$ precision matrices, thereby being describable as a form of stochastic error.
When the eigengap condition is not met, we employ randomized projection to determine the statistical convergence rate as shown in \eqref{thm:initial:eq1*}, which is slower than the rate in \eqref{thm:initial:eq1}. 
A broader result than \eqref{thm:initial:eq1*}, permitting a more general eigen ratio $w_1/w_R$ for part (ii), is detailed in the appendix. 
The proof of Theorem \ref{thm:cp-initilization} is provided in Section \ref{append:proof:initia} in the supplementary material.

If Algorithm \ref{alg:initialize-cp} is used to initialize Algorithm \ref{alg:tensorlda-cp} and $\psi_0$ is taken as the maximum of the right-hand side of \eqref{thm:initial:eq1} or \eqref{thm:initial:eq1*}, then the condition in Theorem \ref{thm:cp-converge} satisfies with $\PP(\Omega_0)\ge 1-n^{-c}-d_1^{-c}-\sum_{m=2}^M e^{-c d_m}$.

\begin{remark} 
We improve the existing initialization methods in tensor CP decomposition from several aspects. First, compare to CPCA in the literature \citep{han2023tensor,han2023cp}, the proposed initialization method allows for repeated singular values in tensor CP decomposition. 
Second, our \textsc{rc-PCA} in Algorithm \ref{alg:initialize-cp} provides weaker incoherence condition for the non-orthogonality of the CP basis vectors under small CP rank. For example, for 3-way tensors, the initial estimator in \cite{anandkumar2014guaranteed} requires the incoherence condition $\max_{i\neq j}\max_{m} \ba_{im}^\top \ba_{jm}  \le {\rm polylog}(d_{\min})/\sqrt{d_{\min}}$. In comparison, \textsc{rc-PCA} initialization $\psi_0\lesssim 1$ leads to $\max_{i\neq j}\max_{m} \ba_{im}^\top \ba_{jm}  \lesssim R^{-2}$. Compared with the previous work, this is a weaker incoherence condition when $R\lesssim d_{\min}^{1/4}$.
\end{remark}

Next, we develop the misclassification error rates of high-dimensional tensor LDA for data with low-rankness structure. We first obtain the upper and lower bounds for the excess misclassification risk of the Algorithm \ref{alg:tensorlda-cp}. Then, the upper and lower bounds together show that the Algorithm \ref{alg:tensorlda-cp} is rate optimal.


Recall that we classify tensor variables under the TGMM, with the tensor LDA rule defined in \eqref{eqn:lda-rule-cp}. 
The classifier's performance is evaluated by its misclassification error
\begin{align*}
\cR_{\btheta}(\hat\Upsilon_{\rm cp}) = \PP_{\btheta}\big(label(\cZ) \neq \hat\Upsilon_{\rm cp}(\cZ)\big),    
\end{align*}
where $\PP_{\btheta}$ denotes the probability with respect to $\cZ \sim \pi_1 \cT\cN(\cM_1; \bSigma) + \pi_2 \cT\cN(\cM_2; \bSigma)$ and $label(\cZ)$ denotes the true class of $\cZ$. In this paper, we use the excess misclassification risk relative to the optimal misclassification error, $\cR_{\btheta}(\hat\Upsilon_{\rm cp}) -\cR_{\rm opt}(\btheta)$, to measure the performance of the classifier $\hat\Upsilon_{\rm cp}$.

\begin{theorem}[Upper bound of misclassification rate]
\label{thm:class-upp-bound}
Assume the conditions in Theorem \ref{thm:cp-converge} all hold. Assume $ w_1\sqrt{\max_m d_m^2 R/(nd)}+\sqrt{\sum_{m=1}^M d_m R/n} = o(\Delta)$ with $n \rightarrow \infty$, where $\Delta=\sqrt{\langle \calB, \; \cD \rangle}$ is the signal-to-noise ratio. 

\noindent (i) If $\Delta\le c_0$, for some $c_0>0$, then with probability at least $1-n^{-c}-d_1^{-c}-\sum_{m=2}^M e^{-cd_m }$, the misclassification rate of classifier $\hat\Upsilon_{\rm cp}$ satisfies
\begin{equation}\label{eqn:lda-mis-cp1}
\cR_{\btheta}(\hat\Upsilon_{\rm cp}) -\cR_{\rm opt}(\btheta) \le C\left(\frac{\sum_{m=1}^M d_m R }{n} \right), 
\end{equation}
for some constant $C,c>0$.\\
(ii) If $\Delta\to\infty$ as $n\to \infty$, then there exists $\vartheta_n=o(1)$, with probability at least $1-n^{-c}-d_1^{-c}-\sum_{m=2}^M e^{-cd_m }$, the misclassification rate of classifier $\hat\Upsilon_{\rm cp}$ satisfies
\begin{equation}\label{eqn:lda-mis-cp2}
\cR_{\btheta}(\hat\Upsilon_{\rm cp}) -\cR_{\rm opt}(\btheta) \le C \exp\left\{-\left(\frac18+\vartheta_n\right)\Delta^2 \right\} \left(\frac{\sum_{m=1}^M d_m R }{n}  \right), 
\end{equation}
for some constant $C,c>0$.
\end{theorem}

Compared with Theorem \ref{thm:cp-converge}, the excess misclassification rates in \eqref{eqn:lda-mis-cp1} and \eqref{eqn:lda-mis-cp2} only contain one part. The error induced by the estimation accuracy of the mode-$m$ precision matrix becomes negligible in the excess misclassification rates.

To understand the difficulty of the tensor classification problem, it is essential to obtain the minimax lower bounds for the excess misclassification risk. We especially consider the following parameter space of CP low-rank discriminant tensors,
\begin{align*}
\calH= \Big\{\;& \btheta=(\cM_1, \cM_2, \bSigma): \cM_1,\cM_2 \in \RR^{d_1 \times \cdots \times d_M}, \; \bSigma = [\Sigma_m]_{m=1}^M, \; \Sigma_m \in \RR^{d_m \times d_m}, \; \text{for some} \; C_0>0, \\
&C_0^{-1} \leq \lambda_{\min}(\otimes_{m=1}^M \Sigma_m) \leq \lambda_{\max}(\otimes_{m=1}^M \Sigma_m) \leq C_0, \; \calB =\sum_{r=1}^R w_r \circ_{m=1}^M \ba_{rm}\; \text{with}\; \|\ba_{rm}\|_2=1 \Big\}   . 
\end{align*}
The following lower bound holds over $\cH$.

\begin{theorem}[Lower bound of misclassification rate] \label{thm:class-lower-bound}
Under the TGMM, the minimax risk of excess misclassification error over the parameter space $\calH$ satisfies the following conditions.

\noindent (i) If $c_1<\Delta \le c_2$ for some $c_1,c_2>0$, then for any $\gamma>0$, there exists some constant $C_{\gamma}>0$ such that
\begin{equation}\label{eqn:lda-lbd-tucker1}
\inf_{\hat \Upsilon_{\rm cp}} \sup_{\btheta \in \calH} \PP\left( \cR_{\btheta}(\hat\Upsilon_{\rm cp}) -\cR_{\rm opt}(\btheta) \ge C_{\gamma}  \frac{\sum_{m=1}^M d_m R}{n}  \right) \ge 1-\gamma   .   
\end{equation}
(ii) If $\Delta\to\infty$ as $n\to \infty$, then then for any $\gamma>0$, there exists some constant $C_{\gamma}>0$ and $\vartheta_n=o(1)$ such that
\begin{equation}\label{eqn:lda-lbd-tucker2}
\inf_{\hat \Upsilon_{\rm cp}} \sup_{\btheta \in \calH} \PP\left(\cR_{\btheta}(\hat\Upsilon_{\rm cp}) -\cR_{\rm opt}(\btheta) \ge C_{\gamma} \exp\left\{-\left(\frac18+\vartheta_n\right)\Delta^2 \right\} \frac{\sum_{m=1}^M d_m R }{n}  \right) \ge 1-\gamma   .   
\end{equation}
\end{theorem}

Combined with the upper bounds of the excess misclassification risk in \eqref{eqn:lda-mis-cp1} and \eqref{eqn:lda-mis-cp2}, the convergence rates are minimax rate optimal.

\section{Simulation}
\label{sec:simu}
This section presents extensive simulations demonstrating the superiority of the DISTIP-CP algorithm in estimating the discriminant tensor, leading to exceptional tensor classification accuracy. We explore various experimental configurations and employ multiple evaluation metrics to substantiate our theoretical findings. The simulations primarily focus on classifications involving order-3 tensor predictors, with results for order-4 tensor predictors provided in Appendix \ref{append:simulation}. 
We incorporate state-of-the-art classification methods as competitors, including the Tensor LDA rule with the sample discriminant tensor as a benchmark. Additionally, we include the CATCH estimator \citep{pan2019covariate} and the Discriminant Tensor Iterative Projection using Tucker low-rank discriminant tensor (DISTIP-Tucker) algorithm for a thorough comparison.

\subsection{Data Generation}
We generate training and testing sets from tensor normal distributions $\cT\cN(\cM_k; \bSigma)$ with $k \in [2]$. The training set comprises $n_k$ samples of $\calX^{(k)}$, while the testing set consists of 500 samples per class. Common covariance matrices $\bSigma=[\Sigma_m]_{m=1}^M$ are either identity matrices or have unit diagonals with off-diagonal elements $2/d_m$.
CP bases for $\calB$ are generated as either orthogonal or non-orthogonal. Orthogonal bases $\bA_m \in \mathbb{R}^{d_m \times R}$ are derived from random matrices with i.i.d. $U(0,1)$ entries and orthonormalized via QR decomposition. Non-orthogonal bases $\tilde \bA_m$ are constructed using $\tilde \ba_{1m} = \ba_{1m}$ and $\tilde \ba_{rm} = (\ba_{1m} + \eta \ba_{rm}) / \norm{\ba_{1m} + \eta \ba_{rm}}_2$ for $r \ge 2$, where $\eta = (\vartheta^{-2/M} - 1)^{1/2}$ and $\vartheta = \delta / (r-1)$. The discriminant tensor $\calB$ is then configured with these bases and the signal strengths $\{w_r, r \in [R]\}$. For order-3 tensors, we explore four configurations combining orthogonal/non-orthogonal CP bases with identity/general covariance matrices. Across all settings, we maintain $d_1=d_2=d_3=30$, $R=5$, $n_1 = n_2 = 200$, and perform 50 repetitions. We vary signal strengths ${w_r}$ as follows:
\begin{itemize}
\item Equal strengths: $w_{\max}/w_{\min}=1$ with $w_{\max} \in {1.5, 2.5, 3.5, 5.0}$;
\item Unequal strengths: $w_{r+1}/w_r=1.25$ with $w_{max} \in {3, 4, 6}$.
\end{itemize}
For identity covariance matrices, we set $\calM_1 = 0$ and $\calM_2=\calB$. For general covariance matrices, we set $\calM_2=\calB \times_{m=1}^M \Sigma_m$. For non-orthogonal CP bases, we set the parameter $\delta = 0.1$ in the data generating process.

\subsection{Justiﬁcation of global convergence}
We employ various metrics to verify the theoretical results on the global convergence of DISTIP-CP in estimating the discriminant tensor $\calB$. Figures \ref{fig:1} and \ref{fig:2} illustrate the logarithm of CP bases estimation errors, measured by $\max_{r,m} \| \hat{\ba}_{rm} \hat{\ba}_{rm}^\top - \ba_{rm} \ba_{rm}^\top \|_2$, across four configurations with varying signal strengths. In the figures, ALS denotes SVD-based initialization \citep{anandkumar2014guaranteed} but still utilizes our iterative algorithm in Algorithm \ref{alg:tensorlda-cp}.
Notably, even in challenging scenarios with non-orthogonal CP bases and equal signal strengths (Figures \ref{fig:2a} and \ref{fig:4a}), the errors of the proposed method remain low. This indicates the effectiveness of our initialization method, particularly when signal strengths are not well-separated and orthogonality assumptions are relaxed.

\vspace{2em}
\begin{figure}[H]
    \centering
    \begin{subfigure}[b]{0.49\textwidth}
        \centering
        \includegraphics[width=\textwidth]{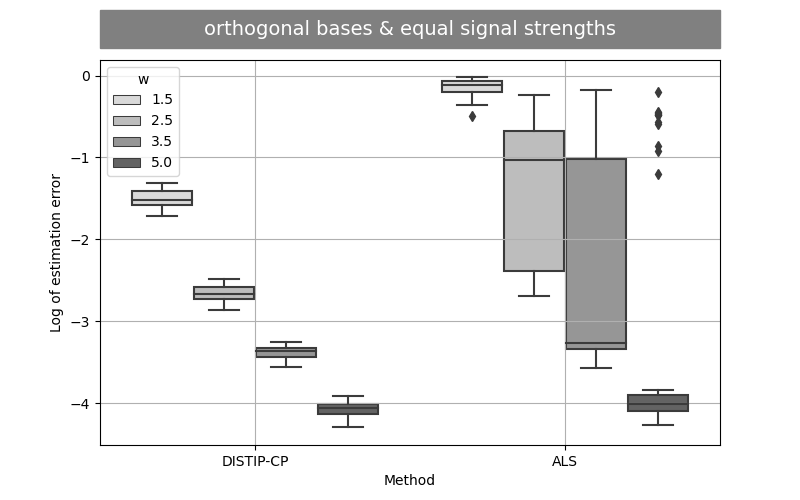}
        \caption{}
        \label{fig:1a}
    \end{subfigure}
    \hfill
    \begin{subfigure}[b]{0.49\textwidth}
        \centering
        \includegraphics[width=\textwidth]{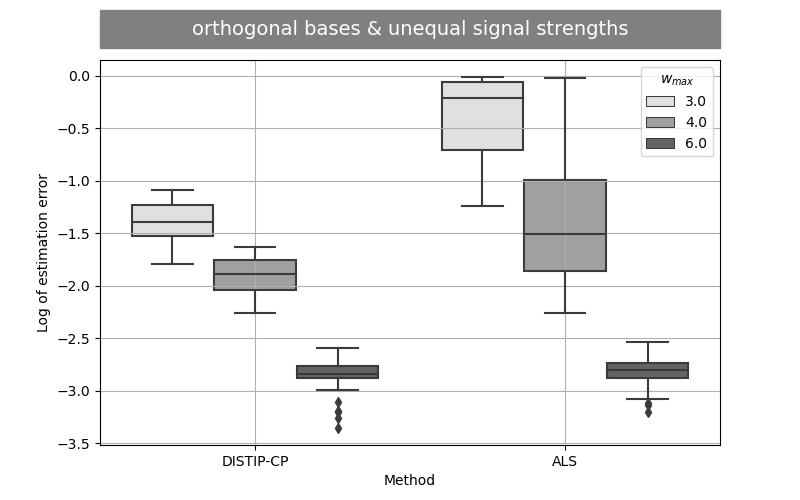}
        \caption{}
        \label{fig:1b}
    \end{subfigure}
    
    \vspace{1em}
    
    \begin{subfigure}[b]{0.49\textwidth}
        \centering
        \includegraphics[width=\textwidth]{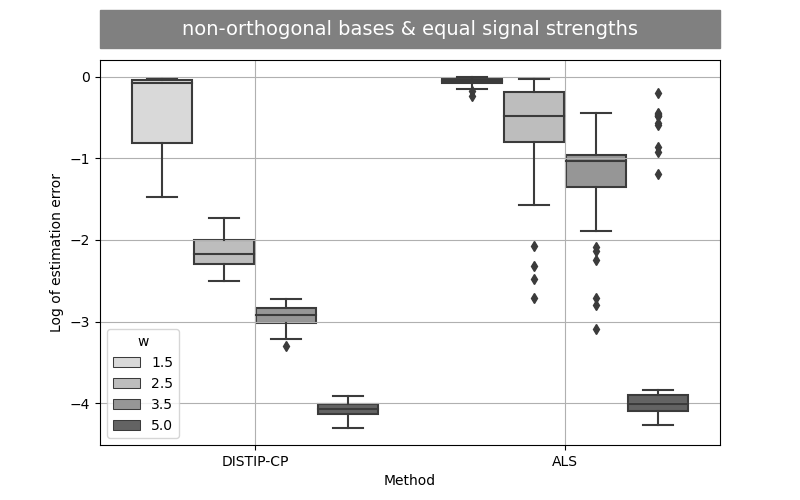}
        \caption{}
        \label{fig:2a}
    \end{subfigure}
    \hfill
    \begin{subfigure}[b]{0.49\textwidth}
        \centering
        \includegraphics[width=\textwidth]{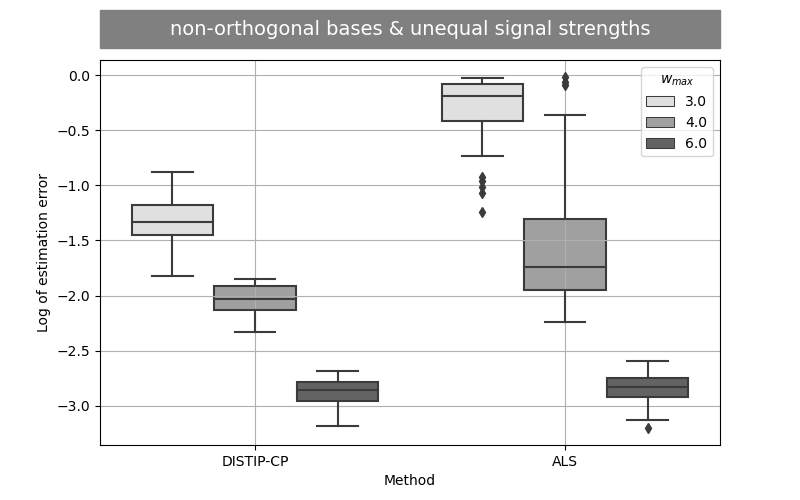}
        \caption{}
        \label{fig:2b}
    \end{subfigure}
    \caption{Logarithmic estimation errors for CP basis using DISTIP-CP and ALS algorithms with identity covariance matrices $\Sigma_m, m \in [3]$.}
    \label{fig:1}
\end{figure}

\begin{figure}[H]
    \centering
    \begin{subfigure}[b]{0.49\textwidth}
        \centering
        \includegraphics[width=\textwidth]{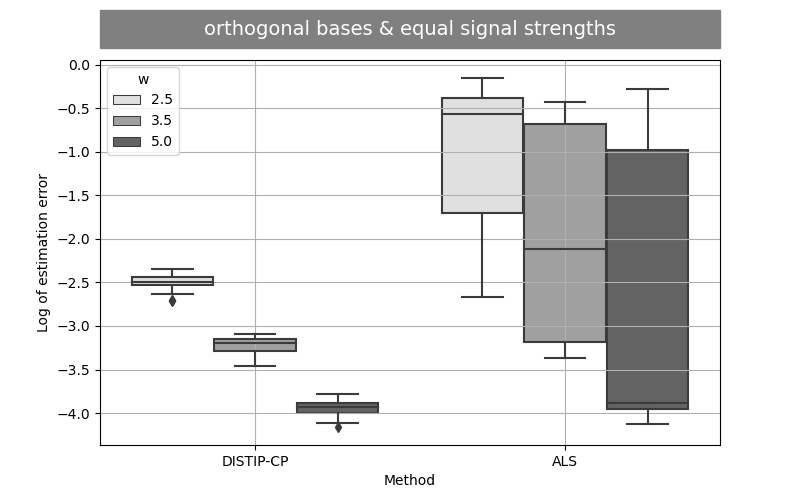}
        \caption{}
        \label{fig:3a}
    \end{subfigure}
    \hfill
    \begin{subfigure}[b]{0.49\textwidth}
        \centering
        \includegraphics[width=\textwidth]{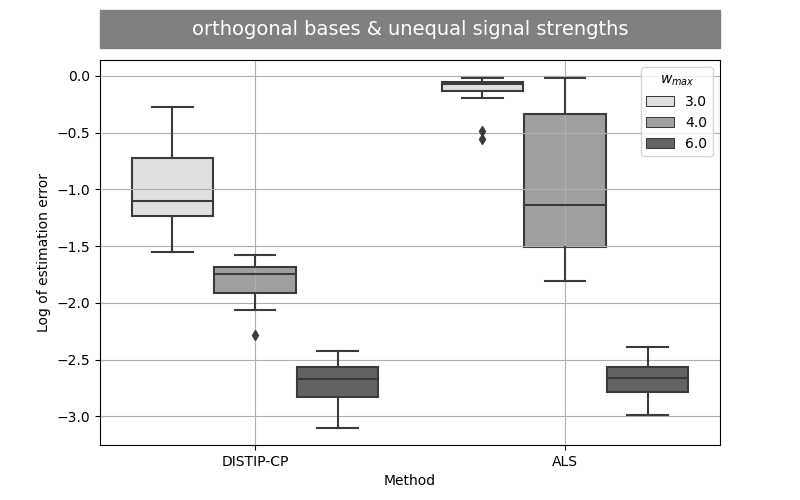}
        \caption{}
        \label{fig:3b}
    \end{subfigure}
    
    \vspace{1em}
    
    \begin{subfigure}[b]{0.49\textwidth}
        \centering
        \includegraphics[width=\textwidth]{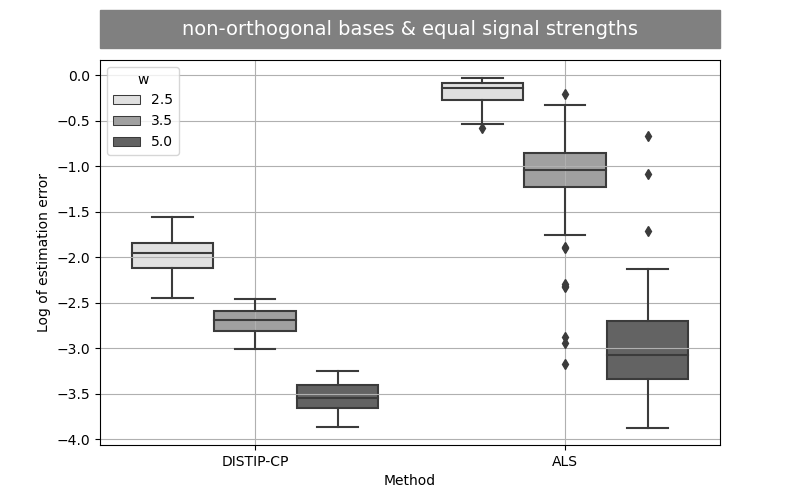}
        \caption{}
        \label{fig:4a}
    \end{subfigure}
    \hfill
    \begin{subfigure}[b]{0.49\textwidth}
        \centering
        \includegraphics[width=\textwidth]{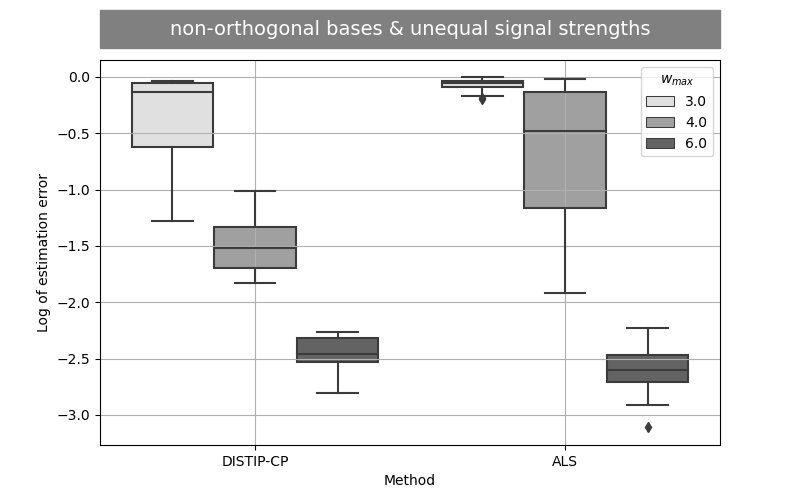}
        \caption{}
        \label{fig:4b}
    \end{subfigure}
    \caption{Logarithmic estimation errors for CP basis using DISTIP-CP and ALS algorithms with general covariance matrices $\Sigma_m, m \in [3]$.}
    \label{fig:2}
\end{figure}
\vspace{1em}

\textbf{Comparison with Alternating Least Squares (ALS)}: Figures \ref{fig:1} and \ref{fig:2} illustrate DISTIP-CP's superior performance over ALS across all scenarios, especially when signal strengths are equal. This highlights DISTIP-CP's robustness in scenarios where signal strengths are not well-separated---a challenging situation for ALS. The consistent convergence of DISTIP-CP, achieved through our \textsc{rc-PCA} initialization, outperforms ALS's SVD-based initialization \citep{anandkumar2014guaranteed}, providing more reliable starting points and ensuring global convergence in diverse settings.

Tables \ref{tab: 1} and \ref{tab: 2} present the convergence rates of the estimation errors for $\calB$, measured by $\|\widehat \calB  - \calB \|_{\rm F} / \| \calB \|_{\rm F}$ under various configurations. DISTIP-CP consistently achieves low estimation errors, although slightly higher under general covariance matrices compared to identity matrices. This aligns with theoretical results, where the upper bound of $\|\widehat \calB  - \calB \|_{\rm F}$ is influenced by the accuracy of estimating sample precision matrices. Orthogonal CP bases typically yield lower errors than non-orthogonal bases, particularly when signal strengths are equal, as demonstrated in Figures \ref{fig:1a}, \ref{fig:2a}, \ref{fig:3a}, and \ref{fig:4a}.
\vspace{1.5em}

\begin{table}[H]
\centering
\begin{tabular}{>{\centering\arraybackslash}p{2.6cm}*{4}{>{\centering\arraybackslash}p{1.45cm}}*{3}{>{\centering\arraybackslash}p{1.6cm}}}
\toprule
\multirow{2}{*}{Algorithms} & \multicolumn{4}{c}{Equal Signal Strengths} & \multicolumn{3}{c}{Unequal Signal Strengths} \\
\cmidrule(lr){2-5} \cmidrule(lr){6-8}
& $w = 1.5$ & $w = 2.5$ & $w = 3.5$ & $w = 5.0$ & $w_{max} = 3$ & $w_{max} = 4$ & $w_{max} = 6$ \\
\midrule
Sample $\calB$ & 4.98\textsubscript{(0.36)} & 2.97\textsubscript{(0.21)} & 2.10\textsubscript{(0.14)} & 1.49\textsubscript{(0.10)} & 3.49\textsubscript{(0.24)} & 2.62\textsubscript{(0.20)} & 1.76\textsubscript{(0.12)} \\
DISTIP-CP & 0.75\textsubscript{(0.05)} & 0.41\textsubscript{(0.03)} & 0.29\textsubscript{(0.03)} & 0.23\textsubscript{(0.03)} & 0.54\textsubscript{(0.04)} & 0.43\textsubscript{(0.05)} & 0.29\textsubscript{(0.04)} \\
DISTIP-Tucker & 1.14\textsubscript{(0.04)} & 0.49\textsubscript{(0.02)} & 0.34\textsubscript{(0.02)} & 0.24\textsubscript{(0.02)} & 0.71\textsubscript{(0.04)} & 0.47\textsubscript{(0.04)} & 0.28\textsubscript{(0.02)} \\
CATCH & 1.10\textsubscript{(0.04)} & 1.02\textsubscript{(0.01)} & 0.99\textsubscript{(0.00)} & 0.98\textsubscript{(0.00)} & 1.04\textsubscript{(0.01)} & 1.01\textsubscript{(0.01)} & 0.98\textsubscript{(0.00)} \\
\midrule
\multirow{2}{*}{Algorithms} & \multicolumn{4}{c}{Equal Signal Strengths} & \multicolumn{3}{c}{Unequal Signal Strengths} \\
\cmidrule(lr){2-5} \cmidrule(lr){6-8}
& $w = 1.5$ & $w = 2.5$ & $w = 3.5$ & $w = 5.0$ & $w_{max} = 3$ & $w_{max} = 4$ & $w_{max} = 6$ \\
\midrule
Sample $\calB$ & 4.92\textsubscript{(0.34)} & 2.87\textsubscript{(0.20)} & 2.07\textsubscript{(0.13)} & 1.43\textsubscript{(0.12)} & 3.54\textsubscript{(0.24)} & 2.62\textsubscript{(0.20)} & 1.81\textsubscript{(0.13)} \\
DISTIP-CP & 0.91\textsubscript{(0.08)} & 0.45\textsubscript{(0.03)} & 0.32\textsubscript{(0.02)} & 0.23\textsubscript{(0.02)} & 0.52\textsubscript{(0.04)} & 0.36\textsubscript{(0.03)} & 0.25\textsubscript{(0.03)} \\
DISTIP-Tucker & 1.09\textsubscript{(0.04)} & 0.49\textsubscript{(0.03)} & 0.33\textsubscript{(0.02)} & 0.23\textsubscript{(0.02)} & 0.70\textsubscript{(0.03)} & 0.47\textsubscript{(0.03)} & 0.28\textsubscript{(0.02)} \\
CATCH & 1.02\textsubscript{(0.01)} & 0.99\textsubscript{(0.00)} & 0.97\textsubscript{(0.01)} & 0.93\textsubscript{(0.00)} & 1.03\textsubscript{(0.01)} & 1.00\textsubscript{(0.00)} & 0.97\textsubscript{(0.01)} \\
\bottomrule
\end{tabular}
\caption{Estimation errors for the discriminant tensor $\calB$ under identity covariance matrices. Results are shown for both orthogonal (upper section) and non-orthogonal (lower section) CP bases, comparing different estimation methods.}
\label{tab: 1}
\end{table}
\vspace{1.5em}

\begin{table}[H]
\centering
\begin{tabular}{>{\centering\arraybackslash}p{2.6cm}*{3}{>{\centering\arraybackslash}p{1.4cm}}*{3}{>{\centering\arraybackslash}p{1.65cm}}}
\toprule
\multirow{2}{*}{Algorithms} & \multicolumn{3}{c}{Equal Signal Strengths} & \multicolumn{3}{c}{Unequal Signal Strengths} \\
\cmidrule(lr){2-4} \cmidrule(lr){5-7}
& $w = 2.5$ & $w = 3.5$ & $w = 5.0$ & $w_{max} = 3$ & $w_{max} = 4$ & $w_{max} = 6$ \\
\midrule
Sample $\calB$ & 3.19\textsubscript{(0.22)} & 2.29\textsubscript{(0.15)} & 1.58\textsubscript{(0.11)} & 3.78\textsubscript{(0.26)} & 2.81\textsubscript{(0.19)} & 1.90\textsubscript{(0.12)} \\
DISTIP-CP & 0.44\textsubscript{(0.04)} & 0.31\textsubscript{(0.03)} & 0.22\textsubscript{(0.02)} & 0.58\textsubscript{(0.05)} & 0.40\textsubscript{(0.03)} & 0.26\textsubscript{(0.02)} \\
DISTIP-Tucker & 0.53\textsubscript{(0.02)} & 0.36\textsubscript{(0.03)} & 0.25\textsubscript{(0.03)} & 0.80\textsubscript{(0.04)} & 0.54\textsubscript{(0.03)} & 0.31\textsubscript{(0.03)} \\
CATCH & 1.12\textsubscript{(0.04)} & 1.01\textsubscript{(0.01)} & 0.98\textsubscript{(0.01)} & 1.04\textsubscript{(0.01)} & 1.00\textsubscript{(0.01)} & 0.97\textsubscript{(0.01)} \\
\midrule
\multirow{2}{*}{Algorithms} & \multicolumn{3}{c}{Equal Signal Strengths} & \multicolumn{3}{c}{Unequal Signal Strengths} \\
\cmidrule(lr){2-4} \cmidrule(lr){5-7}
& $w = 2.5$ & $w = 3.5$ & $w = 5.0$ & $w_{max} = 3$ & $w_{max} = 4$ & $w_{max} = 6$ \\
\midrule
Sample $\calB$ & 3.07\textsubscript{(0.20)} & 2.23\textsubscript{(0.15)} & 1.55\textsubscript{(0.10)} & 3.66\textsubscript{(0.20)} & 2.78\textsubscript{(0.20)} & 1.84\textsubscript{(0.14)} \\
DISTIP-CP & 0.49\textsubscript{(0.03)} & 0.35\textsubscript{(0.03)} & 0.24\textsubscript{(0.02)} & 0.67\textsubscript{(0.05)} & 0.45\textsubscript{(0.04)} & 0.29\textsubscript{(0.02)} \\
DISTIP-Tucker & 0.53\textsubscript{(0.04)} & 0.35\textsubscript{(0.03)} & 0.23\textsubscript{(0.02)} & 0.78\textsubscript{(0.04)} & 0.54\textsubscript{(0.03)} & 0.30\textsubscript{(0.02)} \\
CATCH & 1.04\textsubscript{(0.02)} & 0.99\textsubscript{(0.01)} & 0.96\textsubscript{(0.01)} & 1.03\textsubscript{(0.01)} & 0.99\textsubscript{(0.01)} & 0.95\textsubscript{(0.01)} \\
\bottomrule
\end{tabular}
\caption{
Estimation errors for the discriminant tensor $\calB$ under general covariance matrices. Results are shown for both orthogonal (upper section) and non-orthogonal (lower section) CP bases, comparing different estimation methods.}
\label{tab: 2}
\end{table}

\subsection{Comparison with other methods}
In this subsection, we compare the proposed DISTIP-CP method with CATCH \citep{pan2019covariate}, DISTIP-Tucker, and the sample discriminant tensor for binary classification tasks. We also attempted to implement the adaptive estimation by \cite{cai2019high}, but its resource-intensive nature made it impractical. Tables \ref{tab: 1} and \ref{tab: 2} compare estimation errors for $\calB$, while Tables \ref{tab: 3} and \ref{tab: 4} show classification errors, defined as $\frac{\sum_{i=1}^{n_1}\hat\Upsilon(\calX_i^{(1)}) + \sum_{i=1}^{n_2}(1-\hat\Upsilon(\calX_i^{(2)}))}{n_1 + n_2}$, where $n_1=n_2=500$ in the testing set.

\noindent \textbf{Comparison with Sample Discriminant Tensor $\calB$}: DISTIP-CP significantly outperforms the sample discriminant tensor in both estimation accuracy (Tables \ref{tab: 1} and \ref{tab: 2}) and classification performance (Tables \ref{tab: 3} and \ref{tab: 4}). This superiority stems from DISTIP-CP's incorporation of the low-rank CP structure, which effectively reduces noise and captures the essential features of the discriminant tensor. The advantage is particularly pronounced in low-signal scenarios, where the sample estimator struggles with high-dimensional noise. As signal strength increases, the performance gap narrows, but DISTIP-CP consistently maintains an edge.

\noindent \textbf{Comparison with DISTIP-Tucker}: DISTIP-CP outperforms DISTIP-Tucker, particularly at smaller signal strengths, due to its superior capture of the underlying tensor structure.

\noindent \textbf{Comparison with CATCH}: Although the estimation error of the discriminant tensor using CATCH appears competitive, this perception is misleading due to the choice of metric. CATCH outputs a sparse discriminant tensor with most entries being zero, which keeps the metric $\|\widehat \calB  - \calB \|_{\rm F} / \| \calB \|_{\rm F}$ close to 1. However, its classification performance is poor, even compared to the sample discriminant tensor $\calB$ (Tables \ref{tab: 3} and \ref{tab: 4}). This discrepancy highlights the superiority of DISTIP-CP's decomposable low-rank tensor structure over CATCH's sparsity approach for classification accuracy under our settings. DISTIP-CP better captures the underlying low rank tensor structure, leading to more accurate estimations and, consequently, more precise classifications.

In summary, DISTIP-CP with \textsc{rc-PCA} initialization excels in discriminant tensor estimation and classification across diverse scenarios, demonstrating its robustness to varying signal strengths, CP basis orthogonality, and covariance structures.

\vspace{2em}
\begin{table}[H]
\centering
\begin{tabular}{>{\centering\arraybackslash}p{2.6cm}*{4}{>{\centering\arraybackslash}p{1.3cm}}*{3}{>{\centering\arraybackslash}p{1.65cm}}}
\toprule
\multirow{2}{*}{Algorithms} & \multicolumn{4}{c}{Equal Signal Strengths} & \multicolumn{3}{c}{Unequal Signal Strengths} \\
\cmidrule(lr){2-5} \cmidrule(lr){6-8}
& $w = 1.5$ & $w = 2.5$ & $w = 3.5$ & $w = 5.0$ & $w_{max} = 3$ & $w_{max} = 4$ & $w_{max} = 6$ \\
\midrule
Sample $\calB$ & 0.37\textsubscript{(0.02)} & 0.19\textsubscript{(0.01)} & 0.05\textsubscript{(0.01)} & 0.00\textsubscript{(0.00)} & 0.26\textsubscript{(0.01)} & 0.13\textsubscript{(0.01)} & 0.01\textsubscript{(0.00)} \\
DISTIP-CP & 0.09\textsubscript{(0.01)} & 0.01\textsubscript{(0.00)} & 0.00\textsubscript{(0.00)} & 0.00\textsubscript{(0.00)} & 0.02\textsubscript{(0.01)} & 0.00\textsubscript{(0.00)} & 0.00\textsubscript{(0.00)} \\
DISTIP-Tucker & 0.23\textsubscript{(0.03)} & 0.01\textsubscript{(0.00)} & 0.00\textsubscript{(0.00)} & 0.00\textsubscript{(0.00)} & 0.03\textsubscript{(0.00)} & 0.00\textsubscript{(0.00)} & 0.00\textsubscript{(0.00)} \\
CATCH & 0.45\textsubscript{(0.02)} & 0.38\textsubscript{(0.02)} & 0.26\textsubscript{(0.03)} & 0.11\textsubscript{(0.03)} & 0.41\textsubscript{(0.02)} & 0.34\textsubscript{(0.02)} & 0.19\textsubscript{(0.03)} \\
\midrule
\multirow{2}{*}{Algorithms} & \multicolumn{4}{c}{Equal Signal Strengths} & \multicolumn{3}{c}{Unequal Signal Strengths} \\
\cmidrule(lr){2-5} \cmidrule(lr){6-8}
& $w = 1.5$ & $w = 2.5$ & $w = 3.5$ & $w = 5.0$ & $w_{max} = 3$ & $w_{max} = 4$ & $w_{max} = 6$ \\
\midrule
Sample $\calB$ & 0.36\textsubscript{(0.02)} & 0.19\textsubscript{(0.01)} & 0.05\textsubscript{(0.01)} & 0.00\textsubscript{(0.00)} & 0.26\textsubscript{(0.02)} & 0.13\textsubscript{(0.01)} & 0.01\textsubscript{(0.00)} \\
DISTIP-CP & 0.13\textsubscript{(0.03)} & 0.01\textsubscript{(0.00)} & 0.00\textsubscript{(0.00)} & 0.00\textsubscript{(0.00)} & 0.02\textsubscript{(0.01)} & 0.00\textsubscript{(0.00)} & 0.00\textsubscript{(0.00)} \\
DISTIP-Tucker & 0.19\textsubscript{(0.02)} & 0.01\textsubscript{(0.00)} & 0.00\textsubscript{(0.00)} & 0.00\textsubscript{(0.00)} & 0.03\textsubscript{(0.00)} & 0.00\textsubscript{(0.00)} & 0.00\textsubscript{(0.00)} \\
CATCH & 0.36\textsubscript{(0.02)} & 0.25\textsubscript{(0.03)} & 0.08\textsubscript{(0.02)} & 0.00\textsubscript{(0.00)} & 0.39\textsubscript{(0.02)} & 0.30\textsubscript{(0.02)} & 0.13\textsubscript{(0.03)} \\
\bottomrule
\end{tabular}
\caption{Misclassification rates for binary classification under identity covariance matrices. Upper section: orthogonal CP bases; lower section: non-orthogonal CP bases for $\calB$.}
\label{tab: 3}
\end{table}
\vspace{1.5em}

\begin{table}[H]
\centering
\begin{tabular}{>{\centering\arraybackslash}p{2.6cm}*{3}{>{\centering\arraybackslash}p{1.35cm}}*{3}{>{\centering\arraybackslash}p{1.65cm}}}
\toprule
\multirow{2}{*}{Algorithms} & \multicolumn{3}{c}{Equal Signal Strengths} & \multicolumn{3}{c}{Unequal Signal Strengths} \\
\cmidrule(lr){2-4} \cmidrule(lr){5-7}
& $w = 2.5$ & $w = 3.5$ & $w = 5.0$ & $w_{max} = 3$ & $w_{max} = 4$ & $w_{max} = 6$ \\
\midrule
Sample $\calB$ & 0.18\textsubscript{(0.02)} & 0.05\textsubscript{(0.01)} & 0.00\textsubscript{(0.00)} & 0.26\textsubscript{(0.03)} & 0.13\textsubscript{(0.02)} & 0.01\textsubscript{(0.01)} \\
DISTIP-CP & 0.01\textsubscript{(0.00)} & 0.00\textsubscript{(0.00)} & 0.00\textsubscript{(0.00)} & 0.02\textsubscript{(0.01)} & 0.00\textsubscript{(0.00)} & 0.00\textsubscript{(0.00)} \\
DISTIP-Tucker & 0.01\textsubscript{(0.00)} & 0.00\textsubscript{(0.00)} & 0.00\textsubscript{(0.00)} & 0.05\textsubscript{(0.01)} & 0.00\textsubscript{(0.00)} & 0.00\textsubscript{(0.00)} \\
CATCH & 0.42\textsubscript{(0.02)} & 0.29\textsubscript{(0.04)} & 0.16\textsubscript{(0.04)} & 0.35\textsubscript{(0.04)} & 0.24\textsubscript{(0.05)} & 0.09\textsubscript{(0.04)} \\
\midrule
\multirow{2}{*}{Algorithms} & \multicolumn{3}{c}{Equal Signal Strengths} & \multicolumn{3}{c}{Unequal Signal Strengths} \\
\cmidrule(lr){2-4} \cmidrule(lr){5-7}
& $w = 2.5$ & $w = 3.5$ & $w = 5.0$ & $w_{max} = 3$ & $w_{max} = 4$ & $w_{max} = 6$ \\
\midrule
Sample $\calB$ & 0.17\textsubscript{(0.02)} & 0.04\textsubscript{(0.01)} & 0.00\textsubscript{(0.00)} & 0.25\textsubscript{(0.02)} & 0.13\textsubscript{(0.02)} & 0.01\textsubscript{(0.01)} \\
DISTIP-CP & 0.01\textsubscript{(0.00)} & 0.00\textsubscript{(0.00)} & 0.00\textsubscript{(0.00)} & 0.03\textsubscript{(0.01)} & 0.00\textsubscript{(0.00)} & 0.00\textsubscript{(0.00)} \\
DISTIP-Tucker & 0.01\textsubscript{(0.00)} & 0.00\textsubscript{(0.00)} & 0.00\textsubscript{(0.00)} & 0.04\textsubscript{(0.01)} & 0.00\textsubscript{(0.00)} & 0.00\textsubscript{(0.00)} \\
CATCH & 0.35\textsubscript{(0.03)} & 0.20\textsubscript{(0.05)} & 0.02\textsubscript{(0.02)} & 0.31\textsubscript{(0.04)} & 0.20\textsubscript{(0.05)} & 0.04\textsubscript{(0.03)} \\
\bottomrule
\end{tabular}
\caption{
Misclassification rates for binary classification under identity covariance matrices. Upper section: orthogonal CP bases; lower section: non-orthogonal CP bases for $\calB$.}
\label{tab: 4}
\end{table}


\section{Real Data Analysis}
\label{sec:appl}
We evaluate our method on the \href{https://chrsmrrs.github.io/datasets/docs/home/}{MUTAG dataset}, a benchmark for graph classification tasks. This dataset comprises 188 molecular graphs, each representing a chemical compound. The nodes in these graphs are labeled by atom types (encoded as one-hot vectors), and the undirected edges represent chemical bonds between atoms. Each graph is classified as either mutagenic (1) or non-mutagenic (0), indicating its potential to induce genetic mutations.
To leverage the topological features of these graphs, we preprocess the data using the Tensor-view Topological Graph Neural Network (TTG-NN) approach proposed by \cite{wen2024tensorview}. This preprocessing step transforms each graph into a tensor of dimensions $64 \times 32 \times 32$, capturing the complex structural information of the molecules.

We partition the dataset into 5 non-overlapping pairs of training (150 samples) and testing (38 samples) sets. This design ensures each of the 188 observations appears in exactly one testing set, encompassing the entire dataset. For each pair, we choose the low-rank $R$ using leave-10-out cross-validation on the training set. Table \ref{tab:5} presents the averaged classification accuracy across these 5 partitions.

\vspace{1em}
\begin{table}[h!]
    \centering
    \begin{tabular}{lcccc}
        \hline
        \textbf{Methods} & \textbf{Sample $\calB$} & \textbf{DISTIP-CP} & \textbf{DISTIP-Tucker} & \textbf{CATCH} \\
        \hline
        Classification Accuracy & 0.761 & 0.862 & 0.782 & 0.681 \\
        \hline
    \end{tabular}
    \caption{Averaged classification accuracy across 5 pairs for each method.}
    \label{tab:5}
\end{table}
Recognizing that the results may depend on the randomness of data splits, we conduct a more comprehensive evaluation. We repeat the experiment 100 times, each time randomly sampling 38 observations for testing and using the remaining 150 for training. Table \ref{tab:6} summarizes the mean classification accuracy and standard deviations over these 100 replications.

\vspace{1em}
\begin{table}[h!]
    \centering
    \begin{tabular}{lcccc}
        \hline
        \textbf{Methods} & \textbf{Sample $\calB$} & \textbf{DISTIP-CP} & \textbf{DISTIP-Tucker} & \textbf{CATCH} \\
        \hline
        \multicolumn{1}{c}{Mean} & 0.71 & 0.82 & 0.72 & 0.70 \\
        \multicolumn{1}{c}{Std} & 0.046 & 0.026 & 0.052 & 0.008 \\
        \hline
    \end{tabular}
    \caption{Classification Accuracy (mean and standard deviation) for 100 replications.}
    \label{tab:6}
\end{table}
The MUTAG dataset analysis showcases DISTIP-CP's superiority in real-world tensor classification. Outperforming competitors across evaluations, it demonstrates the highest accuracy and lowest variance, highlighting its effectiveness in capturing complex tensor structures and potential for reliable predictions.

\section{Conclusion} \label{sec:summ}
This paper introduces a novel tensor classification framework that addresses critical gaps in the existing literature. The proposed approach employs a CP low-rank structure for the discriminant tensor, which provides a more flexible and parsimonious representation than current methods. This structure effectively captures complex interactions in tensor data that traditional sparse tensors fail to model, marking a significant innovation in the field.

To capitalize on this structure, we developed the Tensor HD CP-LDA algorithm, incorporating a novel Randomized Composite PCA (\textsc{rc-PCA}) initialization scheme and an iterative refinement algorithm. The CP-LDA algorithm achieves global convergence, particularly in challenging scenarios involving intricate signal strengths and non-orthogonal CP bases, substantially advancing the algorithm design.

Moreover, the paper establishes sharp bounds for estimating the CP low-rank discriminant tensor. Our analysis accounts for correlated noise tensor entries during perturbation analysis, offering a more realistic framework for practical applications where noise correlations are common. This approach can be applied to other tensor learning analyses. Additionally, we derived minimax optimal misclassification rates, providing the first comprehensive theoretical guarantees in tensor classification.

Despite the significant advancements made by the proposed method, several directions for future research remain: 
(1) Investigating the application of this framework to non-Gaussian tensor data is particularly promising. Previous work by \cite{shao2011sparse} has demonstrated the optimality of LDA rules for elliptical distributions in vector classification, suggesting that our tensor classification framework could perform well in non-Gaussian settings. (2) Additionally, exploring the integration of our CP low-rank approach with other tensor learning tasks, such as regression and clustering, could yield valuable insights.
%
%

\ \\
\spacingset{1.18} 
\bibliographystyle{apalike} 
\bibliography{\mybib}

%
%
\clearpage
\setcounter{page}{1}
\begin{appendices}
    \begin{center}
        {\Large Supplementary Material of ``\TITLE''}

        {Authors}
    \end{center}
    


\section{Proofs for Main Theories} \label{append:proof:main}

\subsection{Proofs for Initialization} \label{append:proof:initia}

\subsubsection*{Proof I: Noise Decomposition}

The initial estimator $\hat\cB$ can be viewed as perturbed version of the ground truth $\cB$ where the noise is actually the estimation error. For notational simplicity, we sometimes give explicit expressions only in the case of $m=1$ and $M=3$. 
We first decompose the perturbing error under $M=3$ without loss of generality. The extension towards higher tensor order $M>3$ is straightforward. 
Let $C_{m,\sigma}=(1-2/n_L)\tr(\otimes_{k\neq m}\Sigma_{k})/d_{-m}, \; C_\sigma = \prod_{m=1}^M C_{m,\sigma}=(1-2/n_L)^M [\tr(\bSigma)/d]^{M-1}=(1-2/n_L)^M[\prod_{m=1}^M \tr(\Sigma_m)/d]^{M-1}$.
Recall that from \eqref{eqn:lda-rule} and \eqref{eqn:lda-discrim-tensor},
\begin{equation*}
\calB = \bbrackets{\calM_{2}-\calM_{1}; \Sigma_1^{-1}, \Sigma_2^{-1}, \Sigma_3^{-1}} 
\quad \text{and} \quad 
\hat\calB = \bbrackets{\bar\calX^{(2)}-\bar\calX^{(1)}; \hat\Sigma_1^{-1}, \hat\Sigma_2^{-1}, \hat\Sigma_3^{-1}}  .
\end{equation*}
Recall
\begin{align*}
\hat C_{\sigma} =\frac{\prod_{m=1}^M \hat\Sigma_{m,11}}{\hat{\Var}(\cX_{1\cdots1})} ,
\end{align*}
where $\hat{\Var}(\cX_{1\cdots1})$ is the pooled sample variance of the first element of $\cX_i^{(k)}$.
Define $\bD_S = \mats(\cM_{2}-\cM_{1})$. Let $\bU_{M+m}:=\bU_{m}$ for all $1\le m\le M$.

\noindent Let $\bDelta_m = \hat\Sigma_m^{-1} - C_{m,\sigma}^{-1}\Sigma_m^{-1}:=\hat\Sigma_m^{-1} - \tilde\Sigma_m^{-1}$ for $m=1,...,M-1$, and
$\bDelta_M = \hat\Sigma_M^{-1} - C_{M,\sigma}^{-1} C_{\sigma} \Sigma_m^{-1}:=\hat\Sigma_M^{-1} - \tilde\Sigma_M^{-1}$. Due to the identifiability issues associated with the tensor normal distribution, we can rescale the covariance matrices $\Sigma_{m}$ such that $C_{m,\sigma}=1$ for $1\le m\le M-1$. In this sense, we slightly abuse notation by using $\Sigma_m$ instead of $\tilde\Sigma_m$.

Then, we have the following decomposition of the error
\begin{align}\label{eqn: error decomposition}
 \cE &= \hat\calB - \calB = \bbrackets{ (\bar\calX^{(2)}-\bar\calX^{(1)}) 
- (\calM_{2}-\calM_{1});\; \Sigma_1^{-1}, \Sigma_2^{-1}, \Sigma_3^{-1}}  \notag\\
&\quad  + \bbrackets{(\bar\calX^{(2)}-\bar\calX^{(1)})- (\calM_{2}-\calM_{1});\; \bDelta_1, \Sigma_2^{-1}, \Sigma_3^{-1}} 
+ \bbrackets{(\bar\calX^{(2)}-\bar\calX^{(1)})- (\calM_{2}-\calM_{1});\; \Sigma_1^{-1}, \bDelta_2, \Sigma_3^{-1}} \notag\\
&\quad  + \bbrackets{(\bar\calX^{(2)}-\bar\calX^{(1)})- (\calM_{2}-\calM_{1});\; \Sigma_1^{-1}, \Sigma_2^{-1}, \bDelta_3} 
+ \bbrackets{(\bar\calX^{(2)}-\bar\calX^{(1)})- (\calM_{2}-\calM_{1});\; \bDelta_1, \bDelta_2, \Sigma_3^{-1}} \notag\\
&\quad  + \bbrackets{(\bar\calX^{(2)}-\bar\calX^{(1)})- (\calM_{2}-\calM_{1});\; \bDelta_1, \Sigma_2^{-1}, \bDelta_3} 
+ \bbrackets{(\bar\calX^{(2)}-\bar\calX^{(1)})- (\calM_{2}-\calM_{1});\; \Sigma_1^{-1}, \bDelta_2, \bDelta_3} \notag\\
&\quad  + \bbrackets{(\bar\calX^{(2)}-\bar\calX^{(1)}) - (\calM_{2}-\calM_{1});\; \bDelta_1, \bDelta_2, \bDelta_3} \notag\\
&\quad + 
\bbrackets{(\calM_{2}-\calM_{1});\; \bDelta_1, \Sigma_2^{-1}, \Sigma_3^{-1}} 
+ 
\bbrackets{ (\calM_{2}-\calM_{1});\; \Sigma_1^{-1}, \bDelta_2, \Sigma_3^{-1}} 
+ 
\bbrackets{ (\calM_{2}-\calM_{1});\; \Sigma_1^{-1}, \Sigma_2^{-1}, \bDelta_3} \notag\\
&\quad + 
\bbrackets{ (\calM_{2}-\calM_{1});\; \bDelta_1, \bDelta_2, \Sigma_3^{-1}} 
+ 
\bbrackets{ (\calM_{2}-\calM_{1});\; \bDelta_1, \Sigma_2^{-1}, \bDelta_3} 
+ 
\bbrackets{ (\calM_{2}-\calM_{1});\; \Sigma_1^{-1}, \bDelta_2, \bDelta_3} \notag\\
&\quad + \bbrackets{(\calM_{2}-\calM_{1});\; \bDelta_1, \bDelta_2, \bDelta_3} \notag\\
&:=\cE_0+\sum_{k=1}^3 \cE_{1,k} +\sum_{k=1}^3 \cE_{2,k} + \cE_3 +\sum_{k=1}^3 \cE_{4,k} +\sum_{k=1}^3 \cE_{5,k} + \cE_6,
\end{align}
where
\begin{align*}
\cE_0 & = \bbrackets{ (\bar\calX^{(2)}-\bar\calX^{(1)}) 
- (\calM_{2}-\calM_{1});\; \Sigma_1^{-1}, \Sigma_2^{-1}, \Sigma_3^{-1}} ,\\
\cE_{1,1} & = \bbrackets{(\bar\calX^{(2)}-\bar\calX^{(1)})- (\calM_{2}-\calM_{1});\; \bDelta_1, \Sigma_2^{-1}, \Sigma_3^{-1}} ,\\
\cE_{2,1} & = \bbrackets{(\bar\calX^{(2)}-\bar\calX^{(1)})- (\calM_{2}-\calM_{1});\; \bDelta_1, \bDelta_2, \Sigma_3^{-1}} ,\\
\cE_3 & = \bbrackets{(\bar\calX^{(2)}-\bar\calX^{(1)})- (\calM_{2}-\calM_{1});\; \bDelta_1, \bDelta_2, \bDelta_3} ,\\
\cE_{4,1} & = \bbrackets{(\calM_{2}-\calM_{1});\; \bDelta_1, \Sigma_2^{-1}, \Sigma_3^{-1}} ,\\
\cE_{5,1} & = \bbrackets{(\calM_{2}-\calM_{1});\; \bDelta_1, \bDelta_2, \Sigma_3^{-1}} ,\\
\cE_6 & = \bbrackets{(\calM_{2}-\calM_{1});\; \bDelta_1, \bDelta_2, \bDelta_3} .\\
\end{align*}


\subsubsection*{Proof II: Initialization by \textsc{c-PCA} without Procedure \ref{alg:initialize-random}}

Taken $\hat{\cB}$ are the perturbed version of the ground truth, we then calculate the initialization through the \textsc{rc-PCA} Algorithm \ref{alg:initialize-cp}. 

In this step, we prove a general results of \textsc{rc-PCA} on any noisy tensor $\hat{\cB} = \cB + \cE$ where $\cB$ assumes a CP low-rank structure $\cB=\sum_{r=1}^R w_r \circ_{m=1}^M \ba_{rm}$ and $\cE$ is the perturbation error tensor, decomposed in \eqref{eqn: error decomposition}.  
CP basis for the $m$-th mode $\{\ba_{rm}\}$, $r\in[R]$ are not necessarily orthogonal. 

We first derive the convergence rate of the CP basis when Procedure \ref{alg:initialize-random} is not applied. Specifically, we consider eigengaps satisfy $\min\{ w_i- w_{i+1}, w_{i}-w_{i-1}\} \ge c w_R$ for all $1\le i\le R$, with $w_0=\infty, w_{R+1}=0$, and $c$ is sufficiently small constant. Let $\cS$ be a subset of $[M]$ that maximizes $\min(d_{\cS}, d_{\cS^C})$ with $d_{\cS}=\prod_{m\in\cS}d_m$ and $d_{\cS^C}=\prod_{m\in\cS^C}d_m$.

Let $\ba_{r,S}=\vect ( \circ_{m \in S}\ba_{rm})$ and $\ba_{r,S^c} =\vect (\circ_{m \in [M]\backslash S}~\ba_{rm})$. Define let $\bA_{S}=(\ba_{1,S},...,\ba_{R,S}) $ and $\bA_{S^c}=(\ba_{1,S^c},...,\ba_{R,S^c}) $.
Let $\bU = (\bu_1,\ldots,\bu_R)$ and $\bV = (\bv_1,\ldots,\bv_R)$ be the orthonormal matrices in Lemma \ref{lemma-transform-ext} with $\bA$ and $\bB$ there replaced respectively by $\bA_S$ and $\bA_{S^c}$. By Lemma \ref{lemma-transform-ext},  
\begin{align}\label{eq1:thm:initial}
\|\ba_{r,S} \ba_{r,S}^\top -  \bu_{r} \bu_{r}^{\top} \|_{2}\vee 
\| \ba_{r,S^c} \ba_{r,S^c}^\top -  \bv_{r}  \bv_{r}^{\top} \|_{2} \le \delta,\quad \big\|{\rm mat}_{S}(\cB) - \bU\Lambda \bV^\top\big\|_{2}\le \sqrt{2}\delta w_1. 
\end{align}
Let $\cE^* = {\rm mat}_{S}(\cE)={\rm mat}_{S}(\hat\cB- \cB)$. We have 
\begin{align*}
\big\|{\rm mat}_{S}(\hat\cB) - \bU\Lambda \bV^\top\big\|_{2}\le \sqrt{2}\delta w_1 + \|\cE^*\|_{2}. 
\end{align*}
As $w_1>w_2>...>w_R> w_{R+1}=0$, Wedin's perturbation theorem \citep{wedin1972perturbation} provides 
\begin{align} 
\max\left\{ \|\widehat \ba_{r,S}\widehat \ba_{r,S}^\top - \bu_{r} \bu_{r}^{\top} \|_{2}, \|\widehat \ba_{r,S^c} \widehat \ba_{r,S^c}^\top - \bv_{r} \bv_{r}^{\top} \|_{2} \right\} 
\le \frac{2 \sqrt{2} w_1 \delta + 2 \|\cE^*\|_{2} }{\min\{ w_{r-1}-w_r, w_r-w_{r+1}\}}. \label{eq3:thm:initial}
\end{align}
Combining \eqref{eq1:thm:initial} and \eqref{eq3:thm:initial}, we have
\begin{align}\label{eq4:thm:initial}
\max\left\{ \|\widehat \ba_{r,S} \widehat \ba_{r,S}^\top - \ba_{r,S} \ba_{r,S}^{\top} \|_{2}, \|\widehat \ba_{r,S^c} \widehat \ba_{r,S^c}^\top - \ba_{r,S^c} \ba_{r,S^c}^{\top} \|_{2} \right\} &\le \delta+\frac{2\sqrt{2} w_1 \delta + 2 \|\cE^*\|_{2} }{\min\{ w_{r-1}-w_r, w_r-w_{r+1}\}}.
\end{align}

\noindent Consider the decomposition of the error term in \eqref{eqn: error decomposition}. The first term $\cE_0$ is a Gaussian tensor. Specifically, from the assumption on the tensor-variate $\cX$, we have $\cE_0 \sim \cT\cN(0; \frac{1}{n}\bSigma^{-1})$, where $\bSigma^{-1}= [\Sigma_m^{-1}]_{m=1}^M$. As $\| \otimes_{m=1}^M \Sigma_m\|_2 \le C_0$, by Lemma \ref{lemma:Gaussian matrix}, in an event $\Omega_{11}$ with probability at least $1-\exp(-c_1 (d_S+d_{S^c}))$, for all $m=1,2,3$, we have
\begin{align}
\left\|\mats(\cE_0) \right\|_2&\le C \frac{\sqrt{d_S}+\sqrt{d_{S^c}}}{\sqrt{n}}.    \label{eq:init_E0a}
\end{align}
For $\cE_{1,k},\cE_{2,k},\cE_3$, the operator norm of their matricization is of a smaller order compared to that of a Gaussian tensor. By Lemma \ref{lemma:precision matrix} and as $\| \otimes_{m=1}^M \Sigma_m^{-1}\|_2 \le C_0$, in an event $\Omega_{12}$ with probability at least $1-n^{-c_1}-\sum_{m=1}^M \exp(-c_1 d_m)$,
\begin{align*}
\bDelta_m &\le   C \sqrt{ \frac{d_m}{n d_{-m}}  }.
\end{align*}
For the bound of $\bDelta_3$, we need to use Lemma \ref{lemma:precision matrix}(ii) with $t_1\asymp t_2\asymp \log(n)$ to derive $|\hat C_{\sigma} - C_{\sigma}|=o(1)$ in the event $\Omega_{12}$.
Since $nd_{-m}\gtrsim d_m$ for all $m$, in the event $\Omega_{12}$, $\bDelta_m\lesssim 1$.
Thus, by Lemma \ref{lemma:Gaussian matrix}, in the event $\Omega_{11}\cap \Omega_{12}$ with probability at least $1-n^{-c_2}-\sum_{m=1}^M \exp(-c_2 d_m)$, we have
\begin{align}
\left\| \mats(\cE_{1,1}) \right\|_2    & \le \left\| \mats\left( \bbrackets{(\bar\cX^{(2)}-\bar\cX^{(1)})- (\cM_{2}-\cM_{1});\; \bI_{d_1}, \Sigma_2^{-1}, \Sigma_3^{-1}} \right) \right\|_2 \cdot \|\bDelta_1\|_2 \notag \\
&\le C_1 \frac{\sqrt{d_S}+\sqrt{d_{S^c}}}{\sqrt{n}} \cdot \sqrt{ \frac{d_1}{n d_{-1}}  }  \notag \\
&\le C \frac{\sqrt{d_S}+\sqrt{d_{S^c}}}{\sqrt{n}} , \label{eq:init_E1a}
\end{align}
for all $m=1,2,3$. Similarly, in the event $\Omega_{11}\cap \Omega_{12}$, $\| \mats(\cE_{1,k}) \|_2$, $\| \mats(\cE_{2,k}) \|_2$ and $\| \mats(\cE_3) \|_2$, $k=1, 2,3$, have the same upper bound in \eqref{eq:init_E1a}.
Recall $\bD_S = \mats(\cM_{2}-\cM_{1})$, by lemma \ref{lemma:precision matrix}, in the event $\Omega_{12}$,
\begin{align}
\left\| \mats(\cE_{4,k} ) \right\|_2    &\le C  \|\bD_S\|_2 \max_m \sqrt{\frac{d_m}{n  d_{-m}}}, \qquad k=1,2,3. \label{eq:init_E2a}
\end{align}
Again, in the event $\Omega_{12}$, $\| \mats(\cE_{5,k}) \|_2$, $\| \mats(\cE_6) \|_2$ have the same upper bound in \eqref{eq:init_E2a}.
It follows that, in the event $\Omega_1=\Omega_{11}\cap \Omega_{12}$ with probability at least $1-n^{-c_2}-\sum_{m=1}^M \exp(-c_2 d_m)$, we have
\begin{align}\label{eq5:thm:initial}
\| \mats(\cE)\|_2 =\|\cE^*\|_2 &\le  C \frac{\sqrt{d_S}+\sqrt{d_{S^c}}}{\sqrt{n}} + C  \|\bD_S\|_2 \max_{1\le m\le M} \sqrt{\frac{d_m}{n  d_{-m}}}  .    
\end{align}

We formulate each $\widehat \bu_r\in\RR^{d_S}$ to be a $|S|$-way tensor $\widehat \bU_r$. Let $\widehat \bU_{rm}={\rm mat}_m(\widehat \bU_r)$, which is viewed as an estimate of $\ba_{rm}\vect(\circ_{\ell \in S\backslash \{m\}}~\ba_{r\ell})^\top\in\RR^{d_m\times (d_S/d_m)}$. Then $\hat\ba_{rm}^{\rm rcpca}$ is the top left singular vector of $\widehat \bU_{rm}$. 
By Lemma \ref{prop-rank-1-approx}, for any $m\in S$
\begin{align*}
\|\hat\ba_{rm}^{\rm rcpca} \hat\ba_{rm}^{{\rm rcpca}\top} - \ba_{rm} \ba_{rm}^\top \|_{2}^2 \wedge (1/2) \le 
\|\widehat \ba_{r,S} \widehat \ba_{r,S}^\top - \ba_{r,S} \ba_{r,S}^{\top} \|_{2}^2 .
\end{align*}
Similar bound can be obtained for $\|\hat\ba_{rm}^{\rm rcpca} \hat\ba_{rm}^{{\rm rcpca}\top} - \ba_{rm} \ba_{rm}^\top \|_{2}$ 
for $m\in S^c$. Substituting \eqref{eq4:thm:initial} and \eqref{eq5:thm:initial} into the above equation, as $ \|\mats( \cM_{2} - \cM_{1} ) \|_{2} \asymp w_1$, we have the desired results.

\subsubsection*{Proof III: Initialization by \textsc{rc-PCA}}

Consider the situation when Procedure \ref{alg:initialize-random} is applied.

\noindent
\textsc{Step I.} Random projections in Procedure \ref{alg:initialize-random} create desirable eigen-gaps between the first and second largest eigen-values, and achieve desired upper bounds.

Let $\ba_{r,S_1}=\vect ( \circ_{m \in S_1}\ba_{rm})$ and $\ba_{r,S_1^c} =\vect (\circ_{m \in S_1} \ba_{rm})$, $\Lambda=\diag(w_1,...,w_R)$ and $w_1\ge w_2\ge\cdots \ge w_R$. Define $\bA_{S_1}=(\ba_{1,S_1},...,\ba_{R,S_1}) $ and $\bA_{S_1^c}=(\ba_{1,S_1^c},...,\ba_{R,S_1^c}) $. Let $\delta_k = \| \bA_k^\top  \bA_k - I_{R}\|_{2}$ and $\delta_{S_1} = \| \bA_{S_1}^\top  \bA_{S_1} - I_{R}\|_{2}$. Note that $\|\bD_S\|_2 = \|\mats( \cM_{2} - \cM_{1} ) \|_{2} \asymp w_1$.

\begin{lemma}\label{lem:rcpca}
Assume $\delta_1 (w_1/w_R)\le c$ for a sufficiently small positive constant $c$.
Apply random projection in Procedure \ref{alg:initialize-random} to the whole sample discriminant tensor $\widehat\cB$ with $L\ge Cd_1^2 \vee Cd_1R^{2(w_1/w_R)^2}$. Denote the estimated CP basis vectors as $\widetilde \ba_{\ell m}$, for $1\le \ell\le L, 1\le m\le M$. Then in an event with probability at least $1-n^{-c}-d_1^{-c}-\sum_{m=2}^M e^{-c d_m}$, we have for any CP basis vectors tuple $(\ba_{rm}, 1\le m\le M)$, there exist $j_r\in[L]$ such that
\begin{align}
\left\| \widetilde \ba_{j_r,m} \widetilde \ba_{j_r,m}^\top  - \ba_{rm} \ba_{rm}^\top \right\|_2  &\le \psi_i,  \qquad 2\le m\le M,\\
\left\| \widetilde \ba_{j_r,1} \widetilde \ba_{j_r,1}^\top  - \ba_{r1} \ba_{r1}^\top \right\|_2  &\le \psi_i + C \sqrt{R-1}\prod_{m=2}^M \delta_m (w_1/w_R) + C \sqrt{R-1}\psi_i^{M-1} (w_1/w_R),
\end{align}
where $1\le r\le R$ and 
\begin{align}\label{eq:psi-i}
\psi_i=C \left( \frac{\sqrt{d_1d_{S_1}}+\sqrt{d_1d_{S_1^c}} }{\sqrt{n}w_i} + \frac{\|\bD_S\|_2}{w_i} \max_{1\le k\le M} \sqrt{\frac{d_k}{n  d_{-k}}} \right) + \frac{2\delta_1 w_1}{w_R}.    
\end{align}
\end{lemma}

\begin{proof}
Define
\begin{align*}
\Xi(\theta)={\rm mat}_{S_1} (\widehat \cB\times_1\theta )=\sum_{r=1}^R w_r (\ba_{r1}^\top \theta) \ba_{r,S_1} \ba_{r,S_1^c}^\top + {\rm mat}_{S_1}(\cE\times_1\theta).    
\end{align*}

First, consider the upper bound of $\| {\rm mat}_{S_1}(\cE\times_1\theta) \|_2$. By concentration inequality for matrix Gaussian sequence (see, for example Theorem 4.1.1 in \cite{tropp2015introduction}) and employing similar arguments in \eqref{eq5:thm:initial} of Proof II, we have, in an event $\Omega_0$ with probability at least $1-n^{-c}-d_1^{-c}-\sum_{m=2}^M e^{-c d_m}$,
\begin{align} \label{eq:init-e}
\| {\rm mat}_{S_1}(\cE\times_1\theta) \|_2 &= \left\|\sum_{i=1}^{d_1} \theta_{i} \cE_{i \cdot\cdot } \right\|_2    \le C \max\left\{ \left\|{\rm mat}_{(12),(3)} (\cE) \right\|_2  ,   \left\|{\rm mat}_{(13),(2)} (\cE) \right\|_2 \right\} \cdot \sqrt{\log(d_1)} \notag\\
&\le C \left( \frac{\sqrt{d_1d_{S_1}}+\sqrt{d_1d_{S_1^c}}}{\sqrt{n}} + \|\bD_S\|_2 \max_{1\le m\le M} \sqrt{\frac{d_m}{n  d_{-m}}} \right) \sqrt{\log(d_1)} .
\end{align}
When $M=3$, $\cE_{i \cdot\cdot }$ represents the $i$-th slice of $\cE$, ${\rm mat}_{(12),(3)}(\cdot)$ denotes the reshaping of order-three tensor into a matrix by collapsing its first and second indices as rows, and the third indices as columns. In the last step, we apply the arguments in \eqref{eq5:thm:initial} of Proof II.

Consider the $i$-th factor and rewrite $\Xi(\theta)$ as follows
\begin{align}
\Xi(\theta)= w_i (\ba_{i1}^\top \theta) \ba_{i,S_1} \ba_{i,S_1^c}^\top + \sum_{r\neq i}^R w_r (\ba_{r1}^\top \theta) \ba_{r,S_1} \ba_{r,S_1^c}^\top + {\rm mat}_{S_1}(\cE\times_1\theta).
\end{align}
Suppose now we repeatedly sample $\theta_{\ell}\sim \theta$, for $\ell=1,...,L$. By the anti-concentration inequality for Gaussian random variables (see Lemma B.1 in \cite{anandkumar2014tensor}), we have
\begin{align}
\PP\left( \max_{1\le \ell \le L} \ba_{i1}^\top \theta_{\ell} \le \sqrt{2 \log(L)} - \frac{\log\log(L)}{4 \sqrt{\log(L)}} - \sqrt{2\log(8)} \right)\le \frac14,    
\end{align}
where $\odot$ denotes Kronecker product. Let
\begin{align*}
\ell_*=\arg\max_{1\le \ell\le L}  \ba_{i1}^\top \theta_{\ell}   .
\end{align*}
Note that $\ba_{i1}^\top \theta_{\ell}$ and $( I_{d_1} - \ba_{i1}\ba_{i1}^\top ) \theta_{\ell}$ are independent. Since the definition of $\ell_*$ depends only on $\ba_{i1}^\top \theta_{\ell}$, this implies that the distribution of $( I_{d_1} - \ba_{i1}\ba_{i1}^\top ) \theta_{\ell}$ does not depend on $\ell_*$.

By Gaussian concentration inequality of $1$-Lipschitz function, we have
\begin{align*}
\PP\left( \max_{r\le R}  \ba_{r1}^\top \big( I_{d_1} - \ba_{i1}\ba_{i1}^\top \big) \theta_{\ell} \ge \sqrt{4 \log(R)} +\sqrt{2 \log(8)} \right) \le \frac14   . 
\end{align*}
Moreover, for the reminder bias term $\ba_{r1}^\top \ba_{i1}\ba_{i1}^\top \theta_{\ell}$, we have,
\begin{align*}
\left\| \sum_{r\neq i} w_r \ba_{r 1}^\top  \ba_{i1}\ba_{i1}^\top \theta_{\ell} \cdot \ba_{r,S_1} \ba_{r,S_1^c}^\top \right\|_2    
&\le \ba_{i1}^\top \theta_{\ell} \cdot \left\| \bA_{S_1} \left( \Lambda \odot \left(\diag (\bA_1^\top \ba_{i1} ) - e_{ii}\right) \right)\bA_{S_1^c}^\top \right\|_2 \\
&\le \ba_{i1}^\top \theta_{\ell}  \|\bA_{S_1} \|_2\|\bA_{S_1^c} \|_2  \|\Lambda\|_2 \left\| \diag (\bA_1^\top \ba_{i1} ) - e_{ii} \right\|_2 \\
&\le (1+\delta_{S_1}\vee \delta_{S_1^c}) \delta_1 w_1 \ba_{i1}^\top \theta_{\ell},
\end{align*}
where $\odot$ denotes Hadamard product, $e_{ii}$ is a $d_1\times d_1$ matrix with the $(i,i)$-th element be 1 and all the others be 0.

Thus, we obtain the top eigengap 
\begin{align}\label{eq:lem_gap}
&w_i (\ba_{i1}^\top \theta) -  \left\| \sum_{r\neq i}^R w_r (\ba_{r1}^\top \theta) \ba_{r,S_1} \ba_{r,S_1^c}^\top \right\|_2   \notag\\
&\ge (1- 2\delta_1 w_1/w_i)\left( \sqrt{2 \log(L)} - \frac{\log\log(L)}{4 \sqrt{\log(L)}} - \sqrt{2\log(8)} \right) w_{i} - \left( \sqrt{4 \log(R)} +\sqrt{2\log(8)} \right) w_{i}(w_1/w_i) \notag\\
&\ge C_0 \sqrt{\log(d_1)} w_{i},
\end{align}
with probability at least $\frac12$, by letting $L\ge Cd_1 \vee CR^{2(w_1/w_R)^2}$.

Since $\theta_{\ell}$ are independent samples, we instead take $L_i=L_{i1}+\cdots+L_{iK}$ for $K=\lceil C_1\log (d_1)/\log(2) \rceil$ and $L_{i1},...,L_{iK}\ge Cd_1 \vee CR^{2(w_1/w_R)^2}$. We define
\begin{align*}
\ell_*^{(k)}=\arg\max_{1\le \ell\le L_{ik}}  \ba_{i1}^\top \theta_{\ell}, \quad   \ell_{*}=\arg\max_{1\le \ell\le L_{i}}  \ba_{i1}^\top \theta_{\ell}.   
\end{align*}
We then have, by independence of $\theta_{\ell}$, that the above statement \eqref{eq:lem_gap} for the $i$-th factor holds in an event $\Omega_i$ with probability at least $1-d_1^{-C_1}$. By Wedin's perturbation theory, we have in the event $\Omega_0\cap\Omega_i$,
\begin{align*}
\left\| \widetilde \bu_{\ell_*} \widetilde \bu_{\ell_*}^\top -  \ba_{i,S_1} \ba_{i,S_1}^\top \right\|_2 \vee \left\| \widetilde \bv_{\ell_*} \widetilde \bv_{\ell_*}^\top -  \ba_{i,S_1^c} \ba_{i,S_1^c}^\top \right\|_2 \le \frac{\| {\rm mat}_{S_1}(\cE\times_1\theta) \|_2}{w_{i}\sqrt{\log(d_1)}} + \frac{2\delta_1 w_1}{w_R},    
\end{align*}
where $\widetilde \bu_{\ell_*}$ and $\widetilde \bv_{\ell_*}$ are the top left and right singular vector of $\Xi(\theta_{\ell_*})$. By Lemma \ref{prop-rank-1-approx} and \eqref{eq:init-e},
\begin{align}\label{eq1:lem_rcpca}
\left\| \widetilde \ba_{\ell_*,m} \widetilde \ba_{\ell_*,m}^\top - \ba_{im} \ba_{im}^\top \right\|_2 \le C \left( \frac{\sqrt{d_1d_{S_1}}+\sqrt{d_1d_{S_1^c}} }{\sqrt{n}w_i} + \frac{\|\bD_S\|_2}{w_i} \max_{1\le k\le M} \sqrt{\frac{d_k}{n  d_{-k}}} \right) + \frac{2\delta_1 w_1}{w_R}, 
\end{align}
for all $2\le m\le M,1\le i\le R.$

Now consider to obtain $\widetilde \ba_{\ell_*,1}$. Write $\psi_i$ to be the error bound on the right hand side of \eqref{eq1:lem_rcpca}. Note that
\begin{align*}
\widehat \cB \times_{m=2}^M \widetilde \ba_{\ell_*,m}  =& \prod_{m=2}^M \left(\widetilde \ba_{\ell_*,m}^\top \ba_{im} \right) w_{i}  \ba_{i1}  +  \sum_{r\neq i} \prod_{m=2}^M \left(\widetilde \ba_{\ell_*,m}^\top \ba_{rm} \right) w_{r}  \ba_{r1}  +\cE  \times_{m=2}^M \widetilde \ba_{\ell_*,m}   .
\end{align*}
By \eqref{eq1:lem_rcpca}, in the event $\Omega_0\cap\Omega_1$,
\begin{align*}
&\left\| \cE  \times_{m=2}^M \widetilde \ba_{\ell_*,m}  \right\|_2    \le \| \cE \|_{\rm op},\\
&\prod_{m=2}^M \left(\widetilde \ba_{\ell_*,m}^\top \ba_{im} \right)  \ge (1-\psi_i^2)^{(M-1)/2}  .
\end{align*}
Since 
\begin{align}\label{eq:a_decomp}
\max_{j\neq i}\big| \ba_{j m}^\top \widetilde \ba_{\ell_*,m}  \big| &=\max_{j\neq i}\big| \widetilde \ba_{\ell_*,m}^\top \ba_{im} \ba_{im}^\top \ba_{j m} + \widetilde \ba_{\ell_*,m}^\top (I - \ba_{im} \ba_{im}^\top ) \ba_{j m}   \big|    \notag\\
&\le \max_{j\neq i}\big| \widetilde \ba_{\ell_*,m}^\top \ba_{im} \big| \big| \ba_{im}^\top \ba_{j m} \big| + \max_{j\neq i}\big\|\widetilde \ba_{\ell_*,m}^\top (I - \ba_{im} \ba_{im}^\top ) \big\|_2 \big\| (I - \ba_{im} \ba_{im}^\top )  \ba_{j m}   \big\|_2 \notag\\
&\le \sqrt{1-\psi_i^2} \delta_m + \psi_i \sqrt{1-\delta_m^2} \le \delta_m +\psi_i,
\end{align}
we have
\begin{align*}
\left\| \sum_{r\neq i} \prod_{m=2}^M \left(\widetilde \ba_{\ell_*,m}^\top \ba_{rm} \right) w_{r}  \ba_{r1}  \right\|_2^2 &\le (R-1)(1+\delta_1)\prod_{m=2}^M(\delta_m+\psi_i)^2 w_1^2  \\
&\le C_M (R-1)  \left(\prod\nolimits_{m=2}^M \delta_m^2 +\psi_i^{2M-2} \right) w_1^2 \\
&\le C_M \left(\prod\nolimits_{m=3}^M \delta_m^2 +\psi_i^{2M-4} \right) w_1^2 .
\end{align*}
By Wedin's perturbation theory,
\begin{align}\label{eq2:lem_rcpca}
\left\| \widetilde \ba_{\ell_*,1} \widetilde \ba_{\ell_*,1}^\top - \ba_{i1} \ba_{i1}^\top \right\|_2 \le \psi_i + C \sqrt{R-1}\prod_{m=2}^M \delta_m (w_1/w_R) + C \sqrt{R-1}\psi_i^{M-1} (w_1/w_R) .    
\end{align}

Repeat the same argument again for all $1\le i\le R$ factors, and let $L=\sum_{i} L_i\ge Cd_1^2 \vee Cd_1 R^{2(w_1/w_R)^2} \ge Cd_1 R\log(d_1) \vee CR^{2(w_1/w_R)^2+1}\log(d_1)$. We have, in the event $\Omega_0\cap\Omega_1\cap\cdots\cap\Omega_R$ with probability at least $1-n^{-c}-d_1^{-c}-\sum_{m=2}^M e^{-c d_m}$, \eqref{eq1:lem_rcpca} and \eqref{eq2:lem_rcpca} hold for all $i$.

\end{proof}

\noindent \textsc{Step II.} {Clustering.} 

For simplicity, consider the most extreme case where $\min\{w_i-w_{i+1},w_{i}-w_{i-1}\} \le c w_R$ for all $i$, with $w_0=\infty, w_{R+1}=0$, and $c$ is sufficiently small constant. In such cases, we need to employ Procedure \ref{alg:initialize-random} to the entire sample discriminant tensor $\widehat \cB$. Let the eigenvalue ratio $w_1/w_R=O(R)$. In general, the statement in the theorem holds for number of initialization $L\ge C d_1^2 \vee C d_1 R^{2(w_1/w_R)^2}$, where $a \vee b =\max\{a, b\}$. 
We prove the statements through induction on factor index $i$ starting from $i=1$ proceeding to $i=R$. By the induction hypothesis, we already have estimators such that 
\begin{align} \label{eq:general-init}
\left\| \widehat \ba_{jm}^{\rm rcpca} \widehat \ba_{jm}^{\rm rcpca \top}  - \ba_{jm} \ba_{jm}^\top \right\|_2  &\le C\phi_0,  \qquad 1\le j\le i-1, 1\le m\le M,
\end{align}
in an event $\Omega$ with high probability, where
\begin{align*}
\phi_0^2 = C\left(\psi_R + \sqrt{R-1}\left(\psi_R \right)^{M-1} \left(\frac{w_1}{w_R} \right) + \max_{2\le m\le M} \delta_m \left(\frac{w_1}{w_R} \right) \right) .  
\end{align*}
Recall $\psi_i$ is defined in \eqref{eq:psi-i}.

Applying Lemma \ref{lem:rcpca}, we obtain that at the $i$-th step ($i$-th factor), we have 
\begin{align*}
\left\| \widetilde \ba_{\ell m} \widetilde \ba_{\ell m}^{\top}  - \ba_{im} \ba_{im}^\top \right\|_2  &\le \phi_0^2,  \qquad  1\le m\le M,
\end{align*}
in the event $\Omega$ with probability at least $1-n^{-c}-d_1^{-c}-\sum_{m=2}^M e^{-cd_m}$ for at least one $\ell\in[L]$. It follows that this estimator $\widetilde \ba_{\ell m}$ satisfies
\begin{align*}
\left\| \widehat\cB \times_{m=1}^{M} \widetilde \ba_{\ell m} \right\|_2   &\ge  \left\| \sum_{j=1}^R w_{j} \prod_{m=1}^M \ba_{j m}^\top \widetilde \ba_{\ell m} \right\|_2    -  \left\| \cE \times_{m=1}^{M} \widetilde \ba_{\ell m} \right\|_2 \\
&\ge \left\|  w_{i} \prod_{m=1}^M \ba_{i m}^\top \widetilde \ba_{\ell m}  \right\|_2  - \left\| \sum_{j\neq i}^R w_{j} \prod_{m=1}^M \ba_{j m}^\top \widetilde \ba_{\ell m}  \right\|_2   -   \left\| \cE \times_{m=1}^{M} \widetilde \ba_{\ell m} \right\|_2 .
\end{align*}
By \eqref{eq:a_decomp} and the last part of the proof of Lemma \ref{lem:rcpca}, as $\left\| \cE \times_{m=1}^{M} \widetilde \ba_{\ell m} \right\|_2/w_i\le \phi_0^2$ and $C\sqrt{R-1}(1+\delta_1)\prod_{m=2}^M (\delta_m+\psi_i) (w_1/w_R) \le \phi_0^2$, it follows that
\begin{align*}
\left\| \widehat\cB \times_{m=1}^{M} \widetilde \ba_{\ell m} \right\|_2   &\ge  (1-\phi_0^4)^{\frac{M}{2}} w_{i} - C\sqrt{R-1}(1+\delta_1)\prod_{m=2}^M (\delta_m+\psi_i) w_1 -\phi_0^2 w_{i }    \\
&\ge (1-3\phi_0^2) w_{i}.
\end{align*}
Now consider the best initialization $\ell_*\in [L]$ by using $\ell_* =\arg\max_{s} |\widehat\cB \times_{m=1}^{M} \widetilde \ba_{s m} |$. By the calculation above, it is immediate that
\begin{align}\label{eq:lbd}
\left\| \widehat\cB \times_{m=1}^{M} \widetilde \ba_{\ell m} \right\|_2   &\ge (1-3\phi_0^2) w_{i}.    
\end{align}
If $\| \ba_{i} \ba_{i}^\top - \widetilde \ba_{\ell_*} \widetilde \ba_{\ell_*}^\top \|_2 \ge C \phi_0$ for a sufficiently large constant $C$, we have that
\begin{align*}
\left\| \widehat\cB \times_{m=1}^{M} \widetilde \ba_{\ell m} \right\|_2   &\le  \left\| \sum_{j=1}^R w_{j} \prod_{m=1}^M \ba_{j m}^\top \widetilde \ba_{\ell m}  \right\|_2   + \left\| \cE \times_{m=1}^{M} \widetilde \ba_{\ell m} \right\|_2 \\
&\le  \left\| \sum_{j=i}^R w_{j} \prod_{m=1}^M \ba_{j m}^\top \widetilde \ba_{\ell m}  \right\|_2 + \phi_0^2 w_{i }+ R \nu^{M} w_1  \\
&\le  (1+\max_m \delta_m)(1-C^2\phi_0^2/2) w_{i} + \phi_0^2 w_{i }+ R \nu^{M} w_1 .
\end{align*}
If $\nu$ satisfies $R\nu^{M} (w_1/w_R) \le c\phi_0^2$ for a small positive constant $c$, as $\max_m \delta_m \le \phi_0^2$, we have
\begin{align*}
\left\| \widehat\cB \times_{m=1}^{M} \widetilde \ba_{\ell m} \right\|_2 \le (1-C' \phi_0^2) w_{i}   ,
\end{align*}
where $C'$ is a sufficiently large constant. It contradicts \eqref{eq:lbd} above. This implies that for $\ell=\ell_*$, we have 
\begin{align*}
\| \ba_i \ba_i^\top - \widetilde \ba_{\ell_*} \widetilde \ba_{\ell_*}^\top \|_2 \le C \phi_0    .
\end{align*}
By Lemma \ref{prop-rank-1-approx}, with $\widehat \ba_{im}^{\rm rcpca} =\widetilde \ba_{\ell_* ,m}$, in the event $\Omega$ with probability at least $1-n^{-c}-d_1^{-c}-\sum_{m=2}^M e^{-c d_M}$,
\begin{align*}
\left\| \widehat \ba_{im}^{\rm rcpca} \widehat \ba_{im}^{\rm rcpca \top}  - \ba_{im} \ba_{im}^\top \right\|_2  &\le C\phi_0, \quad 1\le m\le M   .
\end{align*}
This finishes the proof by an induction argument along with the requirements $R\nu^{M} (w_1/w_R) \le c\phi_0^2$. The general upper bound of \textsc{rc-PCA} is provided in \eqref{eq:general-init}.

\subsection{Proof of Theorem \ref{thm:cp-converge} for Algorithm \ref{alg:tensorlda-cp} DISTIP-CP}
\label{sec:proof-cp-converge}


\noindent\textsc{Step I.} \textbf{Upper bound for $\widehat \ba_{rm}^{(t)}$}. 


Recall that $\bB_m = \bA_m(\bA_m^{\top} \bA_m)^{-1} = (\bb_{1m},\cdots,\bb_{Rm})$ with $\bA_m = (\ba_{1m}, \cdots, \ba_{Rm})$. Let $\bg_{rm} = \bb_{rm}/\|\bb_{rm}\|_2$, $\widehat\bg_{rm}^{(t)} = \widehat\bb_{rm}^{(t)}/\|\widehat\bb_{rm}^{(t)}\|_2$, and  
\begin{align*}
\alpha&=\sqrt{(1-\delta_{\max})(1-1/(4R))}-(R^{1/2}+1)\psi_0.
\end{align*}  
Let $\psi_{0,\ell}=\psi_0$ and define sequentially 
\begin{align}
\phi_{t,m-1} &= (M-1)\alpha^{-1}\sqrt{2R/(1-1/(4R))}\max_{1\le \ell<M}\psi_{t,m-\ell}, \notag \\
\psi_{t,m} &= \Big(2\alpha^{1-M}\sqrt{R-1}(w_1/w_R) \prod\nolimits_{\ell=1}^{M-1}\psi_{t,m-\ell}\Big)
\vee \Big(C\alpha^{1-M}\eta^{\ideal}_{Rm,\phi_{t,m-1}} \Big),   \label{psi_mk2} \\
\eta^{\ideal}_{rm,\phi}&=\frac{\sqrt{d_m}}{\sqrt{n} w_r} +   (\phi \wedge 1) \frac{\sum_{k=1}^M \sqrt{d_k}}{\sqrt{n}  w_r}   +   \frac{\|\cM_{2} - \cM_{1}\|_{\rm op} }{ w_r} \max_{1\le k\le M} \sqrt{\frac{d_k}{n  d_{-k}}} ,  \notag
\end{align} 
for $m=1,\ldots,M$, $t=1,2,\ldots$ 
By induction, \eqref{eq:rho} and $\alpha>0,\rho<1$ give $\psi_{t,m}\le\psi_{t-1,m} \le\psi_0$. 
Here and in the sequel, we take the convention that $(t,\ell)=(t-1,M+\ell)$ with the subscript $(t,\ell)$, and that $\times_\ell \widehat\theta_{j,\ell}^{(t)} = \times_{M+\ell} \widehat\theta_{j,M+\ell}^{(t-1)}$ for any estimator $\widehat\theta_{j\ell}^{(t)}$.
Let 
\begin{align*}
\Omega^*_{t,m-1} =\cap_{\ell=1}^{M-1}\Omega_{t,m-\ell}
\end{align*} 
with $\Omega_{t,\ell}=\big\{\max\nolimits_{r\le R}\| \widehat \ba_{r\ell }^{(t)}\widehat \ba_{r\ell }^{(t)\top}  - \ba_{r\ell } \ba_{r\ell }^\top \|_{2} \le \psi_{t,\ell}\big\}$.
Given $\{\widehat \ba_{r,m-\ell}^{(t)}, r\in[R], \ell\in[M-1] \}$, the $t$-th iteration for tensor mode $m$ produces estimates $\widehat \ba_{rm}^{(t)}$ as the normalized version of $\widehat\cB \times_{\ell=m-1}^{m-M+1} \widehat \bb_{r,\ell}^{(t)\top}$. Because 
$\widehat\cB =\sum_{r=1}^R w_r \circ_{m=1}^{M} \ba_{rm} + \cE$, the ``noiseless'' version of this update is given by
\begin{equation}
\hat\calB\times_{l=1, l\ne m}^{M} \bb_{rl}= w_{r} \ba_{rm} + \cE\times_{l=1, l\ne m}^{M} \bb_{rl}^\top.
\end{equation}
Similarly, for any $1\le r\le R$, 
\begin{align*}
\tilde\cB_{rm}^{(t)}:=\hat\cB \times_{\ell=m-1}^{m-M+1} \widehat \bb_{r,\ell}^{(t)\top} 
= \sum_{i=1}^R \widetilde w_{i,r} \ba_{im}  + \cE \times_{\ell=m-1}^{m-M+1} \widehat \bb_{r,\ell}^{(t)\top} \in \RR^{d_m},
\end{align*}
where 
$\widetilde w_{i,r}= w_i\prod_{\ell=1 }^{M-1} \ba_{i,m-\ell}^\top\widehat \bb_{r,m-\ell}^{(t)}$.  
At $t$-th iteration, $\widehat \ba_{rm}^{(t)}=\tilde\cB_{rm}^{(t)}/\|\tilde\cB_{rm}^{(t)}\|_2$. 

We may assume without loss of generality $\ba_{j\ell}^\top\widehat \ba_{j\ell}^{(t)}\ge 0$ for all $(j,\ell)$. 
Similar to the proofs of Proposition 4 and Theorem 3 in \cite{han2023tensor}, we can show
\begin{align}\label{a-bd}
&\max_{r\le R}\|\widehat \ba_{r\ell}^{(t)} -  \ba_{r\ell} \|_2 \le \psi_{t,\ell}/\sqrt{1-1/(4R)}, \ \ 
\displaystyle \big\|\widehat \bb_{r\ell}^{(t)}\big\|_2 \le \|\widehat \bB_\ell^{(t)}\|_{\rm 2}
\le \bigg(\sqrt{1-\delta_\ell}-\frac{R^{1/2}\psi_0}{\sqrt{1-1/(4R)}}\bigg)^{-1}, \\
&\big\|\widehat \bg_{r\ell}^{(t)} -  \bb_{r\ell}/\| \bb_{r\ell}\|_2\big\|_2
\le (\psi_{t,\ell}/\alpha)\sqrt{2R/(1-1/(4R))}.  \label{b-bd}
\end{align}
Moreover, 
\eqref{a-bd} provides 
\begin{align}\label{g-bd}
\max_{i\neq r}\big| \ba_{i\ell}^\top\widehat \bg_{r\ell}^{(t)}\big| \le \psi_{t,\ell}/\sqrt{1-1/(4R)},\ \
\big| \ba_{r\ell}^\top\widehat \bg_{r\ell}^{(t)} \big| \ge \alpha, 
\end{align}
as $\widehat \ba_{i\ell}^{(t)\top}\widehat \bg_{r\ell}^{(t)}=I\{i=r\}/\|\widehat \bb_{r\ell}^{(t)}\|_2$. 
Then, for $i\neq r$,
\begin{align*}
\widetilde w_{i,r}/\widetilde w_{r,r} &= \big(w_{1}/w_{r}\big) \prod_{\ell=1 }^{M-1} \frac{ \left| (\ba_{i,m-\ell}- \widehat \ba_{i,m-\ell}^{(t)})^\top \widehat\bb_{r,m-\ell}^{(t)} \right| }{  \left| 1+ (\ba_{r,m-\ell}- \widehat \ba_{r,m-\ell}^{(t)})^\top \widehat\bb_{r,m-\ell}^{(t)} \right| }  \\
&\le \big(w_{1}/w_{r}\big)   \prod_{\ell=1 }^{M-1}  \frac{ [\psi_{t,m-\ell}/\sqrt{1-1/(4R)}]/[\sqrt{1-\delta_{\ell}}-R^{1/2}\psi_{t,m-\ell}/\sqrt{1-1/(4R)} ]  }{1- [\psi_{t,m-\ell}/\sqrt{1-1/(4R)}]/[\sqrt{1-\delta_{\ell}}-R^{1/2}\psi_{t,m-\ell}/\sqrt{1-1/(4R)} ] }  \\
&\le \big(w_{1}/w_{r}\big)   \prod_{\ell=1 }^{M-1}  \frac{\psi_{t,m-\ell} }{\alpha} .
\end{align*} 
It follows that
\begin{align}\label{eq:contraction}
\left\| \sum_{i=1}^R \widetilde w_{i,r} \ba_{im} / \widetilde w_{r,r} - \ba_{rm} \right\|_2^2  &= \sum_{i\neq r}^R \sum_{j\neq r}^R (\ba_{im}^\top \ba_{jm})  (\widetilde w_{i,r}/\widetilde w_{r,r})  (\widetilde w_{j,r}/\widetilde w_{r,r})  \\
&\le (R-1)(1+\delta_{m}) \big(w_{1}/w_{r}\big)^2   \prod_{\ell=1 }^{M-1}   \left( \frac{\psi_{t,m-\ell} }{\alpha} \right)^2 .
\end{align}
By basic geometry, we have
\begin{align}
\| \widehat \ba_{rm}^{(t)}\widehat \ba_{rm}^{(t)\top}  - \ba_{rm} \ba_{rm}^\top \|_{\rm 2} &=\|\sin\angle(\widehat\ba_{rm}^{(t)}, \ba_{rm}) \|_2 \le   \| \tilde\cB_{rm}^{(t)}/\widetilde w_{r,r} - \ba_{rm}\|_2  \notag\\
& \le \frac{w_1\sqrt{(R-1)(1+\delta_{m})}}{w_r} \prod_{\ell=1 }^{M-1}  \frac{\psi_{t,m-\ell} }{\alpha}  + \frac{\| \cE \times_{\ell=m-1}^{m-M+1} \widehat \bg_{r,\ell}^{(t)\top} \|_2 }{w_r \prod_{\ell=1 }^{M-1} \ba_{r,m-\ell}^\top\widehat \bg_{r,m-\ell}^{(t)} }  .
\end{align}
in $\Omega^*_{t,m-1}$. As $\cE \times_{\ell=m-1}^{m-M+1} \widehat \bg_{r,\ell}^{(t)\top}$ is linear in each $\widehat \bg_{r\ell}^{(t)}$, 
\begin{align*}
\big\|  \cE \times_{\ell=m-1}^{m-M+1} \widehat \bg_{r,\ell}^{(t)\top} \big\|_{2}
\le& (M-1)\max_{\ell<M} \| \widehat \bg_{r,m-\ell}^{(t)}- \bg_{r,m-\ell} \|_2 \| \Delta \|
+ \big\| \cE \times_{\ell\in [M] \backslash\{m\} } \bg_{r\ell}^{\top} \big\|_{2}, 
\end{align*}
where $\|\Delta \|= \max_{v_\ell\in \mathbb S^{d_\ell-1}\forall\ell}\big( \cE  \times_{\ell=1}^M v_\ell^\top\big)$.  
As we also have $\big\|  \cE \times_{\ell=m-1}^{m-M+1} \widehat \bg_{r,\ell}^{(t)\top} \big\|_{2}\le \|\Delta\|$, 
\eqref{b-bd} and \eqref{g-bd} yield 
\begin{align}\label{eqthm:norm-bd}
\big\|  \cE \times_{\ell=m-1}^{m-M+1} \widehat \bg_{r,\ell}^{(t)\top} \big\|_{2}
\le& \min\big\{\|\Delta\|, \phi_{t,m-1} \| \Delta \| 
+ \big\| \cE \times_{\ell\in [M] \backslash\{m\} } \bg_{r\ell}^{\top} \big\|_{2}
\end{align}
in $\Omega^*_{t,m-1}$, in view of the definition of $\phi_{t,m-1}$ in \eqref{psi_mk2}. 
By the Sudakov-Fernique and Gaussian concentration inequalities, similar to the proof of \eqref{eq5:thm:initial}, we can show
\begin{align*}
& \|\Delta \| \le C \frac{\sum_{k=1}^M \sqrt{d_k}}{\sqrt{n}} + C  \|\cM_{2} - \cM_{1}\|_{\rm op} \max_{1\le k\le M} \sqrt{\frac{d_k}{n  d_{-k}}} , \\
& \big\| \cE \times_{\ell\in [M] \backslash\{m\} } \bg_{r\ell}^{\top} \big\|_{2}  \le 
C \frac{\sqrt{d_m}}{\sqrt{n}} + C  \|\cM_{2} - \cM_{1}\|_{\rm op} \max_{1\le k\le M} \sqrt{\frac{d_k}{n  d_{-k}}} ,
\end{align*}
in an event $\Omega_1$ with at least probability $1-n^{-c}-\sum_{k=1}^M e^{-cd_k}$. 
Consequently, by \eqref{eqthm:norm-bd}, in $\Omega_1\cap\Omega^*_{m,k-1}$, 
\begin{align}\label{eqthm:bdd-ce}
&\frac{\| \cE \times_{\ell=m-1}^{m-M+1} \widehat \bg_{r,\ell}^{(t)\top} \|_2 }{w_r \prod_{\ell=1 }^{M-1} \ba_{r,m-\ell}^\top\widehat \bg_{r,m-\ell}^{(t)} }  \notag\\
\le& \frac{C \alpha^{1-M}\sqrt{d_m}}{\sqrt{n} w_r} +   \frac{C \alpha^{1-M}\|\cM_{2} - \cM_{1}\|_{\rm op} }{ w_r} \max_{1\le k\le M} \sqrt{\frac{d_k}{n  d_{-k}}} +   C \alpha^{1-M} (\phi_{t,m-1}\vee 1) \frac{\sum_{k=1}^M \sqrt{d_k}}{\sqrt{n}  w_r} .
\end{align}
Substituting \eqref{eqthm:bdd-ce} into \eqref{eqthm:norm-bd}, we have, in the event $\Omega_1\cap\Omega_{t,m-1}^*$, 
\begin{align}\label{bdd1:thm-projection0}
\| \widehat \ba_{rm}^{(t)}\widehat \ba_{rm}^{(t)\top}  - \ba_{rm} \ba_{rm}^\top \|_{\rm 2} 
\le \psi_{t,r,m} 
\end{align}
with 
\begin{align*}
\psi_{t,r,m} =\max\bigg\{ C\alpha^{1-M}\eta^{\ideal}_{rm,\phi_{t,m-1}} , 
\frac{w_1\sqrt{2R-2}}{w_r\alpha^{M-1}} \prod_{\ell=1}^{M-1}\psi_{t,m-\ell}\bigg\}. 
\end{align*}
Consequently, $\Omega_{t,m}\subset \Omega_1\cap\Omega^*_{t,m-1}$.
Let $\Omega_0=\{\max_{r\le R,m\le M}\;\|\hat\ba_{rm}^{(0)}\hat\ba_{rm}^{(0)\top} - \ba_{rm}\ba_{rm}^\top\|_2\le \psi_0 \}$
for any initial estimates $\hat\ba_{rm}^{(0)}$. Then \eqref{bdd1:thm-projection0} holds in the event $\Omega_0\cap\Omega_1$.

\medskip
\noindent\textsc{Step II.} \textbf{Number of iterations. } 

We now consider the number of iterations and the convergence of $\psi_{t,m}$ in \eqref{psi_mk2}. A simple way of dealing with the dynamics of \eqref{psi_mk2} is to compare $\psi_{t,r,m}$ with 
\begin{align}\label{psi*}
\psi^*_{t,r,m} &= \Big(2\alpha^{1-M}\sqrt{R-1}(w_1/w_r) \prod\nolimits_{\ell=1}^{M-1}\psi^*_{t,m-\ell}\Big)
\vee \Big(C\alpha^{1-M}\eta^{\ideal}_{rm,1} \Big) \notag \\
\psi^*_{t,m} &= \psi^*_{t,R,m}, 
\end{align}
with initialization $\psi^*_{0,r,m}=\psi_0$. Compared with \eqref{psi_mk2}, \eqref{psi*} is easier to analyze due to 
the use of static 1 in $\eta^{\ideal}_{rm,\phi_{t,m-1}}$ and the monotonicity of $\psi^*_{t,m}$ in $m$. It follows that 
\begin{align}\label{psi-compare}
& \psi_{t,r,m}\le\psi^*_{t,r,m}\le \psi^*_{t,m},\quad\forall (t,r,m). 
\end{align}
As $\rho<1$ with $\rho$ defined in \eqref{eq:rho}, $\psi_{1,1}^*\le \rho \psi_0 \vee \eta^{\ideal}_{R1,1}$ and this would contribute the extra factor $\rho$ in the application \eqref{psi*} to $\psi_{1,2}^*$, resulting in $\psi_{1,2}^*\le \rho^2 \psi_0 \vee \eta^{\ideal}_{R1,1}$, so on and so forth. In general, $\psi_{t,m}^* \le (\rho^{T_{(t-1)M+m}}\psi_0 ) \vee \eta^{\ideal}_{R1,1} $ with $T_1 =1$, $T_2=2, \ldots, T_{M}=2^{M-1}$, and $T_{k+1} = 1+\sum_{\ell=1}^{M-1}T_{k+1-\ell}$ for $k>M$. By induction, for $k=M, M+1,\ldots$. 
\begin{align*}
T_{k+1} \ge \gamma_M^{k-1}+\cdots+\gamma_M^{k-M+1} = \gamma_M^k \frac{1-\gamma_M^{-M+1} }{\gamma_M-1} 
= \gamma_M^k. 
\end{align*}
The function $f(\gamma) = \gamma^M - 2\gamma^{M-1}+1$ is decreasing in $(1,2-2/M)$ and increasing $(2-2/M,\infty)$. 
Because $f(1)=0$ and $f(2)=1>0$, we have $2-2/M < \gamma_M <2$. It follows that the required number of iteration is at most $T=\lceil M^{-1}\{1+ (\log\gamma_M)^{-1} \log \log (\psi_0/\eta^{\ideal}_{R1,1})/\log(1/\rho) \} \rceil$.  Furthermore, the desired upper bound for $\widehat\ba_{rm}$ after convergence is
\begin{align}
\| \widehat \ba_{rm}^{(t)}\widehat \ba_{rm}^{(t)\top}  - \ba_{rm} \ba_{rm}^\top \|_{\rm 2}  \le   C\alpha^{1-M}\eta^{\ideal}_{r1,1}
\end{align}
in the event $\Omega_0\cap\Omega_1$.

\medskip
\noindent\textsc{Step III.} \textbf{Upper bound for $\|\widehat \cB^{\rm cp} -\cB\|_{\rm F}$}. 

After convergence, let $\hat\bb_{rm}=\hat\bb_{rm}^{(t)}$. 
For weights estimation, we have
\begin{align*}
\hat w_{r} &= \hat\calB\times_{m=1}^M \hat  \bb_{rm}^{\top}   = \cE\times_{m=1}^M \hat  \bb_{rm}^{\top} + \cB\times_{m=1}^M \hat  \bb_{rm}^{\top}   \\
&= \cE\times_{m=1}^M \hat  \bb_{rm}^{\top} + w_r \prod_{m=1}^M (\ba_{rm}^\top \hat\bb_{rm}) + \sum_{i\neq r} w_i \prod_{m=1}^M (\ba_{im}^\top \hat\bb_{rm})   .
\end{align*}
As $\sqrt{R}\psi_{t,m}<1$ and $\rho<1$, it follows that
\begin{align*}
\left| \hat w_{r} - w_{r}  \right|   &\le \left| \cE\times_{m=1}^M \hat  \bb_{rm}^{\top} \right| + w_r \left| \prod_{m=1}^M (\ba_{rm}^\top \hat\bb_{rm}) -1 \right| + \left| \sum_{i\neq r} w_i \prod_{m=1}^M (\ba_{im}^\top \hat\bb_{rm}) \right|\\
&\le C \alpha^{-M} w_r \eta^{\ideal}_{r1,1} + w_r \left(1-  \prod_{m=1}^M (1-\alpha^{-1}\psi_{t,m})  \right) + \left| \sum_{i\neq r} w_i \alpha^{-M} \prod_{m=1}^M \psi_{t,m}   \right|   \\
&\le C \alpha^{-M} w_r \eta^{\ideal}_{r1,1} +  \alpha^{-1}\sum_{m=1}^M w_r \psi_{t,r,m} + \alpha^{-M} w_R (R-1) (w_1/w_R) \prod_{m=1}^M \psi_{t,m} \\
&\le C \alpha^{-M} w_r \eta^{\ideal}_{r1,1} +  \alpha^{-1}\sum_{m=1}^M w_r \psi_{t,r,m} + \alpha^{-M} w_R \min_m \psi_{t,m} \\
&\le C_{\alpha} w_r \eta^{\ideal}_{r1,1},
\end{align*}
which is free of $w_r$ by the definition of $\eta^{\ideal}_{r1,1}$.

We may assume without loss of generality $\ba_{r\ell}^\top\widehat \ba_{r\ell}\ge 0$ for all $(r,\ell)$.
Let $\ba_{r}=\vect(\circ_{m=1}^M \ba_{rm})$ and $\hat\ba_{r}=\vect(\circ_{m=1}^M \hat\ba_{rm})$.
Employing similar arguments in the proof of \eqref{eq:contraction}, we have
\begin{align*}
\norm{\hat{\cB}^{\rm cp}  - \cB}_F & = \norm{\sum_{r\in[R]}\hat w_r\circ_{m\in[M]}\hat{\ba}_{rm} - \sum_{r\in[R]} w_r\circ_{m\in[M]}\ba_{rm}}_{\rm F}  \\
&= \norm{\sum_{r\in[R]}\hat w_r \hat \ba_{r} - \sum_{r\in[R]} w_r \ba_{r}}_{\rm 2}  \\ 
& \le \norm{\sum_{r\in[R]}(\hat w_r - w_r)\hat{\ba}_{r}}_2  + \norm{\sum_{r\in[R]} w_r\hat{\ba}_{r} - \sum_{r\in[R]} w_r\ba_{r}}_2 \\
& \le \norm{\sum_{r\in[R]}(\hat w_r - w_r)\hat{\ba}_{r}}_2  + \sqrt{R}\max_{r\le R}\norm{ w_r\hat{\ba}_{r} - w_r\ba_{r}}_2  \\
& \le 2 \sqrt{R}\max_{r\le R} \abs{\hat w_r - w_r }  + \sqrt{R}\max_{r\le R}\norm{ w_r\hat{\ba}_{r} - w_r\ba_{r}}_2  \\
&\le C_{\alpha} \sqrt{R} w_r \eta^{\ideal}_{r1,1}.
\end{align*}
Note that $ \|\cM_{2} - \cM_{1}\|_{\rm op} \asymp w_1$. We can further simplify the bounds.

\subsection{Proof of Theorem \ref{thm:class-upp-bound}}

For simplicity, we mainly focus on the proof for a simple scenario where the prior probabilities $\pi_1 = \pi_2 = 1/2$. Additionally, we will provide key steps of the proof for more general settings correspondingly. Let $\hat \Delta = \sqrt{\langle \hat \calB^{\rm cp} \times_{m=1}^M \Sigma_m, \; \hat \calB^{\rm cp} \rangle}$, the misclassification error of $\hat\Upsilon_{\rm cp}$ is
\begin{align*}
\cR_{\btheta}(\hat\Upsilon_{\rm cp}) &= \frac{n_{L_1}}{n_{L_1} + n_{L_2}} \phi\left(\hat \Delta^{-1}\log(n_{L_2}/n_{L_1}) -\frac{\langle \hat \cM - \cM_1, \; \hat \calB^{\rm cp} \rangle}{\hat \Delta} \right) \\
& + \frac{n_{L_2}}{n_{L_1} + n_{L_2}} \bar \phi\left(\hat \Delta^{-1}\log(n_{L_2}/n_{L_1}) - \frac{\langle \hat \cM - \cM_2, \; \hat \calB^{\rm cp} \rangle}{\hat \Delta} \right)
\end{align*} 
and the optimal misclassification error is
\begin{align*}
\cR_{\rm opt}=\pi_1\phi(\Delta^{-1}\log(\pi_2/\pi_1)-\Delta/2)+\pi_2 \bar \phi(\Delta^{-1}\log(\pi_2/\pi_1)+\Delta/2),    
\end{align*}
where $\phi$ is the CDF of the standard normal, and $\bar \phi(\cdot) = 1 - \phi(\cdot)$.
While the simpler version when $\pi_1 = \pi_2 = \frac{1}{2},$ are
\begin{align*}
\cR_{\btheta}(\hat\Upsilon_{\rm cp}) = \frac{1}{2} \phi\left(-\frac{\langle \hat \cM - \cM_1, \; \hat \calB^{\rm cp} \rangle}{\hat \Delta} \right) + \frac{1}{2} \bar \phi\left(-\frac{\langle \hat \cM - \cM_2, \; \hat \calB^{\rm cp} \rangle}{\hat \Delta} \right)    
\end{align*}
and $\cR_{\rm opt}=\phi(-\Delta/2)=\frac12\phi(-\Delta/2)+\frac12\bar\phi(\Delta/2),$ respectively.
Define an intermediate quantity
\begin{align*}
\cR^{*} = \frac{1}{2} \phi\left(-\frac{\langle \cD, \; \hat \calB^{\rm cp} \rangle}{2 \hat \Delta} \right) + \frac{1}{2} \bar \phi\left(\frac{\langle \cD, \; \hat \calB^{\rm cp} \rangle}{2 \hat \Delta} \right).    
\end{align*}
By Theorem \ref{thm:cp-converge}, in an event $\Omega_1$ with probability at least $\PP(\Omega_0)-n^{-c_1} - \sum_{m=1}^M  e^{-c_1 d_m}$, \begin{align}\label{eq:B_cp}
\|\hat\calB^{\rm cp} - \cB\|_{\rm F} \le C\frac{ \sqrt{\sum_{k=1}^M  d_k R}}{\sqrt{n} }   +   C w_1  \max_{1\le k\le M} \sqrt{\frac{d_kR}{n  d_{-k}}} =o (\Delta ) .   
\end{align}

\noindent Firstly, we are going to show that $R^* -\cR_{\rm opt}(\btheta) \lesssim e^{-\Delta^2/8} \cdot \Delta^{-1} \cdot \|\hat \calB^{\rm cp}  - \calB \|_{\rm F}^{2}$. 
Applying Taylor's expansion to the two terms in $\cR^{*}$ at $-\Delta/2$ and $\Delta/2$, respectively, we obtain 
\begin{equation}
\begin{split}
\label{eqn: taylor 1}
\cR^{*} - \cR_{\rm opt}(\btheta) =& \frac{1}{2}\left(\frac{\Delta}{2} -\frac{\langle \cD, \; \hat \calB^{\rm cp} \rangle}{2\hat \Delta} \right) \phi^{\prime}(\frac{\Delta}{2}) + \frac{1}{2}\left(\frac{\Delta}{2} -\frac{\langle \cD, \; \hat \calB^{\rm cp} \rangle}{2\hat \Delta} \right) \phi^{\prime}(-\frac{\Delta}{2}) \\
& + \frac{1}{2}\left(\frac{\langle \cD, \; \hat \calB^{\rm cp} \rangle}{2\hat \Delta} - \frac{\Delta}{2} \right)^2 \phi^{\prime \prime}(t_{1,n}) + \frac{1}{2}\left(\frac{\langle \cD, \; \hat \calB^{\rm cp} \rangle}{2\hat \Delta} - \frac{\Delta}{2} \right)^2 \phi^{\prime \prime}(t_{2,n})
\end{split}
\end{equation}
where $t_{1,n}, \; t_{2,n}$ are some constants satisfying $| t_{1,n} |, \; | t_{2,n} |$ are between $\frac{\Delta}{2}$ and $\frac{\langle \cD, \; \hat \calB^{\rm cp} \rangle}{2\hat \Delta}$.

Since $\big(\frac{\Delta}{2} -\frac{\langle \cD, \; \hat \calB^{\rm cp} \rangle}{2\hat \Delta} \big)$ frequently appears in \eqref{eqn: taylor 1}, we need to bound its absolute value. Let $\gamma = \calB \times_{m=1}^M \Sigma_m^{1/2}$ and $\hat \gamma = \hat \calB^{\rm cp} \times_{m=1}^M \Sigma_m^{1/2}$, 
then by Lemma \ref{lemma:tensor norm inequality}, in the event $\Omega_1$, we have
\begin{align*}
& \left|  \Delta - \frac{\langle \cD, \; \hat \calB^{\rm cp} \rangle}{\hat \Delta} \right| = \left| \|\gamma\|_{\rm F} -  \frac{\langle \gamma ,\; \hat \gamma \rangle}{\|\hat \gamma\|_{\rm F}} \right| = \left| \frac{\|\gamma\|_{\rm F} \cdot \|\hat\gamma\|_{\rm F} - \langle \gamma ,\; \hat \gamma \rangle}{\|\hat\gamma\|_2}  \right| \\
\lesssim& \frac{1}{\Delta}\|\hat \gamma - \gamma\|_{\rm F}^2 \lesssim \frac{1}{\Delta}\|\hat \calB^{\rm cp} - \calB\|_{\rm F}^2.
\end{align*}
In fact, by triangle inequality,
\begin{align*}
| \hat \Delta - \Delta | &= \left\|\hat \calB^{\rm cp}  \times_{m=1}^M \Sigma_m^{1/2} \right\|_{\rm F}- \left\| \calB \times_{m=1}^M \Sigma_m^{1/2} \right\|_{\rm F} \le \left\|\left(\hat \calB^{\rm cp} - \calB\right) \times_{m=1}^M \Sigma_m^{1/2} \right\|_{\rm F} \le \left\| \hat \calB^{\rm cp} - \calB \right\|_{\rm F} \prod_{m=1}^M \left\|\Sigma_m\right\|_{2}^{1/2} \\
& \lesssim \left\|\hat \calB^{\rm cp} - \calB\right\|_{\rm F} \lesssim \frac{ \sqrt{\sum_{k=1}^M d_k R}}{\sqrt{n} }   +   w_1  \max_{1\le k\le M} \sqrt{\frac{d_kR}{n  d_{-k}}} = o(\Delta).
\end{align*}
Since $\|\hat \calB^{\rm cp} - \calB\|_{\rm F} = o(\Delta)$, it follows that $\langle \cD, \; \hat \calB^{\rm cp} \rangle/(2\hat \Delta) \rightarrow \Delta/2$.
Then, we have $| \phi^{\prime \prime}(t_{1,n}) | \asymp | \phi^{\prime \prime}(t_{2,n}) | \asymp \Delta  e^{-\frac{(\Delta/2)^2}{2}} = \Delta  e^{-\Delta^2/8}$.
Hence,
\begin{align*}
&\frac{1}{2}\Big(\frac{\langle \cD, \; \hat \calB^{\rm cp} \rangle}{2\hat \Delta} - \frac{\Delta}{2} \Big)^2 \phi^{\prime \prime}(t_{1,n}) + \frac{1}{2}\Big(\frac{\langle \cD, \; \hat \calB^{\rm cp} \rangle}{2\hat \Delta} - \frac{\Delta}{2} \Big)^2 \phi^{\prime \prime}(t_{2,n}) \\ 
&\asymp  \frac{1}{\Delta^2} \left\|\hat \calB^{\rm cp} - \calB\right\|_{\rm F}^4 \cdot \Delta \cdot e^{-\Delta^2/8}  \asymp  \frac{1}{\Delta} e^{-\Delta^2/8} \left\|\hat \calB^{\rm cp} - \calB\right\|_{\rm F}^4   .
\end{align*}
Then \eqref{eqn: taylor 1} can be further bounded such that
\begin{align*}
\cR^{*} - \cR_{\rm opt}(\btheta) &\asymp \Big(\frac{\Delta}{2} - \frac{\langle \cD, \; \hat \calB^{\rm cp} \rangle}{2\hat \Delta} \Big) e^{-\frac{(\Delta/2)^2}{2}} + O\Big(\frac{1}{\Delta}  e^{-\Delta^2/8} \left\|\hat \calB^{\rm cp} - \calB\right\|_{\rm F}^4 \Big)\\
    & \le e^{-\Delta^2/8} \cdot \left| \frac{\Delta}{2} - \frac{\langle \cD, \; \hat \calB^{\rm cp} \rangle}{2\hat \Delta} \right| +  O\left( \frac{1}{\Delta}  e^{-\Delta^2/8} \left\|\hat \calB^{\rm cp} - \calB\right\|_{\rm F}^4 \right) \\
    & \lesssim \frac{1}{\Delta} e^{-\Delta^2/8}  \left\|\hat \calB^{\rm cp} - \calB\right\|_{\rm F}^2  +    \frac{1}{\Delta}  e^{-\Delta^2/8} \left\|\hat \calB^{\rm cp} - \calB\right\|_{\rm F}^4  .
\end{align*}
Eventually we obtain $\cR^{*} - \cR_{\rm opt}(\btheta) \lesssim \Delta^{-1} e^{-\Delta^2/8} (\|\hat \calB^{\rm cp} - \calB \|_{\rm F}^2 \vee \|\hat \calB^{\rm cp} - \calB \|_{\rm F}^4)$ in the event $\Omega_1$ with probability at least $\PP(\Omega_0)-n^{-c_1} - \sum_{m=1}^M  e^{-c_1 d_m}$.

\noindent Next, focus on $\cR_{\btheta}(\hat\Upsilon_{\rm cp}) - \cR^{*}$. We apply Taylor's expansion to $\cR_{\btheta}(\hat\Upsilon_{\rm cp})$:
\begin{align}
\label{eqn: taylor 2}
   \cR_{\btheta}(\hat\Upsilon_{\rm cp}) &= \frac{1}{2} \left\{ \phi\Big(-\frac{\langle \cD, \; \hat \calB^{\rm cp} \rangle}{2\hat \Delta} \Big) + \frac{\langle \cD, \; \hat \calB^{\rm cp} \rangle/2 - \langle \hat \cM - \cM_1, \; \hat \calB^{\rm cp} \rangle}{\hat \Delta}\phi^{\prime} \Big(\frac{\langle \cD, \; \hat \calB^{\rm cp} \rangle}{2\hat \Delta} \Big)  \right.\notag\\
   & \left.+ O\left( \Delta \cdot e^{-\Delta^2/8} \right) \Big( \frac{\langle \hat \cM - \cM_1, \; \hat \calB^{\rm cp} \rangle - \langle \cD, \; \hat \calB^{\rm cp} \rangle/2}{\hat \Delta} \Big)^2 \;  \right\} \notag\\
   &+ \frac{1}{2} \left\{ \bar \phi \Big(\frac{\langle \cD, \; \hat \calB^{\rm cp} \rangle}{2\hat \Delta} \Big) + \frac{\langle \cD, \; \hat \calB^{\rm cp} \rangle/2 + \langle \hat \cM - \cM_2, \; \hat \calB^{\rm cp} \rangle}{\hat \Delta}\phi^{\prime}\Big(\frac{\langle \cD, \; \hat \calB^{\rm cp} \rangle}{2\hat \Delta} \Big) \right.\notag\\ 
   & \left.+ O\big( \Delta \cdot e^{-\Delta^2/8} \big) \Big( \frac{\langle \hat \cM - \cM_2, \; \hat \calB^{\rm cp} \rangle + \langle \cD, \; \hat \calB^{\rm cp} \rangle/2}{\hat \Delta} \Big)^2 \;  \right\} 
\end{align}
where the remaining term can be obtained similarly as \eqref{eqn: taylor 1} by using the fact that $|\phi^{\prime \prime}(t_n)| = O(\Delta \cdot e^{-\Delta^2/8})$. Now we aim to bound the following term:
\begin{align*} 
 \left| \frac{\langle \hat \cM - \cM_1, \; \hat \calB^{\rm cp} \rangle - \langle \cD, \; \hat \calB^{\rm cp} \rangle/2}{\hat \Delta} \right| 
\lesssim & \frac{1}{\Delta} \left| \langle \bar\calX^{(2)} - \cM_2 + \bar\calX^{(1)} - \cM_1, \; \hat \calB^{\rm cp} \rangle \right| \notag\\
\lesssim & \frac{1}{\Delta} \left| \langle  \bar\calX^{(1)} - \cM_1, \; \hat \calB^{\rm cp} \rangle \right|  + \frac{1}{\Delta} \left| \langle \bar\calX^{(2)} - \cM_2, \; \hat \calB^{\rm cp} \rangle \right|.  
\end{align*}
Note that, in the event $\Omega_1$, $\|\hat \calB^{\rm cp}\|_{\rm F} \le \|\hat\calB^{\rm cp} - \cB\|_{\rm F} + \|\cB\|_{\rm F} \lesssim \Delta $. By Lemma \ref{lemma:low-rank-tensor}, in an event $\Omega_2$ with probability at least $1-e^{-c_2\sum_{m=1}^M d_m R}$,
\begin{align*}
\left| \langle  \bar\calX^{(k)} - \cM_k, \; \hat \calB^{\rm cp} \rangle \right| &\lesssim \Delta \sqrt{\frac{\sum_{m=1}^M d_m R}{n}} ,\quad k=1,2. 
\end{align*}
It follows that, in the event $\Omega_1\cap \Omega_2$, 
\begin{align}
\left| \frac{\langle \hat \cM - \cM_1, \; \hat \calB^{\rm cp} \rangle - \langle \cD, \; \hat \calB^{\rm cp} \rangle/2}{\hat \Delta} \right| &\lesssim \Delta \sqrt{\frac{\sum_{m=1}^M d_m R}{n}},    \label{eqn: upper bound of taylor terms 2}\\
\left| \frac{\langle \hat \cM - \cM_2, \; \hat \calB^{\rm cp} \rangle + \langle \cD, \; \hat \calB^{\rm cp} \rangle/2}{\hat \Delta} \right| &\lesssim \Delta \sqrt{\frac{\sum_{m=1}^M d_m R}{n}} .   \label{eqn: upper bound of taylor terms 3}
\end{align}
Substituting \eqref{eqn: upper bound of taylor terms 2} and \eqref{eqn: upper bound of taylor terms 3} into \eqref{eqn: taylor 2}, we obtain,
\begin{align*}
\left| \cR_{\btheta}(\hat\Upsilon_{\rm cp}) - \cR^{*} \right| \lesssim& \left| \frac{\langle \cD, \; \hat \calB^{\rm cp} \rangle/2 - \langle \hat \cM - \cM_1, \; \hat \calB^{\rm cp} \rangle}{\hat \Delta} \phi^{\prime}(\frac{\langle \cD, \; \hat \calB^{\rm cp} \rangle}{2\hat \Delta})  \right.\\
& \left. + \frac{\langle \cD, \; \hat \calB^{\rm cp} \rangle/2 + \langle \hat \cM - \cM_2, \; \hat \calB^{\rm cp} \rangle}{\hat \Delta} \phi^{\prime}(\frac{\langle \cD, \; \hat \calB^{\rm cp} \rangle}{2\hat \Delta}) + O \left(\Delta^3 e^{-\Delta^2/8} \left(\frac{\sum_{m=1}^M d_m R}{n} \right) \right) \right|
\end{align*}
Since $\cD/2- (\hat \cM-\cM_1) + \cD/2 + (\hat \cM-\cM_2) = \cD - (\cM_2-\cM_1) = 0$, then it follows that
\begin{align*}
\left| \cR_{\btheta}(\hat\Upsilon_{\rm cp}) - \cR^{*} \right| \lesssim \Delta^3 e^{-\Delta^2/8} \left(\frac{\sum_{m=1}^M d_m R}{n} \right)  .  
\end{align*}

Finally, combining the two pieces, we obtain
\begin{align*}
\cR_{\btheta}(\hat\Upsilon_{\rm cp}) -\cR_{\rm opt}(\btheta) \le& \cR_{\btheta}(\hat\Upsilon_{\rm cp}) -  \cR^{*} + \cR^{*} - \cR_{\rm opt}(\btheta) \\
\lesssim & \frac{1}{\Delta} e^{-\Delta^2/8}  \left\|\hat \calB^{\rm cp} - \calB\right\|_{\rm F}^2  +    \frac{1}{\Delta}  e^{-\Delta^2/8} \left\|\hat \calB^{\rm cp} - \calB\right\|_{\rm F}^4 + \Delta^3 e^{-\Delta^2/8} \left(\frac{\sum_{m=1}^M d_m R}{n} \right) ,
\end{align*}
in the event $\Omega_1\cap \Omega_2$ with probability at least $\PP(\Omega_0)-n^{-c}-\sum_{m=1}^M e^{-cd_m}$.

Now consider the two case. On the one hand, when $\Delta = O(1)$, by \eqref{eq:B_cp}, with probability at least $\PP(\Omega_0)-n^{-c}-\sum_{m=1}^M e^{-cd_m }$, we have 
\begin{align*}
\cR_{\btheta}(\hat\Upsilon_{\rm cp}) - \cR_{\rm opt}(\btheta) \le C \frac{\sum_{m=1}^M d_m R}{n}  +     C \frac{ w_1^2 R }{\Delta^2 }  \max_{1\le m\le M} \frac{d_m}{n  d_{-m}}  .
\end{align*}
On the other hand, when $\Delta\to\infty$ as $n\to \infty$, by\eqref{eq:B_cp}, with probability at least $1-n^{-c}-\sum_{m=1}^M e^{-cd_m }$, we have 
\begin{align*}
& \cR_{\btheta}(\hat\Upsilon_{\rm cp}) - \cR_{\rm opt}(\btheta) \\
\le& C \Delta^3 e^{-\Delta^2/8} \left(\frac{\sum_{m=1}^M d_m R }{n} \right) +  C \Delta e^{-\Delta^2/8} \frac{w_1^2 R \max_{1\le m\le M}d_m^2}{\Delta^2 nd  } + C \Delta^3 e^{-\Delta^2/8} \left(\frac{w_1^2 R \max_{1\le m\le M}d_m^2}{\Delta^2 nd} \right)^2 \\
\le& C \Delta^3 e^{-\Delta^2/8} \left(\frac{\sum_{m=1}^M d_m R }{n}  +\frac{w_1^2 R \max_{1\le m\le M}d_m^2}{\Delta^2 nd} \right)\\
=& C \exp\left(-\left(\frac18-\frac{3\log(\Delta)}{\Delta^2} \right)\Delta^2\right) \left(\frac{\sum_{m=1}^M d_m R }{n}  +\frac{w_1^2 R \max_{1\le m\le M}d_m^2}{\Delta^2 nd} \right),
\end{align*}
where $3\log(\Delta)/\Delta^2$ is an $o(1)$ term as $n\to\infty$.

As $\Delta^2\asymp w_1^2+\cdots +w_R^2\gg w_1^2 (d_{\max}/d)$, we have
\begin{align*}
\frac{w_1^2 R \max_{1\le m\le M}d_m^2}{\Delta^2 nd} \ll   \frac{\sum_{m=1}^M d_m R }{n}.  
\end{align*}
That is, the second part in the excess misclassification rate, which comes from the estimation accuracy of the mode-$m$ precision matrix is negligible.

\subsection{Proof of Theorem \ref{thm:class-lower-bound}}

Note that the proof is not straightforward, partly because the excess risk $\cR_{\btheta}(\hat\Upsilon_{\rm cp}) -\cR_{\rm opt}(\btheta)$ does not satisfy the triangle inequality required by standard lower bound techniques. A crucial approach in this context is establishing a connection to an alternative risk function.
For a general classification rule $\Upsilon$, we define $L_{\btheta}(\Upsilon)=\PP_{\btheta}(\Upsilon(\cZ) \neq \Upsilon_{\theta}(\cZ))$, where $\Upsilon_{\theta}(\cZ)$ is the Fisher’s linear discriminant rule introduced in (\ref{eqn:lda-rule}). Lemma \ref{lemma:the first reduction} 
allows us to transform the excess risk $\cR_{\btheta}(\hat\Upsilon_{\rm cp}) -\cR_{\rm opt}(\btheta)$ into the risk function $L_{\btheta}(\hat \Upsilon_{\rm cp})$, as shown below:
\begin{equation} \label{eqn:loss function reduction}
\cR_{\btheta}(\hat\Upsilon_{\rm cp}) -\cR_{\rm opt}(\btheta) \ge \frac{\sqrt{2\pi}\Delta}{8} e^{\Delta^2/8} \cdot L_{\theta}^2(\hat \Upsilon_{\rm cp}).    
\end{equation}
We then apply Lemma \ref{lemma:Tsybakov variant} to derive the minimax lower bound for the risk function $L_{\btheta}(\hat \Upsilon_{\rm cp})$.

We carefully construct a finite collection of subsets of the parameter space $\calH$ that characterizes the hardness of the problem. 
Any $M$-th order tensor $\cM \in \mathbb{R}^{d_1 \times \cdots \times d_M}$ with CP rank $R$ can be expressed as $\cM= \cF \times_{m=1}^M \bA_m$. Here, the latent core tensor $\cF=\diag(w_1,...,w_R)$ is a diagonal tensor of dimensions $R \times \cdots \times R$, i.e. $\cF_{i,...,i}$ is non zero for all $i=1,..,R$, and all the other elements of $\cF$ are zero.
Denote the loading matrices $\bA_m=(\ba_{1m},...,\ba_{Rm}) \in \mathbb{R}^{d_m \times R}$ for each mode $m=1,\ldots,M$. 

First, let $\bA_m$ be a fixed matrix where the $(i,i)$-th elements, $i=1,...,R$, are set to one and all other elements are zero. Denote 
$\bA=\bA_M\otimes \bA_{M-1}\otimes \cdots\otimes \bA_1$. According to basic tensor algebra, this setup implies that 
$\vect(\cM) = \bA \vect(\cF)$ and $\|\cM\|_{\rm F} = \|\vect(\cF)\|_2$. Let $\be_1$ be the basis vector in the standard Euclidean space whose first entry is 1 and 0 elsewhere, and $\bI_{d}=[\bI_{d_m}]_{m=1}^M$.
Define the following parameter space
\begin{align*}
\cH_0 =& \big\{ \theta = (\cM_1, \; \cM_2, \; \bI_{d}) : \; \cM_1 = \cF \times_{m=1}^M \bA_m, \; \cM_2= -\cM_1; 
\ \vect(\cF)=\epsilon \bff+ \lambda \be_1, \bff \in \{ 0,1\}^R, \bff^\top \be_1=0 \big\},
\end{align*}
where $\epsilon=c/\sqrt{n}$, $c=O(1)$ and $\lambda$ is chosen to ensure that $\theta\in\cH$ such that
\begin{align*}
\Delta=(\cM_2-\cM_1)^\top\bSigma^{-1} (\cM_2 -\cM_1) = 4 \| \epsilon \bff+ \lambda \be_1\|_2^2 = 4 \epsilon^2 \| \bff\|_2^2 + 4 \lambda^2. 
\end{align*}
In addition to $\cH_0$, we also define $\bA_{\ell}$, for $\ell\neq m$, as fixed matrices where the $(i,i)$-th elements, $i=1,...,R$, are set to one and all other elements are zero. Let $\cF$ be a diagonal tensor such that the $(i,...,i)$-th elements, $i=1,..., R$, are set to one, and all other elements are zero. It implies that $\|\cM\|_{\rm F} = \| \bA_m\|_{\rm F}$. For $m=1,...,M$, define the following parameter spaces 
\begin{align*}
\cH_m = \big\{& \theta = (\cM_1, \; \cM_2, \; \bI_{d}) : \; \cM_1 = \cF \times_{k=1}^M \bA_k, \; \cM_2= -\cM_1; 
\ \vect(\bA_m)=\epsilon \bg_m+ \lambda_m \be_1, \bg_m \in \{ 0,1\}^{d_m R},\\
&\bg_m^\top \be_1=0 \big\},
\end{align*}
where $\epsilon=c/\sqrt{n}$, $c=O(1)$ and $\lambda_m$ is chosen to ensure that $\theta\in\cH$ such that
\begin{align*}
\Delta=(\cM_2-\cM_1)^\top\bSigma^{-1} (\cM_2 -\cM_1) = 4 \| \epsilon \bg_m+ \lambda_m \be_1\|_2^2 = 4 \epsilon^2 \| \bg_m\|_2^2 + 4 \lambda_m^2. 
\end{align*}
It is clear that $\cap_{\ell=0}^M \cH_{\ell} \subset \cH$. We shall show below separately for the minimax risks over each parameter space $\cH_{\ell}$.

First, consider $\cH_0$. By Lemma \ref{lemma:Varshamov-Gilbert Bound}, we can construct a sequence of $R$-dimensional vectors $\bff_1, \ldots ,\bff_N \in \{0,1\}^R$, such that $\bff_{i}^\top \be_1=0$, $\rho_H(\bff_i, \bff_j) \ge R/8, \; \forall 0 \le i < j \le N$, and $R \le (8/\log 2)\log N$, where $\rho_H$ denotes the Hamming distance.
To apply Lemma \ref{lemma:Tsybakov variant}, for $\forall \theta_{\bu}, \; \theta_{\bv} \in \cH_0, \; \theta_{\bu} \neq \theta_{\bv},$ we need to verify two conditions:  
\begin{enumerate}
\item[(i)] the upper bound on the Kullback-Leibler divergence between $\PP_{\theta_{\bu}}$ and $\PP_{\theta_{\bv}}$, and
\item[(ii)] the lower bound of $L_{\theta_{\bu}}(\hat \Upsilon_{\rm cp}) + L_{\theta_{\bv}}(\hat \Upsilon_{\rm cp})$ for $\bu \neq \bv$ and $\bu^\top \be_1=0, \bv^\top \be_1=0$.
\end{enumerate}

We calculate the Kullback-Leibler divergence first. For $\bff_{\bu} \in\{ 0,1\}^R$ and $\bff_{\bu}^\top \be_1=0$, define
\begin{align*}
\cM_{\bu} = \cF_{\bu} \times_{m=1}^M \bA_m, \; \vect(\cF_{\bu})= \epsilon \bff_{\bu} + \lambda \be_1, \; \theta_{\bu} = \big(\cM_{\bu}, \; -\cM_{\bu}, \; \bI_d \big) \in \cH_0 .   
\end{align*}
and consider the distribution $\cT\cN(\cM_{\bu}, \; \bI_d)$. 
Then the Kullback-Leibler divergence between $\PP_{\theta_{\bu}}$ and $\PP_{\theta_{\bv}}$ can be bounded by
\begin{align*}
    {\rm KL}(\PP_{\theta_{\bu}}, \PP_{\theta_{\bv}}) &= \frac{1}{2} \norm{\vect(\cM_{\bu}) - \vect(\cM_{\bv})}_2^2 = \frac{1}{2}\norm{\vect(\cF_{\bu}) - \vect(\cF_{\bv})}_2^2 \le \frac{c^2 R}{2n}  .
\end{align*}
In addition, by applying Lemma \ref{lemma:probability inequality}, we have that for any $\bff_{\bu}, \bff_{\bv} \in\{ 0,1\}^R$,
\begin{align*}
    L_{\theta_{\bu}}(\hat \Upsilon_{\rm cp}) + L_{\theta_{\bv}}(\hat \Upsilon_{\rm cp}) &\ge \frac{1}{\Delta} e^{-\Delta^2/8} \cdot \norm{\vect(\cF_{\bu}) - \vect(\cF_{\bv})}_2 \\
    &\ge \frac{1}{\Delta} e^{-\Delta^2/8} \sqrt{\frac{R}{8} \cdot \frac{c^2}{n}  } \\
    &\gtrsim \frac{1}{\Delta} e^{-\Delta^2/8} \sqrt{\frac{R}{n}} .
\end{align*}
So far we have verified the aforementioned conditions (i) and (ii). Lemma \ref{lemma:Tsybakov variant} immediately implies that, there exists some constant $C_{\gamma} > 0$, such that 
\begin{equation}\label{eqn:inter1}
\inf_{\hat \Upsilon_{\rm cp}} \sup_{\theta \in \calH_0} \PP\left(L_{\theta}(\hat \Upsilon_{\rm cp}) \ge C_{\gamma} \frac{1}{\Delta} e^{-\Delta^2/8} \sqrt{ \frac{R}{n} } \right) \ge 1-\gamma
\end{equation}
Combining \eqref{eqn:inter1} and \eqref{eqn:loss function reduction}, we have
\begin{equation}\label{eqn:inter2}
\inf_{\hat \Upsilon_{\rm cp}} \sup_{\theta \in \calH_0} \PP\left(\cR_{\btheta}(\hat\Upsilon_{\rm cp}) -\cR_{\rm opt}(\btheta) \ge C_{\gamma} \frac{1}{\Delta} e^{-\Delta^2/8} \cdot \frac{R}{n}  \right) \ge 1-\gamma   .
\end{equation}

Similarly, for each $\cH_m$, $m=1,...,M$, we can obtain that,  there exists some constant $C_{\gamma} > 0$, such that 
\begin{equation}\label{eqn:inter3}
\inf_{\hat \Upsilon_{\rm cp}} \sup_{\theta \in \calH_m} \PP\left(\cR_{\btheta}(\hat\Upsilon_{\rm cp}) -\cR_{\rm opt}(\btheta) \ge C_{\gamma} \frac{1}{\Delta} e^{-\Delta^2/8} \cdot \frac{d_m R }{n}  \right) \ge 1-\gamma   .
\end{equation}

Finally combining \eqref{eqn:inter2} and \eqref{eqn:inter3}, we obtain the desired lower bound for the excess misclassficiation error
\begin{equation}\label{eqn:mis}
\inf_{\hat \Upsilon_{\rm cp}} \sup_{\theta \in \calH} \PP\left(\cR_{\btheta}(\hat\Upsilon_{\rm cp}) -\cR_{\rm opt}(\btheta) \ge C_{\gamma} \frac{1}{\Delta} e^{-\Delta^2/8} \cdot \frac{\sum_{m=1}^M d_m R + R}{n}  \right) \ge 1-\gamma   .
\end{equation}

This implies that if $c_1<\Delta \le c_2$ for some $c_1,c_2>0$, we have 
\begin{align*}
\inf_{\hat \Upsilon_{\rm cp}} \sup_{\theta \in \calH} \PP\left(\cR_{\btheta}(\hat\Upsilon_{\rm cp}) -\cR_{\rm opt}(\btheta) \ge C_{\gamma} \cdot \frac{\sum_{m=1}^M d_m R }{n}  \right) \ge 1-\gamma   .    
\end{align*}
On the other hand, if $\Delta\to\infty$ as $n\to \infty$, then for any $\vartheta>0$
\begin{align*}
\inf_{\hat \Upsilon_{\rm cp}} \sup_{\theta \in \calH} \PP\left(\cR_{\btheta}(\hat\Upsilon_{\rm cp}) -\cR_{\rm opt}(\btheta) \ge C_{\gamma} \exp\left\{-\left(\frac18+\vartheta\right)\Delta^2 \right\} \frac{\sum_{m=1}^M d_m R }{n}  \right) \ge 1-\gamma   .    
\end{align*}


\section{Technical Lemmas}
\label{append:proof:lemma}

We collect all technical lemmas that has been used in the theoretical proofs throughout the paper in this section. Let $d=d_1d_2\cdots d_M$ and $d_{-m}=d/d_m.$ Denote $\|A\|_2$ or $\|A\|$ as the spectral norm of a matrix $A$. Also let $\otimes$ be the Kronecker product, and $\circ$ be the tensor outer product.

\begin{lemma}\label{lemma-transform-ext} 
Let  $\bA \in \RR^{d_1\times r}$ and $\bB \in \RR^{d_2\times r}$ with 
$\|\bA^\top \bA - \bI_r\|_{\rm 2}\vee \|\bB^\top \bB - \bI_r\|_{\rm 2} \le\delta$ and $d_1\wedge d_2\ge r$. 
Let $\bA=\widetilde \bU_1 \widetilde \bD_1 \widetilde \bU_2^\top$ be the SVD of $\bA$, 
$\bU = \widetilde \bU_1\widetilde \bU_2^\top$, $\bB=\widetilde \bV_1 \widetilde \bD_2 \widetilde \bV_2^\top$ 
the SVD of $\bB$, and $\bV = \widetilde \bV_1\widetilde \bV_2^\top$. 
Then, $\|\bA \Lambda \bA^\top - \bU \Lambda \bU^{\top}\|_{\rm 2}\le \delta \|\Lambda\|_{\rm 2}$  
for all nonnegative-definite matrices $\Lambda$ in $\RR^{r\times r}$, and 
$\|\bA \bQ \bB^\top - \bU \bQ \bV^{\top}\|_{\rm 2}\le \sqrt{2}\delta \|\bQ\|_{\rm 2}$  
for all $r\times r$ matrices $\bQ$. 
\end{lemma}

\begin{lemma}\label{prop-rank-1-approx} 
Let $\bM\in \RR^{d_1\times d_2}$ be a matrix with $\|\bM\|_{\rm F}=1$ and 
${\ba}$ and ${\bb}$ be unit vectors respectively in $\RR^{d_1}$ and $\RR^{d_2}$. 
Let $\widehat \ba$ be the top left singular vector of $\bM$. 
Then, 
\begin{equation}\label{prop-rank-1-approx-1}  
\big(\|\hat\ba\hat\ba^{\top} - \ba\ba^{\top}\|_{\rm 2}^2\big) \wedge (1/2)
\le \|\vect(\bM)\vect(\bM)^{\top} - \vect(\ba\bb^\top)\vect(\ba\bb^\top)^\top\|_{\rm 2}^2. 
\end{equation}
\end{lemma}

Lemmas \ref{lemma-transform-ext} and \ref{prop-rank-1-approx} are Propositions 5 and 3 in \cite{han2023tensor}, respectively.

\begin{lemma}\label{lemma:epsilonnet}
Let $d, d_j, d_*, r\le d\wedge d_j$ be positive integers, $\epsilon>0$ and
$N_{d,\epsilon} = \lfloor(1+2/\epsilon)^d\rfloor$. \\
(i) For any norm $\|\cdot\|$ in $\RR^d$, there exist
$M_j\in \RR^d$ with $\|M_j\|\le 1$, $j=1,\ldots,N_{d,\epsilon}$,
such that
\begin{align*}
\max_{\|M\|\le 1}\min_{1\le j\le N_{d,\epsilon}}\|M - M_j\|\le \epsilon .    
\end{align*}
Consequently, for any linear mapping $f$ and norm $\|\cdot\|_*$,
\begin{align*}
\sup_{M\in \RR^d,\|M\|\le 1}\|f(M)\|_* \le 2\max_{1\le j\le N_{d,1/2}}\|f(M_j)\|_*.    
\end{align*}
(ii) Given $\epsilon >0$, there exist $U_j\in \RR^{d\times r}$
and $V_{j'}\in \RR^{d'\times r}$ with $\|U_j\|_{2}\vee\|V_{j'}\|_{2}\le 1$ such that
\begin{align*}
\max_{M\in \RR^{d\times d'},\|M\|_{2}\le 1,\text{rank}(M)\le r}\
\min_{j\le N_{dr,\epsilon/2}, j'\le N_{d'r,\epsilon/2}}\|M - U_jV_{j'}^\top\|_{2}\le \epsilon.    
\end{align*}
Consequently, for any linear mapping $f$ and norm $\|\cdot\|_*$ in the range of $f$,
\begin{equation}\label{lm-3-2}
\sup_{M, \widetilde M\in \RR^{d\times d'}, \|M-\widetilde M\|_{2}\le \epsilon
\atop{\|M\|_{2}\vee\|\widetilde M\|_{2}\le 1\atop
\text{rank}(M)\vee\text{rank}(\widetilde M)\le r}}
\frac{\|f(M-\widetilde M)\|_*}{\epsilon 2^{I_{r<d\wedge d'}}}
\le \sup_{\|M\|_{2}\le 1\atop \text{rank}(M)\le r}\|f(M)\|_*
\le 2\max_{1\le j \le N_{dr,1/8}\atop 1\le j' \le N_{d'r,1/8}}\|f(U_jV_{j'}^\top)\|_*.
\end{equation}
(iii) Given $\epsilon >0$, there exist $U_{j,k}\in \RR^{d_k\times r_k}$
and $V_{j',k}\in \RR^{d'_k\times r_k}$ with $\|U_{j,k}\|_{2}\vee\|V_{j',k}\|_{2}\le 1$ such that
\begin{align*}
\max_{M_k\in \RR^{d_k\times d_k'},\|M_k\|_{2}\le 1\atop \text{rank}(M_k)\le r_k, \forall k\le K}\
\min_{j_k\le N_{d_kr_k,\epsilon/2} \atop j'_k\le N_{d_k'r_k,\epsilon/2}, \forall k\le K}
\Big\|\otimes_{k=2}^K M_k - \otimes_{k=2}^K(U_{j_k,k}V_{j_k',k}^\top)\Big\|_{2}\le \epsilon (K-1).    
\end{align*}
For any linear mapping $f$ and norm $\|\cdot\|_*$ in the range of $f$,
\begin{equation}\label{lm-3-3}
\sup_{M_k, \widetilde M_k\in \RR^{d_k\times d_k'},\|M_k-\widetilde M_k\|_{2}\le\epsilon\atop
{\text{rank}(M_k)\vee\text{rank}(\widetilde M_k)\le r_k \atop \|M_k\|_{2}\vee\|\widetilde M_k\|_{2}\le 1\ \forall k\le K}}
\frac{\|f(\otimes_{k=2}^KM_k-\otimes_{k=2}^K\widetilde M_k)\|_*}{\epsilon(2K-2)}
\le \sup_{M_k\in \RR^{d_k\times d_k'}\atop {\text{rank}(M_k)\le r_k \atop \|M_k\|_{2}\le 1, \forall k}}
\Big\|f\big(\otimes_{k=2}^K M_k\big)\Big\|_*
\end{equation}
and
\begin{equation}\label{lm-3-4}
\sup_{M_k\in \RR^{d_k\times d_k'},\|M_k\|_{2}\le 1\atop \text{rank}(M_k)\le r_k\ \forall k\le K}
\Big\|f\big(\otimes_{k=2}^K M_k\big)\Big\|_*
\le 2\max_{1\le j_k \le N_{d_kr_k,1/(8K-8)}\atop 1\le j_k' \le N_{d_k'r_k,1/(8K-8)}}
\Big\|f\big(\otimes_{k=2}^K U_{j_k,k}V_{j_k',k}^\top\big)\Big\|_*.
\end{equation}
\end{lemma}

\begin{lemma}\label{lm-pertubation}
Let $r \le d_1\wedge d_2$, $M$ be a $d_1\times d_2$ matrix, 
$U$ and $V$ be, respectively, the left and right singular matrices associated 
with the $r$ largest singular values of $M$,
$U_{\perp}$ and $V_{\perp}$ be the orthonormal complements of $U$ and $V$, 
and $\lam_r$ be the $r$-th largest singular value of $M$. 
Let $\widehat M = M + \Delta$ be a noisy version of $M$, 
$\{\widehat U, \widehat V, {\widehat U}_{\perp}, {\widehat V}_{\perp}\}$ 
be the counterpart of $\{U,V,V_\perp,V_\perp\}$, and 
${\widehat \lam}_{r+1}$ be the $(r+1)$-th largest singular value of $\widehat M$. 
Let $\|\cdot\|$ be a
matrix norm satisfying $\|ABC\|\le \|A\|_{2}\|C\|_{2}\|B\|$, 
$\epsilon_1= \|U^\top \Delta {\widehat V}_{\perp}\|$ and $\epsilon_2 = \|{\widehat U}_{\perp}^\top \Delta V\|$. Then,
\begin{align}\label{wedin+}
\| U_{\perp}^\top \widehat U \| 
\le \frac{{\widehat \lam}_{r+1}\epsilon_1+\lam_r\epsilon_2}{\lambda_r^2 - {\widehat \lam}_{r+1}^2}
\le \frac{\epsilon_1\vee\epsilon_2}{\lambda_r - {\widehat \lam}_{r+1}}.   
\end{align}
In particular, for the spectral norm $\|\cdot\|=\|\cdot\|_{2}$, $\hbox{\rm error}_1 =\|\Delta\|_{2}/\lambda_r$ 
and $\hbox{\rm error}_2 =\epsilon_2/\lambda_r$, 
\begin{align}\label{wedin-2}
\|\widehat U \widehat U^\top  -U U^\top\|_{2}\le \frac{\hbox{\rm error}_1^2+\hbox{\rm error}_2}{1-\hbox{\rm error}_1^2}.
\end{align}
\end{lemma}

Lemma \ref{lemma:epsilonnet} applies an $\epsilon$-net argument for matrices, as derived in Lemma G.1 in \cite{han2020iterative}. Lemma \ref{lm-pertubation} enhances the matrix perturbation bounds of \cite{wedin1972perturbation}, as derived in Lemma 4.1 in \cite{han2020iterative}. This sharper perturbation bound, detailed in the middle of \eqref{wedin+}, improves the commonly used version of the \cite{wedin1972perturbation} bound on the right-hand side,  
compared with Theorem 1 of \cite{cai2018rate} and Lemma 1 of \cite{chen2022rejoinder}. As \cite{cai2018rate} pointed out, such variations of the \cite{wedin1972perturbation} bound offer more precise convergence rates when when $\hbox{\rm error}_2\le \hbox{\rm error}_1$ in \eqref{wedin-2},  typically in the case of $d_1\ll d_2$.

The following lemma characterizes the accuracy of estimating the inverse of the covariance matrix $\Sigma_m^{-1}, m=1,\dots,M$ using sample covariance, and the convergence rate of the normalization constant. 
\begin{lemma} \label{lemma:precision matrix}
(i) Let $\bX_1, \cdots ,\bX_n \in \RR^{d_m \times d_{-m}}$ be i.i.d. random matrices, each following the matrix normal distribution $\bX \sim \cM\cN_{d_m \times d_{-m}}(\mu, \; \Sigma_m, \; \Sigma_{-m})$. To estimate $\Sigma_m$ and its inverse $\Sigma_m^{-1}$, we utilize the sample covariance $\hat\Sigma_m := (nd_{-m})^{-1} \sum\nolimits_{i=1}^n (\bX_i-\bar \bX) (\bX_i-\bar \bX)^\top$ with $\bar \bX=n^{-1}\sum_{i=1}^n \bX_i$, and its inverse $\hat\Sigma_m^{-1}$. Then there exists constants $C>0$ such that if $n d_{-m} \gtrsim d_m(1 \vee \|\Sigma_{-m} \|_2^2)$, in an event with probability at least $1-\exp(-cd_m)$, we have
\begin{align}
& \left\|\hat\Sigma_m - C_{m,\sigma} \Sigma_m\right\|_2 \leq C \cdot \left\| \Sigma_m \right\|_2 \left\| \Sigma_{-m}\right\|_2 \sqrt{\frac{d_m}{nd_{-m}}} ,  \label{eqn:covariance matrix} \\
& \left\| \hat\Sigma_m^{-1} - (C_{m,\sigma})^{-1} \Sigma_m^{-1} \right\|_2 \leq C (C_{m,\sigma}^{-2} \vee C_{m,\sigma}^{-4}) \left\|\Sigma_m^{-1} \right\|_2 \left\|\Sigma_{-m}\right\|_2 \sqrt{\frac{d_m}{nd_{-m}}}, \label{eqn:inverse matrix}
\end{align}
where $C_{m,\sigma}=(1-n^{-1})\tr(\Sigma_{-m})/d_{-m}$.

(ii) Let $\cX_1, \cdots ,\cX_n \in \RR^{d_1\times \cdots \times d_{M}}$ be i.i.d. random tensors, each following the tensor normal distribution $\cX \sim \cT\cN(\mu, \; \bSigma)$, where $\bSigma=[\Sigma_m]_{m=1}^M$ and $\Sigma_m\in\RR^{d_m \times d_m}$. Define $\hat{\Var}(\cX_{1\cdots1})=n^{-1}\sum_{i=1}^n (\cX_{i,1\cdots1}- \bar \cX_{1\cdots1})^2$ as the sample variance of the first element of $\cX$, and $\Var(\cX_{1\cdots 1})=\prod_{m=1}^M\Sigma_{m,11}$, $\bar \cX_{1\cdots1}=n^{-1}\sum_{i=1}^n\cX_{i,1\cdots1}$. Then in an event with probability at least $1-\exp(-c(t_1+t_2))$, we have
\begin{align}
& \left|\frac{\prod_{m=1}^M \hat\Sigma_{m,11}}{\hat{\Var}(\cX_{1\cdots1})} - C_{\sigma} \right| \leq C_{M} \Var(\cX_{1\cdots 1}) \left(\max_m \frac{\|\otimes_{k\neq m}\Sigma_{k}\|_2}{\sqrt{d_{-m}}} \cdot \sqrt{\frac{ t_1}{n }} +   \sqrt{\frac{ t_2}{n }}   \right),  \label{eqn:const} 
\end{align}
where $C_{\sigma}=\prod_{m=1}^M C_{m,\sigma} =(1-n^{-1})^M [\prod_{m=1}^M\tr(\Sigma_m)/d]^{M-1} = (1-n^{-1})^M [\tr(\bSigma)/d]^{M-1}$ and $C_{M}$ depends on $M$.
\end{lemma}

\begin{proof}
We first show \eqref{eqn:covariance matrix}.
Note that $\EE[(\bX-\mu)(\bX-\mu)^\top] = \tr(\Sigma_{-m})\cdot \Sigma_m = (n/(n-1))C_{m,\sigma}d_{-m} \cdot \Sigma_m$, and thus $\Sigma_m = \EE[(\bX-\mu)(\bX-\mu)^\top] (n-1) / (n C_{m,\sigma} d_{-m})$. 
Consider a sequence of independent copies $\bZ_1, \dots, \bZ_n$ of $\bZ \in \RR^{d_m \times d_{-m}}$ with entries $z_{ij}$ that are i.i.d. and follow $N(0, 1)$. The Gaussian random matrices $\bX_i$ are then given by $\bX_i - \mu = \bA\bZ_i\bB^\top$, where $\bA\bA^\top = \Sigma_m$ and $\bB\bB^\top = \Sigma_{-m}$. Then, $\bar \bX=n^{-1}\sum_{i=1}^n \bX_i = n^{-1}\sum_{i=1}^n \bA\bZ_i\bB^\top$. 
Note that
\begin{align*}
\hat\Sigma_m - C_{m,\sigma} \Sigma_m =& \left[ \frac{1}{nd_{-m}} \sum_{i=1}^n (\bX_i-\mu)(\bX_i-\mu)^\top - \frac{\tr(\Sigma_{-m})}{d_{-m}} \cdot \Sigma_m   \right]   \\
&\quad -  \left[ \frac{1}{d_{-m}} (\bar\bX-\mu)(\bar\bX-\mu)^\top  - \frac{\tr(\Sigma_{-m})}{nd_{-m}} \cdot \Sigma_m   \right].
\end{align*}
For any unit vector $v \in \RR^{d_m}$, 
$v^\top(\bX_i-\mu) = v^\top \bA \bZ_i \bB^\top =\vec1(v^\top \bA \bZ_i \bB^\top) = (\bB\otimes v^\top \bA) \vec1(\bZ_i)$.
It follows that
\begin{align*} 
&  \sum_{i=1}^n v^\top\left( (\bX_i-\mu)(\bX_i-\mu)^\top - \EE[(\bX-\mu)(\bX-\mu)^\top] \right) v \\
=&  \sum_{i=1}^n v^\top\left(\bA\bZ_i\bB^\top \bB \bZ_i^\top\bA^\top - \bA\EE(\bZ\bB^\top \bB\bZ^\top )\bA^\top \right) v \\
=&  \sum_{i=1}^n {\vec1}^\top(\bZ_i) (\bB^\top \otimes \bA^\top v) (\bB \otimes v^\top \bA) \vec1(\bZ_i) - \tr(\bB^\top \bB \otimes \bA^\top v v^\top \bA) \\
=&  \sum_{i=1}^n {\vec1}^\top(\bZ_i) (\bB^\top \otimes \bA^\top v) (\bB \otimes v^\top \bA) \vec1(\bZ_i) - \tr(\Sigma_{-m}) v^\top \Sigma_m v.
\end{align*}
By Hanson-Wright inequality, for any $t>0$,
\begin{align*}
&\PP\left( \sum_{i=1}^n \left[ {\vec1}^\top(\bZ_i) (\bB^\top \otimes \bA^\top v) (\bB \otimes v^\top \bA) \vec1(\bZ_i) - \tr(\Sigma_{-m}) v^\top \Sigma_m v     \right] \ge t \right) \\
\le& 2 \exp\left( - c \min\left\{ \frac{t^2}{16n \left\| \bB^\top \bB \otimes \bA^\top v v^\top \bA \right\|_{\rm F}^2} , \frac{t }{4 \left\| \bB^\top \bB \otimes \bA^\top v v^\top \bA \right\|_{2}}\right\} \right) \\
\le& 2 \exp\left( - c \min\left\{ \frac{t^2}{16n \left\| \Sigma_{-m} \right\|_{\rm F}^2 (v^\top \Sigma_m v)^2 } , \frac{t }{4 \left\| \Sigma_{-m} \right\|_{2}(v^\top \Sigma_m v) }\right\} \right) \\ 
\le&2 \exp\left( - c \min\left\{ \frac{t^2}{16nd_{-m} \left\| \Sigma_{-m} \right\|_{2}^2 (v^\top \Sigma_m v)^2 } , \frac{t }{4 \left\| \Sigma_{-m} \right\|_{2}(v^\top \Sigma_m v) }\right\} \right).
\end{align*}
Similarly, by Hanson-Wright inequality, for any $t>0$,
\begin{align*}
&\PP\left( v^\top\left( n^2(\bar\bX-\mu)(\bar\bX-\mu)^\top - n\EE[(\bX-\mu)(\bX-\mu)^\top] \right) v \ge t \right) \\
=&\PP\left( \left[\sum_{i,j=1}^n  {\vec1}^\top (\bZ_i) (\bB^\top \otimes \bA^\top v) (\bB \otimes v^\top \bA) \vec1(\bZ_j) - n\tr(\Sigma_{-m}) v^\top \Sigma_m v     \right] \ge t \right) \\
\le&2 \exp\left( - c \min\left\{ \frac{t^2}{16n^2 d_{-m} \left\| \Sigma_{-m} \right\|_{2}^2 (v^\top \Sigma_m v)^2 } , \frac{t }{4 n \left\| \Sigma_{-m} \right\|_{2}(v^\top \Sigma_m v) }\right\} \right).
\end{align*}
By $\eps-net$ argument, there exist unit vectors $v_1,...,v_{5^p}$ such that for all $p\times p$ symmetric matrix $M$,
\begin{align}\label{eq:epsilon_net}
\left\| M\right\|_2 \le  4 \max_{j\le 5^p} \left| v_j^\top M v_j\right|.   
\end{align} 
See also Lemma 3 in \cite{cai2010optimal}. Then
\begin{align*}
\PP\left( \left\| \hat\Sigma_m - C_{m,\sigma} \Sigma_m \right\|_2 \ge x) \right) &\le \PP\left( 4 \max_{j\le 5^{d_m}} \left| v_j^\top \left( \hat\Sigma_m - C_{m,\sigma} \Sigma_m \right) v_j\right| \ge x  \right)   \\
&\le 5^{d_m} \PP\left( 4 \left| v_j^\top \left( \hat\Sigma_m - C_{m,\sigma} \Sigma_m \right) v_j\right| \ge x  \right).
\end{align*}
As $d_m\lesssim n d_{-m}$, this implies that with $x\asymp \|\Sigma_m\|_2 \|\Sigma_{-m}\|_2 \sqrt{d_m/(n d_{-m})}$,
\begin{align*}
\PP\left( \left\| \hat\Sigma_m - C_{m,\sigma} \Sigma_m \right\|_2 \ge C \|\Sigma_m\|_2 \|\Sigma_{-m}\|_2 \sqrt{\frac{d_m}{n d_{-m}}} \right) \le    5^{d_m} \exp\left(-c_0 d_{m} \right) \le \exp\left(-c d_m \right).
\end{align*}

\noindent Second, we prove \eqref{eqn:inverse matrix}. For simplicity, denote $\Delta_m:= \hat\Sigma_m - C_{m,\sigma} \Sigma_m=\hat\Sigma_m - \tilde \Sigma_m$ with $\tilde \Sigma_m=C_{m,\sigma} \Sigma_m$. Then write
\begin{equation*}
\hat\Sigma_m^{-1} =  \tilde\Sigma_m^{-1/2}\left( \bI_{d_m} + \tilde\Sigma_m^{-1/2}\Delta_m\tilde\Sigma_m^{-1/2}\right)^{-1} \tilde\Sigma_m^{-1/2}   .
\end{equation*}
Using Neumann series expansion, we obtain
\begin{align*}
\hat\Sigma_m^{-1} &= \tilde\Sigma_m^{-1/2} \sum_{k=0}^\infty \left(-\tilde\Sigma_m^{-1/2} \Delta_m \tilde\Sigma_m^{-1/2} \right)^k \tilde\Sigma_m^{-1/2} \\
&= \tilde\Sigma_m^{-1} + \tilde\Sigma_m^{-1/2} \sum_{k=1}^\infty \left(-\tilde\Sigma_m^{-1/2} \Delta_m \tilde\Sigma_m^{-1/2} \right)^k \tilde\Sigma_m^{-1/2}
\end{align*}
Rearranging the term, we have
\begin{align*}
\hat\Sigma_m^{-1} - \tilde\Sigma_m^{-1} = -\tilde\Sigma_m^{-1} \Delta_m \tilde\Sigma_m^{-1} + \tilde\Sigma_m^{-1} \Delta_m \tilde\Sigma_m^{-1/2} \sum_{k=0}^\infty \left(-\tilde\Sigma_m^{-1/2} \Delta_m \tilde\Sigma_m^{-1/2} \right)^k \tilde\Sigma_m^{-1/2} \Delta_m \tilde\Sigma_m^{-1}
\end{align*}
Employing similar arguments in the proof of \eqref{eqn:covariance matrix}, we can show that in an event $\Omega$ with probability at least $1-\exp(-cd_m)$,
\begin{align*}
\left\| \tilde\Sigma_m^{-1/2} \Delta_m \tilde\Sigma_m^{-1/2}  \right\|_2 \le C C_{m,\sigma}^{-1} \| \Sigma_{-m}\|_2 \sqrt{\frac{d_m}{n d_{-m}} }    .
\end{align*}
As $d_m \|\Sigma_{-m}\|_2^2 \lesssim n d_{-m}$, in the same event $\Omega$, 
\begin{align*}
\left\|\tilde\Sigma_m^{-1} \Delta_m \tilde\Sigma_m^{-1}\right\|_2 &\leq C C_{m,\sigma}^{-2}\|\Sigma_m^{-1}\|_2 \|\Sigma_{-m}\|_2 \sqrt{\frac{d_m}{nd_{-m}}}, \\
\left\|\tilde\Sigma_m^{-1} \Delta_m \tilde\Sigma_m^{-1/2}\right\|_2 &\leq C C_{m,\sigma}^{-3/2} \norm{\Sigma_{-m}} \sqrt{\frac{d_m}{nd_{-m}}}   , \\
\left\|\tilde\Sigma_m^{-1/2} \Delta_m \tilde\Sigma_m^{-1/2}\right\|_2 &\leq C C_{m,\sigma}^{-1} \| \Sigma_{-m}\|_2 \sqrt{\frac{d_m}{n d_{-m}} } \le \frac12  C_{m,\sigma}^{-1} .
\end{align*}
Combining the above bounds together, we have, in the event $\Omega$,
\begin{align*}
\left\| \hat\Sigma_m^{-1} - \tilde\Sigma_m^{-1} \right\|_2 &\le C C_{m,\sigma}^{-2} \|\Sigma_m^{-1}\|_2 \| \Sigma_{-m}\|_2 \sqrt{\frac{d_m}{n d_{-m}} } + C C_{m,\sigma}^{-3} \left( \| \Sigma_{-m}\|_2 \sqrt{\frac{d_m}{n d_{-m}} } \right)^2 \cdot C_{m,\sigma}^{-1} \\
&\le C (C_{m,\sigma}^{-2} \vee C_{m,\sigma}^{-4})   \|\Sigma_m^{-1}\|_2 \| \Sigma_{-m}\|_2 \sqrt{\frac{d_m}{n d_{-m}} } .
\end{align*}

\noindent Next, we prove \eqref{eqn:const}. Employing similar arguments in the proof of \eqref{eqn:covariance matrix}, we can show
\begin{align*}
\PP\left( | \hat\Sigma_{m,11} - C_{m,\sigma} \Sigma_{m,11}| \ge C \Sigma_{m,11}\|\otimes_{k\neq m}\Sigma_{k}\|_2  \sqrt{\frac{ t_1}{n d_{-m}}} \right) \le\exp(-c_1 t_1).    
\end{align*}
Using tail probability bounds for $\chi_n^2$ (see e.g. Lemma D.2 in \cite{ma2013sparse}), we have
\begin{align*}
\PP\left( \left| \hat{\rm Var}(\cX_{1\cdots1}) - \prod_{m=1}^M\Sigma_{m,11} \right| \ge C \prod_{m=1}^M\Sigma_{m,11} \sqrt{\frac{ t_2}{n }} \right) \le\exp(-c_1 t_2).    
\end{align*}
It follows that in an event with probability at least $1-\exp(-c(t_1+t_2))$,
\begin{align*}
\left|\frac{\prod_{m=1}^M \hat\Sigma_{m,11}}{\hat{\Var}(\cX_{1\cdots1})} - C_{\sigma} \right| \le  C_{M} \prod_{m=1}^M\Sigma_{m,11} \left(\max_m \frac{\|\otimes_{k\neq m}\Sigma_{k}\|_2}{\sqrt{d_{-m}}} \cdot \sqrt{\frac{ t_1}{n }} +   \sqrt{\frac{ t_2}{n }}   \right)
\end{align*}
where $C_{M}$ depends on $M$. As $\Var(\cX_{1\cdots 1})=\prod_{m=1}^M\Sigma_{m,11}$, we finish the proof of \eqref{eqn:const}.

\end{proof}

The following lemma presents the tail bound for the spectral norm of a Gaussian random matrix.
\begin{lemma} \label{lemma:Gaussian matrix}
Let $\bE$ be an $p_1 \times p_2$ random matrix with $\bE\sim \cM\cN_{p_1 \times p_2}(0, \; \Sigma_1, \; \Sigma_{2}).$ Then for any $t>0$, with constant $C>0$, we have
\begin{equation} \label{eqn: Gaussian tail bound}
\PP\left(\|\bE\|_2 \ge C\|\Sigma_1\|_2^{1/2} \|\Sigma_{2}\|_2^{1/2} (\sqrt{p_1} + \sqrt{p_2} + t) \right) \leq \exp(-t^2) 
\end{equation}
Let $\OO_{p_1, r} = \{\bU \in \RR^{p_1\times r}, \; \bU^\top \bU = \bI_{r} \}$ be the set of all $p_1 \times r$ orthonormal columns and let $\bU_{\perp}$ be the orthogonal complement of $\bU$. Denote $\bE_{12} = \bU^\top\bE\bV_{\perp}, \; \bE_{21} = \bU_{\perp}^\top\bE\bV$, where $\bU \in \OO_{p_1\times r}, \; \bV \in \OO_{p_2\times r}.$ We have
\begin{align}
&\PP \left( \|\bE_{21}\|_2 \geq C\|\Sigma_1\|_2^{1/2} \|\Sigma_{2}\|_2^{1/2} (\sqrt{p_1} + t) \right) \leq \exp(-t^2)   \\
& \PP \left( \|\bE_{12}\|_2 \geq C\|\Sigma_1\|_2^{1/2} \|\Sigma_{2}\|_2^{1/2} (\sqrt{p_2} + t) \right) \leq \exp(-t^2)
\end{align}
\end{lemma}

\begin{proof}
Similar to \eqref{eq:epsilon_net}, using $\eps$-net argument for unit ball (see e.g. Lemma 5 in \cite{cai2018rate}), we have
\begin{align*}
\PP \left(\|\bE\|_2 \ge 3 u \right) \le 7^{p_1 +p_2} \cdot \max_{\|x \|_2 \le 1, \|y\|_2\le 1} \PP\left( |x^\top \bE y|\ge u \right).    
\end{align*}
Decompose $\bE=\Sigma_1^{1/2}\bZ\Sigma_{2}^{1/2}$, where $\bZ \in \RR^{p_1\times p_2}, \; \bZ_{ij} \stackrel {\text{i.i.d}}{\sim} N(0, 1).$ 
Then,
\begin{align*}
x^\top \bE y=  \left( y^\top \Sigma_{2}^{1/2} \otimes x^\top \Sigma_{1}^{1/2} \right) \vec1 (\bZ) \sim N\left(0,  y^\top \Sigma_{2}y \cdot x^\top\Sigma_1 x \right)    .
\end{align*}
By Chernoff bound of Gaussian random variables,
\begin{align*}
\PP\left( |x^\top \bE y|\ge u \right) \le 2 \exp \left( -\frac{ u^2}{y^\top \Sigma_{2}y \cdot x^\top\Sigma_1 x }  \right) \le 2   \exp\left( -\frac{ u^2}{\| \Sigma_1\|_2 \|\Sigma_{2}\|_2 }  \right)  .
\end{align*}
Setting $u\asymp \|\Sigma_1\|_2^{1/2} \|\Sigma_{2}\|_2^{1/2} (\sqrt{p_1} + \sqrt{p_2} + t)$, for certain $C>0$, we have
\begin{align*}
\PP\left(\|\bE\|_2 \ge C\|\Sigma_1\|_2^{1/2} \|\Sigma_{2}\|_2^{1/2} (\sqrt{p_1} + \sqrt{p_2} + t) \right) \le 2 \cdot 7^{p_1 +p_2} \exp(-c(p_1+p_2)-t^2)  \le \exp(-t^2) .    
\end{align*}

For $\bE_{21}$, following the same $\epsilon$-net arguments, we have
\begin{align*}
\PP \left(\| \bE_{21}\|_2 \ge 3 u \right) \le 7^{(p_1-r) +r} \cdot \max_{\|x \|_2 \le 1, \|y\|_2\le 1} \PP\left( |x^\top \bE_{21} y|\ge u \right).    
\end{align*}
Note that
\begin{align*}
x^\top \bE_{21} y&=  x^\top \bU_{\perp}^\top \Sigma_1^{1/2} \bZ \Sigma_{2}^{1/2} \bV y=  \left( y^\top \bV^\top \Sigma_{2}^{1/2} \otimes x^\top \bU_{\perp}^\top \Sigma_{1}^{1/2} \right) \vec1 (\bZ) \\
&\sim N\left(0, y^\top \bV^\top \Sigma_{2} \bV y \cdot x^\top \bU_{\perp}^\top \Sigma_1 \bU_{\perp} x \right)    .
\end{align*}
By Chernoff bound of Gaussian random variables, setting $u\asymp \|\Sigma_1\|_2^{1/2} \|\Sigma_{2}\|_2^{1/2} (\sqrt{p_1} + t)$, we have
\begin{align*}
\PP\left(\|\bE_{21}\| \ge C\|\Sigma_1\|_2^{1/2} \|\Sigma_{2}\|_2^{1/2} (\sqrt{p_1} + t) \right) \le 2 \cdot 7^{p_1} \exp(-c(p_1)-t^2)  \le \exp(-t^2) .    
\end{align*}
Similarly, we can derive the tail bound of $\|\bE_{12}\|_2$.
\end{proof}

The following lemma characterizes the maximum of norms for zero-mean Gaussian tensors after any projections.
\begin{lemma} \label{lemma:Guassian tensor projection}
Let $\cE \in \RR^{d_1 \times d_2 \times d_3}$ be a Gaussian tensor, $\cE \sim \cT\cN(0; \; \frac{1}{n}\Sigma_1^{-1}, \Sigma_2^{-1}, \Sigma_3^{-1})$, where there exists a constant $C_0>0$ such that $C_0^{-1} \le \mathop{\min}\limits_{m \in \{1,2,3 \}} \lam_{\min}(\Sigma_m) \le \mathop{\max}\limits_{m \in \{1,2,3 \}} \lam_{\max}(\Sigma_m) \le C_0$. Then we have the following bound for projections, with probability at most $C\exp(-Ct(d_2r_2 + d_3r_3) )$,
\begin{align}
\label{eqn:matrice version}
\mathop{\max}\limits_{\bV_2 \in \mathbb{R}^{d_2 \times r_2}, 
\bV_3 \in \mathbb{R}^{d_3 \times r_3} \atop \|\bV_2\|_2 \le 1, \|\bV_3\|_2 \le 1} \left\| \mat1(\cE \times_2 \bV_2^\top \times_3 \bV_3^\top)\right\|_2 \ge C\|\Sigma_1^{-1/2}\|_2 \|\Sigma_{-1}^{-1/2}\|_2 \frac{\sqrt{d_1 + r_2r_3} + \sqrt{1+t}(\sqrt{d_2r_2 + d_3r_3})}{\sqrt{n}} ,
\end{align}
for any t>0. Similar results also hold for ${\rm mat}_2(\cE \times_1 \bV_1^\top \times_3 \bV_3^\top)$ and ${\rm mat}_3(\cE \times_1 \bV_1^\top \times_2 \bV_2^\top)$.

Meanwhile, we have with probability at most $\exp(-Ct(d_1r_1 + d_2r_2 + d_3r_3) )$
\begin{align}
\label{eqn:tensor version}
\mathop{\max}\limits_{\bV_1, \bV_2, \bV_3\in \mathbb{R}^{d_m \times r_m} \atop
\|\bV_m\|_2 \le 1,\; m=1,2,3 } \left\| \cE \times_1 \bV_1^\top \times_2 \bV_2^\top \times_3 \bV_3^\top \right\|_{\rm F}^2 \ge C \|\Sigma_1^{-1}\|_2 \|\Sigma_2^{-1}\|_2 \|\Sigma_3^{-1}\|_2 \cdot \frac{r_1r_2r_3 + (1+t)(d_1r_1 + d_2r_2 + d_3r_3)}{n}, 
\end{align}
for any $t>0.$
\end{lemma}

\begin{proof} The key idea for the proof of this lemma is via $\epsilon$-net. We first prove \eqref{eqn:matrice version}. By Lemma \ref{lemma:epsilonnet}, for $m=1,2,3$, there exists $\epsilon$-nets: $\bV_m^{(1)}, \dots ,\bV_m^{(N_m)}$ for $\{ \bV_m \in \RR^{d_m \times r_m}: \|\bV_m\|_2 \le 1 \}$, $|\cN_m| \le ((4+\eps)/\eps)^{d_mr_m}$, such that for any $\bV_m \in \RR^{d_m \times r_m}$ satisfying $\|\bV_m\|_2 \le 1$, there exists $\bV_m^{(j)}$ such that
$\|\bV_m^{(j)} - \bV_m\|_2 \le \eps.$

For fixed $\bV_2^{(i)}$ and $\bV_3^{(j)}$, we define
\begin{align*}
\bZ_1^{(ij)} = \mat1 \left( \calE \times_2 (\bV_2^{(i)})^\top \times_3 (\bV_3^{(j)})^\top \right) \in \RR^{d_1 \times (r_2r_3)}.  
\end{align*}
It is easy to obtain that
$\bZ_1^{(ij)} \sim \cM\cN_{d_1 \times r_2r_3}\left(0; \; \frac{1}{n}\Sigma_1^{-1}, \; (\bV_2^{(i)} \otimes \bV_3^{(j)}) \cdot \Sigma_{-1}^{-1} \cdot (\bV_2^{(i)} \otimes \bV_3^{(j)})^\top \right).$ Then employing similar arguments of Lemma \ref{lemma:Gaussian matrix},
\begin{equation*}
\PP\left(\|\bZ_1^{(ij)}\|_2 \le C \|\Sigma_1^{-1/2}\|_2 \|\Sigma_{-1}^{-1/2}\|_2 \left(\frac{\sqrt{d_1} + \sqrt{r_2r_3} + t}{\sqrt{n}}\right)\right) \ge 1 - 2\exp(-t^2).   
\end{equation*}
Then we further have:
\begin{equation}
\label{eqn: lemma 4 proof}
\PP\left( \max\limits_{i,j} \|\bZ_1^{(ij)}\|_2 \le C \|\Sigma_1^{-1/2}\|_2 \|\Sigma_{-1}^{-1/2}\|_2 \left(\frac{\sqrt{d_1} + \sqrt{r_2r_3} + t}{\sqrt{n}}\right) \right) \ge 1 - 2((4+\eps)/\eps)^{d_2r_2 + d_3r_3} \exp(-t^2) 
\end{equation}
for all $t>0$. Denote
\begin{align*}
\bV_2^{*}, \bV_3^{*} &= \mathop{\rm argmax}\limits_{\bV_2 \in \RR^{d_2 \times r_2}, \bV_3 \in \RR^{d_3 \times r_3} \atop
\|\bV_2\|_2 \le 1, \|\bV_3\|_2 \le 1} \left\|\mat1 \left(\cE \times_2 \bV_2^\top \times_3 \bV_3^\top\right) \right\|_2 \\
M &= \mathop{\max}\limits_{\bV_2 \in \RR^{d_2 \times r_2}, \bV_3 \in \RR^{d_3 \times r_3} \atop
\|\bV_2\|_2 \le 1, \|\bV_3\|_2 \le 1} \left\| \mat1 \left(\calE \times_2 \bV_2^\top \times_3 \bV_3^\top\right) \right\|_2
\end{align*}
Using $\eps$-net arguments, we can find $1 \le i \le N_2$ and $1 \le j \le N_3$ such that $\|\bV_2^{(i)} - \bV_2^{*}\|_2 \le \eps$ and $\|\bV_3^{(i)} - \bV_3^{*}\|_2 \le \eps$. In this case, under \eqref{eqn: lemma 4 proof},
\begin{align*}
M =& \left\| \mat1 \left( \cE \times_2 (\bV_2^{*})^\top \times_3 ( \bV_3^{*})^\top\right) \right\|_2 \\
\le & \left\|\mat1 \left(\cE \times_2 (\bV_2^{(i)})^\top \times_3 (\bV_3^{(j)})^\top\right) \right\|_2
+ \left\|\mat1 \left(\cE \times_2 (\bV_2^{*} -\bV_2^{(i)})^\top \times_3 (\bV_3^{(j)})^\top\right) \right\|_2 \\
+& \left\|\mat1 \left(\cE \times_2 (\bV_2^{*})^\top \times_3 (\bV_3^{*} - \bV_3^{(j)})^\top\right) \right\|_2 \\
\le & C \left\|\Sigma_1^{-1/2} \right\|_2 \left\|\Sigma_{-1}^{-1/2}\right\|_2 \left(\frac{\sqrt{d_1} + \sqrt{r_2r_3} + t}{\sqrt{n}} \right) + \eps M + \eps M,
\end{align*}
Therefore, we have
\begin{equation*}
\PP \left(M \le C\cdot \frac{1}{1-2\eps}\left\|\Sigma_1^{-1/2}\right\|_2 \left\|\Sigma_{-1}^{-1/2}\right\|_2 \left(\frac{\sqrt{d_1} + \sqrt{r_2r_3} + t}{\sqrt{n}}\right) \right) \ge 1 - 2((4+\eps)/\eps)^{d_2r_2 + d_3r_3} \exp(-t^2)    
\end{equation*}
By setting $\eps=1/3$, and $t^2 = 2\log(13)(d_2r_2 + d_3r_3)(1+x)$, we have proved the first part of the lemma.

\noindent Then we prove the claim \eqref{eqn:tensor version}. Consider a Gaussian tensor $\cZ \in \RR^{d_1 \times d_2 \times d_3}$ with entries $z_{ijk} \stackrel{\text{i.i.d}}{\sim} N(0, \frac{1}{n}),$ then we have $\cE \times_1 \bV_1^\top \times_2 \bV_2^\top \times_3 \bV_3^\top := \cZ \times_1 ( \bV_1^\top \Sigma_1^{-1/2}) \times_2 (\bV_2^\top \Sigma_2^{-1/2}) \times_3 (\bV_3^\top \Sigma_1^{-1/2}).$ 
By Lemma 8 in \cite{zhang2018tensor}, we know
\begin{align*}
&\PP \Bigg( \left\|\cZ \times_1 (\Sigma_1^{-1/2} \bV_1)^\top \times_2 (\Sigma_2^{-1/2} \bV_2)^\top \times_3 (\Sigma_3^{-1/2} \bV_3)^\top\right\|_{\rm F}^2 - \frac{1}{n}\left\|\left(\Sigma_1^{-1/2} \bV_1\right) \otimes \left(\Sigma_2^{-1/2} \bV_2\right) \otimes \left(\Sigma_3^{-1/2} \bV_3\right) \right\|_{\rm F}^2 \\ 
&\quad \ge \frac{2}{n} \sqrt{t \left\|\left(\bV_1^\top \Sigma_1^{-1} \bV_1\right) \otimes \left(\bV_2^\top \Sigma_2^{-1} \bV_2\right) \otimes \left(\bV_3^\top \Sigma_3^{-1} \bV_3\right) \right\|_{\rm F}^2}  + \frac{2t}{n} \left\|\left(\Sigma_1^{-1/2} \bV_1\right) \otimes \left(\Sigma_2^{-1/2} \bV_2\right) \otimes \left(\Sigma_3^{-1/2} \bV_3\right) \right\|_2^2 \Bigg) \\
&\le \exp(-t).
\end{align*}

Note that for any given $\bV_k \in \RR^{d_k \times r_k}$ satisfying $\|\bV_k\|_2 \le 1, k=1, 2, 3,$ we have 
\begin{equation*}
\left\|\left(\Sigma_1^{-1/2} \bV_1\right) \otimes \left(\Sigma_2^{-1/2} \bV_2\right) \otimes \left(\Sigma_3^{-1/2} \bV_3 \right)\right\|_2 \le \|\Sigma_1^{-1/2}\|_2 \|\Sigma_2^{-1/2}\|_2 \|\Sigma_3^{-1/2}\|_2 := C_\lam^{1/2}.    
\end{equation*}
Then,
\begin{equation*}
\left\|\left(\Sigma_1^{-1/2} \bV_1 \right) \otimes \left(\Sigma_2^{-1/2} \bV_2\right) \otimes \left(\Sigma_3^{-1/2} \bV_3 \right)\right\|_{\rm F}^2 \le C_\lam r_1r_2r_3,    
\end{equation*}
and 
\begin{align*}
\left\|\left( \bV_1^\top \Sigma_1^{-1} \bV_1\right) \otimes \left(\bV_2^\top \Sigma_2^{-1} \bV_2 \right) \otimes \left( \bV_3^\top \Sigma_3^{-1} \bV_3 \right) \right\|_{\rm F}^2 \le C_\lam^2 r_1r_2r_3,
\end{align*}
we have for any fixed $\bV_1, \bV_2, \bV_3$ and $t > 0$ that
\begin{equation*}
\PP \left( \left\|\cZ \times_1 (\Sigma_1^{-1/2} \bV_1)^\top \times_2 (\Sigma_2^{-1/2} \bV_2)^\top \times_3 (\Sigma_3^{-1/2} \bV_3)^\top\right\|_{\rm F}^2 
\ge \frac{C_\lam r_1r_2r_3 + 2C_\lam \sqrt{r_1r_2r_3 t} + 2C_\lam t}{n} \right) \le \exp(-t).    
\end{equation*}
By geometric inequality, $2\sqrt{r_1r_2r_3t} \le r_1r_2r_3 + t$, then we further have
\begin{equation*}
\PP \left( \left\|\cE \times_1 \bV_1^\top \times_2 \bV_2^\top \times_3 \bV_3^\top\right\|_{\rm F}^2 
\ge \frac{C_\lam(2r_1r_2r_3 + 3t)}{n} \right) \le \exp(-t).   
\end{equation*}

The rest proof for this claim is similar to the first part. One can find three $\eps$-nets: $\bV_m^{(1)}, \dots ,\bV_m^{(N_m)}$ for $\{ \bV_m \in \RR^{d_m \times r_m}: \|\bV_m\|_2 \le 1 \}$, $N_m \le ((4+\eps)/\eps)^{d_mr_m}$, and we have the tail bound:
\begin{align} \label{eqn:lemma 4 proof 2}
        &\max\limits_{\bV_1^{(i)}, \bV_2^{(j)}, \bV_3^{(k)}} \PP \left( \left\|\cE \times_1 (\bV_1^{(i)})^\top \times_2 (\bV_2^{(j)})^\top \times_3 (\bV_3^{(k)})^\top\right\|_{\rm F}^2 \ge \frac{C_\lam (2r_1r_2r_3 + 3t)}{n} \right) \notag \\
        &\le \exp(-t) \cdot ((4+\eps)/\eps)^{d_1r_1 + d_2r_2 + d_3r_3} ,
\end{align}
for all $t > 0$. Assume
\begin{align*}
\bV_1^{*}, \bV_2^{*}, \bV_3^{*} &= \mathop{\rm argmax}\limits_{\bV_m \in \mathbb{R}^{d_m \times r_m} \atop
\|\bV_m\|_2 \le 1} \left\|\cE \times_1 \bV_1^\top \times_2 \bV_2^\top \times_3 \bV_3^\top\right\|_{\rm F}^2 \\
T &= \left\|\cE \times_1 (\bV_1^*)^\top \times_2 (\bV_2^*)^\top \times_3 (\bV_3^*)^\top\right\|_{\rm F}^2
\end{align*}
Then we can find $\bV_1^{(i)}, \bV_2^{(j)}, \bV_3^{(k)}$ in the corresponding $\eps$-nets such that $\|\bV_m^{*} - \bV_m\|_2 \le \eps$, $m=1,2,3$. And
\begin{align*}
T =& \left\|\cE \times_1 (\bV_1^*)^\top \times_2 (\bV_2^*)^\top \times_3 (\bV_3^*)^\top\right\|_{\rm F}^2 \\
\le & \left\|\cE \times_1 (\bV_1^{(i)})^\top \times_2 (\bV_2^{(j)})^\top \times_3 (\bV_3^{(k)})^\top\right\|_{\rm F}^2
+ \left\|\cE \times_1 (\bV_1^{(i)} - \bV_1^{*})^\top \times_2 (\bV_2^{*})^\top \times_3 (\bV_3^{*})^\top\right\|_{\rm F}^2 \\
+ & \left\|\cE \times_1 (\bV_1^{(i)})^\top \times_2 (\bV_2^{(j)} - \bV_2^*)^\top \times_3 (\bV_3^*)^\top\right\|_{\rm F}^2
+ \left\|\cE \times_1 (\bV_1^{(i)})^\top \times_2 (\bV_2^{(j)})^\top \times_3 (\bV_3^{(k)} - \bV_3^{*})^\top\right\|_{\rm F}^2 \\
\le & \frac{C_\lam (2r_1r_2r_3 + 3t)}{n} + 3\eps T
\end{align*}
which implies 
\begin{align*}
\left\|\cE \times_1 (\bV_1^*)^\top \times_2 (\bV_2^*)^\top \times_3 (\bV_3^*)^\top\right\|_{\rm F}^2 \le \frac{C_\lam \cdot (2r_1r_2r_3 + 3t)}{n(1-3\eps)}.    
\end{align*}
If we set $\eps = 1/9$ and $t=(1+x)\log(37)\cdot (d_1r_1 + d_2r_2 + d_3r_3)$, by \eqref{eqn:lemma 4 proof 2} we have proved \eqref{eqn:tensor version}. 
\end{proof}

\begin{lemma}\label{lemma:low-rank-tensor}
Suppose that $\cX_1,...,\cX_n \in \RR^{d_1\times\cdots\times d_M}$ are i.i.d. $\cT\cN(\mu,\bSigma)$ with $\bSigma=[\Sigma_m]_{m=1}^M$, and $\bar\cX$ is the sample mean. Then, with probability at least $1-e^{-c\sum_{m=1}^M d_m r}$,
\begin{align}
\sup_{\substack{\cV=\sum_{i=1}^r w_i \circ_{m=1}^M u_{im}\in \RR^{d_1\times\cdots\times d_M} \\ \sum_{i=1}^r w_i^2\le 1 , \| u_{im}\|_2=1, \forall i\le r, m\le M }}  \left| <\bar\cX-\mu,\cV> \right| \lesssim \sqrt{\frac{\sum_{m=1}^M d_m r }{n}}.
\end{align}
\end{lemma}

\begin{proof}
For any CP low rank tensor $\cV=\sum_{i=1}^r w_i \circ_{m=1}^M u_{im}$ with $w_i>0$ and $\| u_{im}\|_2=1$, we can reformulate a Tucker type decomposition, $\cV=\cS\times_{m=1}^M U_m$ with $\cS=\diag(w_1,...,w_r) \in \RR^{r\times\cdots\times r}$ and $U_m=(u_{1m},...,u_{rm})$. Let $N_0=\lfloor(1+2/\epsilon)^{r}\rfloor$ and $N_m=\lfloor(1+2/\epsilon)^{d_mr}\rfloor$. There exists $\epsilon$-nets $\cS_{j_0}^*  \in \RR^{r\times\cdots\times r}$ with diagonal elements $\cS_{j_0,k,...,k}^*\neq0$ and all the off-diagonal elements being 0, $\|\cS_{j_0}^*\|_{\rm F}\le 1, j_0=1,...,N_0$, and $U_{m,j_m}^* \in \RR^{d_m\times r}$ with $\|U_{m,j_m}^*\|_{2}\le 1,  j_m=1,...,N_m, 1\le m\le M$, such that
\begin{align*}
\max_{\|\cS\|_{\rm F}\le 1} \min_{1\le j_0\le N_0} \|\cS-\cS_{j_0}^*\|_{\rm F}\le \epsilon, \qquad \max_{\|U_m\|_{2}\le 1} \min_{1\le j_m\le N_m} \|U_m- U_{m,j_m}^*\|_{2}\le \epsilon, 1\le m\le M.   
\end{align*}
Note that $\sum_{i=1}^r w_i^2\le 1$ is equivalent to $\|\cS\|_{\rm F}\le 1$. Let $\cY=\bar X-\mu$.
Then by the ``subtraction argument",
\begin{align*}
&\sup_{\substack{\cV=\sum_{i=1}^r w_i \circ_{m=1}^M u_{im}\in \RR^{d_1\times\cdots\times d_M} \\ \sum_{i=1}^r w_i^2\le 1 , \| u_{im}\|_2=1, \forall i\le r, m\le M }} \left| <\cY,\cV> \right| -   \max_{\substack{\|\cS_{j_0}^*\|_{\rm F}\le 1, \|U_{m,j_m}^*\|_{2}\le 1 ,\\ j_0\le N_0,   j_m\le N_m, \forall 1\le m\le M}} \left| <\cY,\cS_{j_0}^*\times_{m=1}^M U_{m,j_m}^*> \right|   \\
=& \sup_{\substack{\cV=\sum_{i=1}^r w_i \circ_{m=1}^M u_{im}\in \RR^{d_1\times\cdots\times d_M} \\ \sum_{i=1}^r w_i^2\le 1 , \| u_{im}\|_2=1, \forall i\le r, m\le M }} \left| <\cY,\cV> \right| -   \max_{\substack{\|\cS_{j_0}^*\|_{\rm F}\le 1, j_0\le N_0}} \left| <\cY,\cS_{j_0}^*\times_{m=1}^M U_{m}> \right| \\
&+\sum_{k=0}^{M-1} \left( \max_{\substack{\|\cS_{j_0}^*\|_{\rm F}\le 1, \|U_{m,j_m}^*\|_{2}\le 1 ,\\ j_0\le N_0,   j_m\le N_m, \forall m\le k}} \left| <\cY,\cS_{j_0}^*\times_{m=1}^k U_{m,j_m}^*\times_{m=k+1}^M U_{m}> \right| \right.\\
&\qquad\qquad \left. - \max_{\substack{\|\cS_{j_0}^*\|_{\rm F}\le 1, \|U_{m,j_m}^*\|_{2}\le 1 ,\\ j_0\le N_0,   j_m\le N_m, \forall m\le k+1}} \left| <\cY,\cS_{j_0}^*\times_{m=1}^{k+1} U_{m,j_m}^*\times_{m=k+2}^M U_{m}> \right| \right)\\
\le& (M+1)\epsilon  \sup_{\substack{\cV=\sum_{i=1}^r w_i \circ_{m=1}^M u_{im}\in \RR^{d_1\times\cdots\times d_M} \\ \sum_{i=1}^r w_i^2\le 1 , \| u_{im}\|_2=1, \forall i\le r, m\le M }} \left| <\cY,\cV> \right|
\end{align*}
Setting $\epsilon=1/(2M+2)$,
\begin{align*}
\sup_{\substack{\cV=\sum_{i=1}^r w_i \circ_{m=1}^M u_{im}\in \RR^{d_1\times\cdots\times d_M} \\ \sum_{i=1}^r w_i^2\le 1 , \| u_{im}\|_2=1, \forall i\le r, m\le M }} \left| <\cY,\cV> \right| \le 2     \max_{\substack{\|\cS_{j_0}^*\|_{\rm F}\le 1, \|U_{m,j_m}^*\|_{2}\le 1 ,\\ j_0\le N_0,   j_m\le N_m, \forall 1\le m\le M}} \left| <\cY,\cS_{j_0}^*\times_{m=1}^M U_{m,j_m}^*> \right| .
\end{align*}
It follows that
\begin{align*}
&\PP\left(\sup_{\substack{\cV=\sum_{i=1}^r w_i \circ_{m=1}^M u_{im}\in \RR^{d_1\times\cdots\times d_M} \\ \sum_{i=1}^r w_i^2\le 1 , \| u_{im}\|_2=1, \forall i\le r, m\le M }} \left| <\cY,\cV> \right| \ge x\right) \\
\le& \PP\left(     \max_{\substack{\|\cS_{j_0}^*\|_{\rm F}\le 1, \|U_{m,j_m}^*\|_{2}\le 1 ,\\ j_0\le N_0,   j_m\le N_m, \forall 1\le m\le M}} \left| <\cY,\cS_{j_0}^*\times_{m=1}^M U_{m,j_m}^*> \right|  \ge x/2 \right)\\
\le& \prod_{k=0}^M N_k \cdot \PP\left(     \left| <\cY,\cS_{j_0}^*\times_{m=1}^M U_{m,j_m}^*> \right|  \ge x/2 \right) \\
\le& (4M+5) ^{\sum_{m=1}^M d_m r +r} \cdot \PP\left(     \left| <\cY,\cS_{j_0}^*\times_{m=1}^M U_{m,j_m}^*> \right|  \ge x/2 \right).
\end{align*}
Since $\cY=\bar X-\mu\sim \cT\cN(0,\frac1n \bSigma)$, we can show $| <\cY,\cS_{j_0}^*\times_{m=1}^M U_{m,j_m}^*> |$ is a $n^{-1/2}\prod_{m=1}^M \|\Sigma_m\|_{2}^{1/2}$ Lipschitz function, and $\E{| <\cY,\cS_{j_0}^*\times_{m=1}^M U_{m,j_m}^*> |}\le \sqrt{2/\pi} n^{-1/2}\prod_{m=1}^M \|\Sigma_m\|_{2}^{1/2}$. Then, by Gaussian concentration inequalities for Lipschitz functions,
\begin{align*}
\PP\left(     \left| <\cY,\cS_{j_0}^*\times_{m=1}^M U_{m,j_m}^*> \right|  \ge \E{| <\cY,\cS_{j_0}^*\times_{m=1}^M U_{m,j_m}^*> |} + t \right) \le \exp\left( -\frac{nt^2}{\prod_{m=1}^M \|\Sigma_m\|_{2}} \right).
\end{align*}
Setting $x\asymp t\asymp \sqrt{(\sum_{m=1}^M d_m r +r)/n}$, in an event with at least probability at least $1-e^{-c\sum_{m=1}^M d_m r}$,
\begin{align*}
\sup_{\substack{\cV=\sum_{i=1}^r w_i \circ_{m=1}^M u_{im}\in \RR^{d_1\times\cdots\times d_M} \\ \sum_{i=1}^r w_i^2\le 1 , \| u_{im}\|_2=1, \forall i\le r, m\le M }}  \left| <\bar\cX-\mu,\cV> \right| \lesssim \sqrt{\frac{\sum_{m=1}^M d_m r}{n}}.
\end{align*}
\end{proof}

The following lemma gives an inequality in terms of Frobenius norm between two tensors. 
\begin{lemma} 
\label{lemma:tensor norm inequality}
For two $M$-th order tensors $\gamma, \; \hat \gamma \in \RR^{d_1 \times \cdots \times d_M}$, if $\norm{\gamma - \hat\gamma}_{\rm F} = o(\|\gamma\|_{\rm F})$ as $n \rightarrow \infty$, and $\norm{\gamma}_{\rm F} \ge c$ for some constant $c>0$, then when $n \rightarrow \infty$,
\begin{equation*}
\norm{\gamma}_{\rm F} \cdot \norm{\hat\gamma}_{\rm F} - \langle \gamma ,\; \hat \gamma \rangle \asymp \norm{\gamma - \hat\gamma}_{\rm F}^2.    
\end{equation*}
\end{lemma}

\begin{proof}
Let $\calE = \hat \gamma - \gamma$, when $\norm{\gamma - \hat \gamma}_{\rm F} = o(\|\gamma\|_F)$ and $\norm{\gamma}_{\rm F} \ge c,$ we have
\begin{align*}
\norm{\gamma}_{\rm F} \cdot \norm{\hat \gamma}_{\rm F} - \langle \gamma, \; \hat \gamma \rangle &= \norm{\gamma}_{\rm F} \cdot \norm{\gamma + \calE}_{\rm F} - \langle \gamma, \; \gamma+\calE \rangle \\
&= \norm{\gamma}_{\rm F} \sqrt{\norm{\gamma}_{\rm F}^2 + 2\langle \gamma, \; \calE \rangle + \norm{\calE}_{\rm F}^2} - \norm{\gamma}_{\rm F}^2 - \langle \gamma, \; \calE \rangle \\
&= \norm{\gamma}_{\rm F}^2 \sqrt{1 + \frac{2\langle \gamma, \; \calE \rangle + \norm{\calE}_{\rm F}^2}{\norm{\gamma}_{\rm F}^2}} - \norm{\gamma}_{\rm F}^2 - \langle \gamma, \; \calE \rangle \\
&\asymp \norm{\gamma}_{\rm F}^2 \big( 1+\frac{1}{2} \frac{2\langle \gamma, \; \calE \rangle + \norm{\calE}_{\rm F}^2}{\norm{\gamma}_{\rm F}^2} \big) - \norm{\gamma}_{\rm F}^2 - \langle \gamma, \; \calE \rangle \\
&= \frac{\norm{\calE}_{\rm F}^2}{2} \asymp \norm{\hat \gamma - \gamma}_{\rm F}^2.
\end{align*}
\end{proof}

The following lemma illustrates the relationship between the risk function $\cR_{\btheta}(\hat\delta) -\cR_{\rm opt}(\btheta)$ and a more “standard” risk function $L_{\theta}(\hat \delta)$, which fulfills a role similar to that of the triangle inequality, as demonstrated in Lemma \ref{lemma:probability inequality}. Lemmas \ref{lemma:the first reduction} and \ref{lemma:probability inequality} correspond to Lemmas 3 and 4, respectively, in \cite{cai2019high}.

\begin{lemma} \label{lemma:the first reduction}
Let $\cZ \sim \frac{1}{2}\cT\cN(\cM_1; \; \bSigma) + \frac{1}{2}\cT\cN(\cM_2; \; \bSigma)$ with parameter $\theta = (\cM_1, \; \cM_2, \; \bSigma)$ where $\bSigma= [\Sigma_m]_{m=1}^M$. If a classifier $\hat \delta$ satisfies $L_{\theta}(\hat \delta) = o(1)$ as $n \rightarrow \infty$, then for sufficiently large n, 
\begin{align*}
\cR_{\btheta}(\hat\delta) -\cR_{\rm opt}(\btheta) \ge \frac{\sqrt{2\pi}\Delta}{8} e^{\Delta^2/8} \cdot L_{\theta}^2(\hat \delta) .    
\end{align*}
\end{lemma}

\begin{lemma} \label{lemma:probability inequality}
Let $\theta = (\cM, \; -\cM, \; [I_{d_m}]_{m=1}^M)$ and $\Tilde{\theta} = (\Tilde{\cM}, \; -\Tilde{\cM}, \; [I_{d_m}]_{m=1}^M)$ with $\norm{\cM}_{\rm F} = \norm{\Tilde{\cM}}_{\rm F} = \Delta/2$. For any classifier $\delta$, if $\norm{\cM - \Tilde{\cM}}_{\rm F} = o(1)$ as $n \rightarrow \infty$, then for sufficiently large n,
\begin{align*}
L_{\theta}(\delta) + L_{\Tilde{\theta}}(\delta) \ge \frac{1}{\Delta} e^{-\Delta^2/8} \cdot \norm{\cM - \Tilde{\cM}}_{\rm F} .    
\end{align*}
\end{lemma}


Although $L_{\theta}(\hat \delta)$ is not a distance function and does not satisfy an exact triangle inequality, the following lemma provides a variant of Fano's lemma. It suggests that it suffices to provide a lower bound for $L_{\theta}(\hat \delta)$, and $L_{\theta}(\hat \delta)$ satisfies an approximate triangle inequality (Lemma \ref{lemma:probability inequality}). 

\begin{lemma}[\cite{tsybakov2009}] \label{lemma:Tsybakov variant}
Let $N \ge 2$ and $\theta_0, \theta_1, \ldots ,\theta_N \in \Theta_d$. For some constants $\alpha_0 \in (0, 1/8), \gamma > 0$ and any classifier $\hat\delta$, if ${\rm KL}(\PP_{\theta_i}, \PP_{\theta_j}) \le \alpha_0 \log N/n$ for all $1 \le i \le N$, and $L_{\theta_i}(\hat\delta) < \gamma$ implies $L_{\theta_j}(\hat\delta) \ge \gamma$ for all $0 \le i \neq j \le N$, then 
\begin{align*}
\inf_{\hat\delta} \sup_{i \in [N]} \PP_{\theta_i}(L_{\theta_i}(\hat\delta)) \ge \gamma) \ge \frac{\sqrt{N}}{1+\sqrt{N}} (1-2\alpha_0-\sqrt{\frac{2\alpha_0}{\log N}}) >0  .    
\end{align*}
\end{lemma}

\begin{lemma}[Varshamov-Gilbert Bound, \cite{tsybakov2009}] \label{lemma:Varshamov-Gilbert Bound}
Consider the set of all binary sequences of length m: $\Omega = \big\{ \omega= (\omega_1,\ldots,\omega_m), \omega_i \in \{ 0,1\} \big\} = \{0,1\}^m$. Let $m \ge 8$, then there exists a subset $\{ \omega^{(0)}, \omega^{(1)},$ $\ldots, \omega^{(N)} \}$ of $\Omega$ such that $\omega^{(0)}=(0,\ldots, 0)$, $\rho_H(\omega^{(i)}, \omega^{(j)}) \ge m/8, \forall 0 \le i < j \le N$, and $N \ge 2^{m/8}$, where $\rho_H$ denotes the Hamming distance.
\end{lemma}


\section{Additional Simulation Results}
\label{append:simulation}
In this subsection, we present the simulation results of order-4 tensor with dimensions $d_1=d_2=d_3=d_4=20$, sample size $n_0=n_1=250$ and a low-rank $R=5$. Simulations were conducted using identity matrices $\Sigma_m, m \in [4]$ for $\bSigma$. In this scenario, we set $\calM_1 = 0$, and $\calM_2=\calB$. We vary signal strengths ${w_r}$ as follows:
\begin{itemize}
\item Equal strengths: $w_{\max}/w_{\min}=1$ with $w_{\max} \in \{2.5, 3.5, 5, 7.5\}$
\item Unequal strengths: $w_{r+1}/w_r=1.25$ with $w_{\max} \in {4, 5, 8}$
\end{itemize}
The estimation errors of the CP bases (i.e., $\max_{r,m} \| \hat{\ba}_{rm} \hat{\ba}_{rm}^\top - \ba_{rm} \ba_{rm}^\top \|_2$) are presented in Figure \ref{fig:3}, the estimation errors of the discriminant tensor (i.e., $\|\widehat \calB - \calB \|_{\rm F} / \| \calB \|_{\rm F}$) are shown in Table \ref{tab:7}, and the misclassification rate are listed in Table \ref{tab:8}.

The results for order-$4$ tensor predictors are primarily consistent with those for order-3 tensor predictors. In both cases, DISTIP-CP consistently outperforms DISTIP-Tucker and CATCH, showing lower estimation errors and misclassification rates across all settings, including orthogonal and non-orthogonal CP bases as well as equal and unequal signal strengths. DISTIP-Tucker generally performs better than CATCH, but both methods show significantly higher errors than DISTIP-CP.

When comparing the two tensor orders (see results for order-$3$ tensors in Section \ref{sec:simu}), the estimation and misclassification errors for DISTIP-CP remain low in both settings, but the gap between DISTIP-CP and the other methods widens with the order-$4$ tensors. This suggests that higher-order tensors introduce additional challenges, particularly due to the increased dimensionality and complexity, which DISTIP-CP handles more effectively. Meanwhile, DISTIP-Tucker and CATCH struggle more with non-orthogonal bases and unequal signal strengths in the order-$4$ setting, reflecting the added difficulty in capturing the discriminant structure.

\begin{figure}[H]
    \centering
    \begin{subfigure}[b]{0.22\textwidth}
        \centering
        \includegraphics[width=\textwidth]{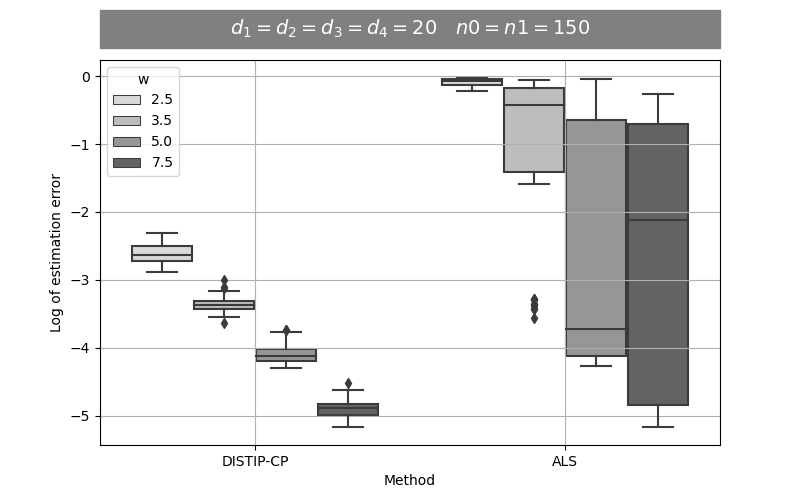}
        \caption{orthogonal, equal weights}
        \label{fig:5a}
    \end{subfigure}
    \hfill
    \begin{subfigure}[b]{0.22\textwidth}
        \centering
        \includegraphics[width=\textwidth]{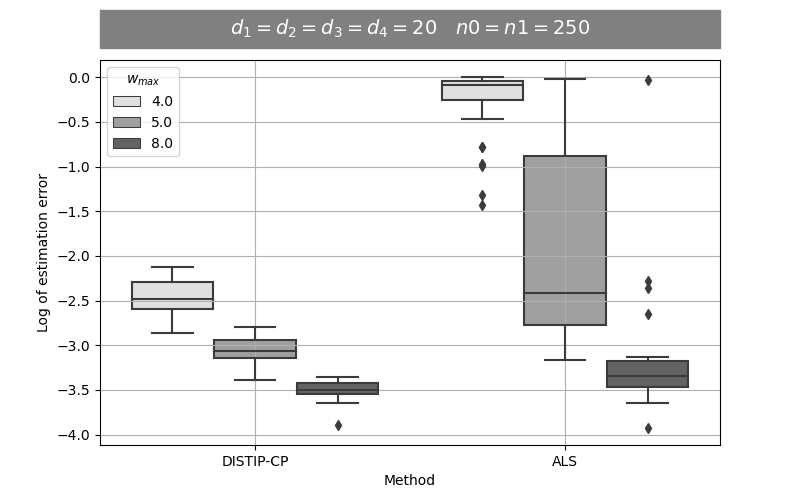}
        \caption{orthogonal, unequal weights}
        \label{fig:5b}
    \end{subfigure}
    \hfill
    \begin{subfigure}[b]{0.22\textwidth}
        \centering
        \includegraphics[width=\textwidth]{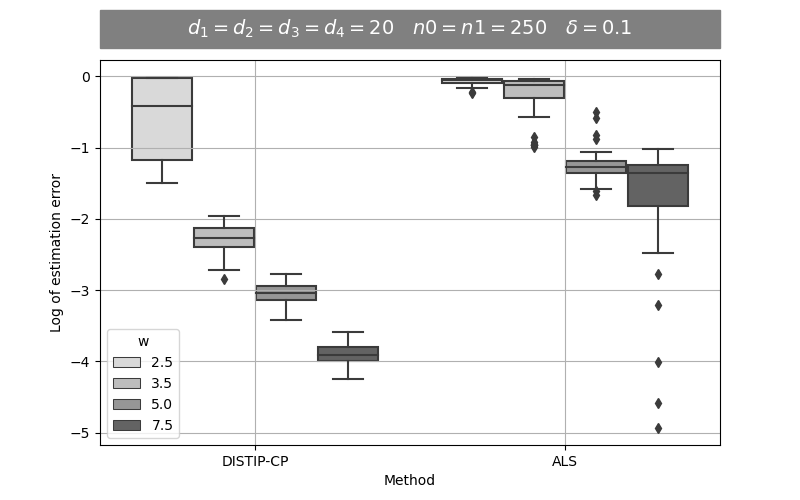}
        \caption{non-orthogonal, equal weights}
        \label{fig:6a}
    \end{subfigure}
    \hfill
    \begin{subfigure}[b]{0.22\textwidth}
        \centering
        \includegraphics[width=\textwidth]{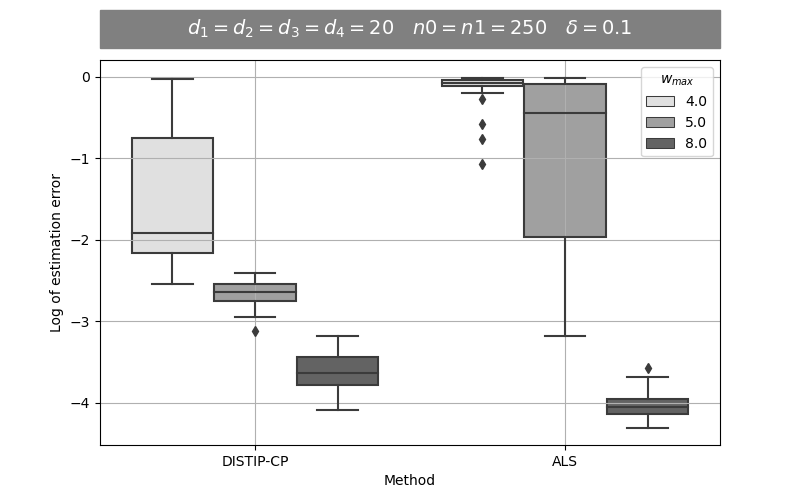}
        \caption{non-orthogonal, unequal weights}
        \label{fig:6b}
    \end{subfigure}
    \caption{Logarithmic estimation errors for CP basis using DISTIP-CP and ALS algorithms with identity covariance matrices.}
    \label{fig:3}
\end{figure}

\begin{table}[H]
\centering
\makebox[\linewidth]{%
    \hspace*{-1cm}\resizebox{\textwidth-1cm}{!}{%
\footnotesize
\begin{tabular}{>{\centering\arraybackslash}p{2.25cm}*{4}{>{\centering\arraybackslash}p{1.3cm}}*{3}{>{\centering\arraybackslash}p{1.5cm}}}
\toprule
\multirow{2}{*}{Algorithms} & \multicolumn{4}{c}{Equal Weights ($n_0 = n_1 = 250$)} & \multicolumn{3}{c}{Unequal Weights ($n_0 = n_1 = 250$)} \\
\cmidrule(lr){2-5} \cmidrule(lr){6-8}
& $w = 2.5$ & $w = 3.5$ & $w = 5.0$ & $w = 7.5$ & $w_{max} = 4$ & $w_{max} = 5$ & $w_{max} = 8$ \\
\midrule
Initial $\calB$ & 6.33\textsubscript{(0.54)} & 4.79\textsubscript{(0.37)} & 3.08\textsubscript{(0.24)} & 1.93\textsubscript{(0.16)} & 5.75\textsubscript{(0.29)} & 4.63\textsubscript{(0.28)} & 2.75\textsubscript{(0.24)} \\
{\footnotesize DISTIP-CP} & 0.32\textsubscript{(0.03)} & 0.24\textsubscript{(0.02)} & 0.14\textsubscript{(0.03)} & 0.13\textsubscript{(0.02)} & 0.62\textsubscript{(0.23)} & 0.34\textsubscript{(0.02)} & 0.13\textsubscript{(0.01)} \\
{\footnotesize DISTIP-Tucker} & 0.57\textsubscript{(0.04)} & 0.39\textsubscript{(0.03)} & 0.23\textsubscript{(0.03)} & 0.18\textsubscript{(0.02)} & 0.65\textsubscript{(0.03)} & 0.46\textsubscript{(0.03)} & 0.25\textsubscript{(0.02)} \\
{\footnotesize CATCH} & 1.03\textsubscript{(0.04)} & 1.02\textsubscript{(0.02)} & 1.01\textsubscript{(0.01)} & 0.99\textsubscript{(0.00)} & 1.12\textsubscript{(0.04)} & 1.07\textsubscript{(0.03)} & 0.99\textsubscript{(0.01)} \\
\toprule
\multirow{2}{*}{Algorithms} & \multicolumn{4}{c}{Equal Weights ($n_0 = n_1 = 250$)} & \multicolumn{3}{c}{Unequal Weights ($n_0 = n_1 = 250$)} \\
\cmidrule(lr){2-5} \cmidrule(lr){6-8}
& $w = 2.5$ & $w = 3.5$ & $w = 5.0$ & $w = 7.5$ & $w_{max} = 4$ & $w_{max} = 5$ & $w_{max} = 8$ \\
\midrule
{\small Initial $\calB$} & 6.36\textsubscript{(0.42)} & 4.54\textsubscript{(0.27)} & 3.15\textsubscript{(0.19)} & 2.13\textsubscript{(0.14)} & 5.10\textsubscript{(0.24)} & 4.02\textsubscript{(0.21)} & 2.53\textsubscript{(0.15)} \\
{\footnotesize DISTIP-CP} & 0.58\textsubscript{(0.11)} & 0.37\textsubscript{(0.04)} & 0.25\textsubscript{(0.03)} & 0.18\textsubscript{(0.02)} & 0.80\textsubscript{(0.06)} & 0.45\textsubscript{(0.03)} & 0.31\textsubscript{(0.02)} \\
{\footnotesize DISTIP-Tucker} & 0.64\textsubscript{(0.04)} & 0.40\textsubscript{(0.03)} & 0.26\textsubscript{(0.03)} & 0.19\textsubscript{(0.03)} & 
0.83\textsubscript{(0.17)} & 0.50\textsubscript{(0.05)} & 0.34\textsubscript{(0.03)} \\
{\footnotesize CATCH} & 1.16\textsubscript{(0.06)} & 1.10\textsubscript{(0.03)} & 1.00\textsubscript{(0.01)} & 0.96\textsubscript{(0.01)} & 1.12\textsubscript{(0.03)} & 1.06\textsubscript{(0.02)} & 0.98\textsubscript{(0.01)} \\
\bottomrule
\end{tabular}%
}%
}
\caption{\footnotesize Estimation errors for the discriminant tensor $\calB$ under identity covariance matrices. Results are shown for both orthogonal (upper section) and non-orthogonal (lower section) CP bases, comparing different estimation methods.}
\label{tab:7}
\end{table}

\begin{table}[H]
\centering
\makebox[\linewidth]{%
    \hspace*{-1cm}\resizebox{\textwidth-1cm}{!}{%
\footnotesize
\begin{tabular}{>{\centering\arraybackslash}p{2.25cm}*{4}{>{\centering\arraybackslash}p{1.3cm}}*{3}{>{\centering\arraybackslash}p{1.5cm}}}
\toprule
\multirow{2}{*}{Algorithms} & \multicolumn{4}{c}{Equal Weights ($n_0 = n_1 = 250$)} & \multicolumn{3}{c}{Unequal Weights ($n_0 = n_1 = 250$)} \\
\cmidrule(lr){2-5} \cmidrule(lr){6-8}
& $w = 2.5$ & $w = 3.5$ & $w = 5.0$ & $w = 7.5$ & $w_{max} = 4$ & $w_{max} = 5$ & $w_{max} = 8$ \\
\midrule
{\small Initial $\calB$} & 0.33\textsubscript{(0.02)} & 0.18\textsubscript{(0.01)} & 0.03\textsubscript{(0.01)} & 0.00\textsubscript{(0.00)} & 0.30\textsubscript{(0.01)} & 0.20\textsubscript{(0.01)} & 0.03\textsubscript{(0.00)} \\
{\footnotesize DISTIP-CP} & 0.004\textsubscript{(0.00)} & 0.00\textsubscript{(0.00)} & 0.00\textsubscript{(0.00)} & 0.00\textsubscript{(0.00)} & 0.058\textsubscript{(0.00)} & 0.00\textsubscript{(0.00)} & 0.00\textsubscript{(0.00)} \\
{\footnotesize DISTIP-Tucker} & 0.014\textsubscript{(0.00)} & 0.00\textsubscript{(0.00)} & 0.00\textsubscript{(0.00)} & 0.00\textsubscript{(0.00)} & 0.010\textsubscript{(0.00)} & 0.00\textsubscript{(0.00)} & 0.00\textsubscript{(0.00)} \\
{\footnotesize CATCH} & 0.44\textsubscript{(0.02)} & 0.37\textsubscript{(0.02)} & 0.22\textsubscript{(0.01)} & 0.10\textsubscript{(0.00)} & 0.41\textsubscript{(0.02)} & 0.36\textsubscript{(0.02)} & 0.24\textsubscript{(0.01)} \\
\toprule
\multirow{2}{*}{Algorithms} & \multicolumn{4}{c}{Equal Weights ($n_0 = n_1 = 250$)} & \multicolumn{3}{c}{Unequal Weights ($n_0 = n_1 = 250$)} \\
\cmidrule(lr){2-5} \cmidrule(lr){6-8}
& $w = 2.5$ & $w = 3.5$ & $w = 5.0$ & $w = 7.5$ & $w_{max} = 4$ & $w_{max} = 5$ & $w_{max} = 8$ \\
\midrule
{\small Initial $\calB$} & 0.33\textsubscript{(0.02)} & 0.19\textsubscript{(0.01)} & 0.04\textsubscript{(0.01)} & 0.00\textsubscript{(0.00)} & 0.27\textsubscript{(0.02)} & 0.17\textsubscript{(0.01)} & 0.01\textsubscript{(0.00)} \\
{\footnotesize DISTIP-CP} & 0.06\textsubscript{(0.04)} & 0.00\textsubscript{(0.00)} & 0.00\textsubscript{(0.00)} & 0.00\textsubscript{(0.00)} & 0.01\textsubscript{(0.00)} & 0.00\textsubscript{(0.00)} & 0.00\textsubscript{(0.00)} \\
{\footnotesize DISTIP-Tucker} & 0.06\textsubscript{(0.03)} & 0.00\textsubscript{(0.00)} & 0.00\textsubscript{(0.00)} & 0.00\textsubscript{(0.00)} & 0.00\textsubscript{(0.00)} & 0.00\textsubscript{(0.00)} & 0.00\textsubscript{(0.00)} \\
{\footnotesize CATCH} & 0.42\textsubscript{(0.02)} & 0.34\textsubscript{(0.02)} & 0.17\textsubscript{(0.03)} & 0.01\textsubscript{(0.01)} & 0.38\textsubscript{(0.02)} & 0.30\textsubscript{(0.02)} & 0.07\textsubscript{(0.02)} \\
\bottomrule
\end{tabular}%
}%
}
\caption{\footnotesize Misclassification rates for binary classification under identity covariance matrices. Upper section: orthogonal CP bases; lower section: non-orthogonal CP bases for $\calB$.}
\label{tab:8}
\end{table}

\end{appendices}

\end{document}